\documentclass[journal,twocolumn]{IEEEtran}

\IEEEoverridecommandlockouts

\usepackage[cmex10]{amsmath} 
\usepackage{amsfonts,stmaryrd,acronym,amssymb,amsthm,dsfont}
\usepackage{graphicx}
\usepackage{url}
\usepackage{cite}
\usepackage{hyperref}
\usepackage{enumitem}
\usepackage{enumerate}
\usepackage{lipsum}
\usepackage{diagbox}
\usepackage[dvipsnames]{xcolor}
\usepackage{mathrsfs}
\usepackage{bm}
\usepackage{dblfloatfix}

\graphicspath{{graphics/}}

\newtheorem{theorem}{Theorem}

\newtheorem{definition}{Definition}

\newtheorem{remark}{Remark}

\newtheorem{example}{Example}

\newtheorem{lemma}{Lemma}[]
\makeatletter
\def\blfootnote{\xdef\@thefnmark{}\@footnotetext}
\makeatother
\newcommand{\delf}[1]{\ensuremath{\mathds{1}}}
\newcommand{\indic}[1]{\ensuremath{\mathds{1}}}
\newcommand{\abs}[1]{\ensuremath{\left|#1\right|}}   

\newcommand{\intseq}[2]{\ensuremath{\llbracket{#1}{\,:\,}{#2}\rrbracket}}
\newcommand{\blr}[2]{\ensuremath{\llbracket{#1}{\,:\,}{#2}\rrbracket}}

\newcounter{mytempeqcounter}

\allowdisplaybreaks
\allowbreak

\makeatletter%
\if@twocolumn%
\newcommand{\Figwidth}{\columnwidth}%
\else
\newcommand{\Figwidth}{4.5in}%

\fi%
\makeatother%

\newcommand{\calA}{\mathcal{A}}
\newcommand{\bbA}{\mathbb{A}}

\newcommand{\calB}{\mathcal{B}}
\newcommand{\bbB}{\mathbb{B}}

\newcommand{\calC}{\mathcal{C}}

\newcommand{\calE}{\mathcal{E}}
\newcommand{\bbE}{\mathbb{E}}

\newcommand{\calI}{\mathcal{I}}

\newcommand{\calJ}{\mathcal{J}}

\newcommand{\calK}{\mathcal{K}}

\newcommand{\calL}{\mathcal{L}}

\newcommand{\calM}{\mathcal{M}}

\newcommand{\calN}{\mathcal{N}}
\newcommand{\bbN}{\mathbb{N}}

\newcommand{\bbP}{\mathbb{P}}

\newcommand{\calR}{\mathcal{R}}
\newcommand{\bbR}{\mathbb{R}}

\newcommand{\calS}{\mathcal{S}}

\newcommand{\calT}{\mathcal{T}}

\newcommand{\calU}{\mathcal{U}}

\newcommand{\calV}{\mathcal{V}}

\newcommand{\calX}{\mathcal{X}}

\newcommand{\calY}{\mathcal{Y}}

\newcommand{\MI}{I}
\newcommand{\ent}{H}

\newcommand{\FI}{\mathbb{J}}

\newcommand{\dent}{\mathds{h}}

\DeclareMathOperator*{\argmin}{argmin}
\DeclareMathOperator*{\argmax}{argmax}

\DeclareMathOperator*{\diag}{diag}
\def\on{{\otimes n}}
\DeclareMathOperator{\tr}{tr}
\acrodef{ACDIS}[ACDIS]{Adaptive Communication Decision and Information Systems}
\acrodef{AEP}{Asymptotic Equipartition Property}
\acrodef{AoA}{Angle of Arrival}
\acrodef{AWGN}{Additive White Gaussian Noise}
\acrodef{AVC}[AVC]{Arbitrarily Varying Channel}
\acrodef{BER}{Bit-Error-Rate}
\acrodef{BEC}{Binary Erasure Channel}
\acrodef{BPSK}{Binary Phase-Shift Keying}
\acrodef{BSC}{Binary Symmetric Channel}
\acrodef{BICM}[BICM]{Bit-Interleaved Coded-Modulation}
\acrodef{CDF}[CDF]{Cumulative Distribution Function}
\acrodef{CGF}[CGF]{Cumulant Generating Function}
\acrodef{CLT}[CLT]{Central Limit Theorem}
\acrodef{CSI}[CSI]{Channel State Information}
\acrodef{DMC}[DMC]{Discrete Memoryless Channel}
\acrodef{DMS}[DMS]{Discrete Memoryless Source}
\acrodef{ERM}[ERM]{Empirical Risk Minimization}
\acrodef{FER}[FER]{Frame Error Rate}
\acrodef{ICA}[ICA]{Independent Component Analysis}
\acrodef{iid}[i.i.d.]{independent and identically distributed}
\acrodef{IoT}[IoT]{Internet of Things}
\acrodef{KKT}[KKT]{Karush-Kuhn Tucker}
\acrodef{LASSO}[LASSO]{Least Absolute Shrinkage and Selection Operator}
\acrodef{LPD}[LPD]{Low Probability of Detection}
\acrodef{LDPC}[LDPC]{Low-Density Parity-Check}
\acrodef{LLMS}[LLMS]{Linear Least Mean Square}
\acrodef{LMS}[LMS]{Least Mean Square}
\acrodef{MAC}[MAC]{multiple-access channel}
\acrodef{MGF}[MGF]{Moment Generating Function}
\acrodef{MLC}[MLC]{Multi-Level Coding}
\acrodef{MLE}[MLE]{Maximum Likelihood Estimate}
\acrodef{MIMO}[MIMO]{Multiple-Input Multiple-Output}
\acrodef{MISO}{Multiple-Input Single-Output}
\acrodef{MSD}[MSD]{Multi-Stage Decoding}
\acrodef{MMSE}[MMSE]{Minimum Mean-Square Error}
\acrodef{PAC}[PAC]{Probably Approximately Correct}
\acrodef{PCA}[PCA]{Principal Component Analysis}
\acrodef{PDF}[PDF]{Probability Density Function}
\acrodef{PMF}[PMF]{Probability Mass Function}
\acrodef{PPM}[PPM]{Pulse Position Modulation}
\acrodef{PSD}{Power Spectral Density}
\acrodef{PSK}{Phase Shift Keying}
\acrodef{QKD}{Quantum Key Distribution}
\acrodef{ROC}{Receiver Operating Characteristic}
\acrodef{CVQKD}{Continuous-Variable \ac{QKD}}
\acrodef{QPSK}{Quadrature Phase-Shift Keying}
\acrodef{RV}{random variable}
\acrodef{SIMO}{Single-Input Multiple-Output}
\acrodef{SNR}{Signal-to-Noise Ratio}
\acrodef{SVM}[SVM]{Support Vector Machine}
\acrodef{TPCP}{Trace-Preserving Completely-Positive}
\acrodef{wrt}[w.r.t.]{with respect to}
\acrodef{WSS}{Wide Sense Stationary}
\acrodef{RHS}{Right Hand Side}
\acrodef{LHS}{Left Hand Side}

\begin{document}

\title{Secure Source Coding Resilient Against Compromised Users via an Access Structure}

\author{
\IEEEauthorblockN{Hassan ZivariFard and R\'{e}mi A. Chou}\\
\thanks{H.~ZivariFard is with the Department of Electrical Engineering, Columbia University, New York, NY 10027. R.~Chou is with the Department of Computer Science and Engineering, The University of Texas at Arlington, Arlington, TX 76019. This work is supported in part by NSF grant CCF-2425371. E-mails: hz2863@columbia.edu and remi.chou@uta.edu. Part of this work is presented at the 2022 IEEE Information Theory Workshop \cite{Secure_SC22}. Contact hz2863@columbia.edu and remi.chou@uta.edu for further questions about this work.}
}
\maketitle
\date{}

\begin{abstract}
Consider a source and multiple users who observe the \ac{iid} copies of correlated Gaussian random variables. The source wishes to compress its observations and store the result in a public database such that (i) authorized sets of users are able to reconstruct the source with a certain distortion level, and (ii) information leakage to non-authorized sets of colluding users is minimized. In other words, the recovery of the source is restricted to a predefined access structure. The main result of this paper is a closed-form characterization of the fundamental trade-off between the source coding rate and the information leakage rate. As an example, threshold access structures are studied, i.e., the case where any set of at least $t$ users is able to reconstruct the source with some predefined distortion level and the information leakage at any set of users with a size smaller than $t$ is~minimized.
\end{abstract}

\section{Introduction}
\label{sec:Intro}
A solution to the storage of private data that is resilient to compromised users is secure distributed storage via traditional cryptographic solutions such as secret sharing \cite{Secret_Sharing_1,Secret_Sharing_2}. Specifically, a solution based on secret sharing consists in encoding the private data and distributing parts of the encoded data among multiple users, via individual secure channels, such that any $t$ users that pool their information together can reconstruct the private data, while any $z(<t)$ colluding users cannot learn any information about the private data. The set of all sets of users capable of reconstructing the private data is referred to as the access structure. For instance,  the users could represent servers. 

\subsection{Problem overview}
In this paper, we aim to propose a secure distributed data storage strategy that solely relies on a public database and accounts for side information at the users by considering three main modifications of the secret sharing solution described above. First, we do not assume that  secure channels are available to transmit the encoded private data to the users, as secure channels come with a cost in practice, instead, we solely rely on the availability of a public database. Second, we consider that the users have side information about the private data. While this consideration is not relevant in the original secret sharing problem where the secret is an arbitrary sequence of symbols and does not represent information, it becomes relevant in a data storage context. Not accounting for the fact that the users can have side information raises the following two challenges that cannot be addressed with results for  traditional secret sharing: (i) it leads to overestimating the security guarantees of the protocol, and (ii) it leads to inefficiency in terms of data storage size. Third, in our proposed setting, we relax the lossless reconstruction constraint of traditional secret~sharing to a lossy reconstruction constraint~\cite{Berger71}. 

Two distinct bodies of work on secure data storage are related to our model. The first one is secret sharing, which specifically addresses the presence of access structures -- we refer to \cite{beimel2011secret} for a comprehensive literature review. The second one is secure source coding~\cite{Prabhakaran07,GunduzITW08,GunduzISIT08,VillardPianta,Tandon13,Chia13,Kittichokechai16,EkremUlukus13_Lossy}, which mainly addresses the presence of side information at the users, but in the absence of access structures. By contrast, in this paper, we propose to simultaneously address the presence of an access structure \emph{and} side information at the users within a single framework. 
Specifically, we consider a source and multiple users who observe the \ac{iid} copies of correlated Gaussian random variables. The source wants to compress its observations and store the result in a public database such that (i) only pre-defined sets of authorized users can reconstruct, up to a prescribed distortion level,  the source by pooling all their available information, and (ii) information leakage about the source to any other sets of colluding users is minimized. The main result of this paper is a closed-form characterization of the fundamental trade-off between source coding rate and information leakage rate. Our result indicates that if the source is more correlated, in a sense that we make precise in the sequel, with the side information of the authorized sets of users than with the side information of any unauthorized set of users, then the optimal information leakage rate grows linearly with the optimal source coding rate. On the other hand, if this is not the case, the optimal information leakage rate grows non-linearly with the optimal source coding rate. Additionally, for threshold access structures, i.e., when a fixed number of users, denoted by $t$, are needed to reconstruct the source (independently of the specific identities of those users), we show that the capacity region is, in general, not a monotonic function of the threshold~$t$.

\subsection{Novelties and main challenges}
Next, we discuss the novelties and main challenges of the main result of this paper, which is a characterization of the optimal rate-leakage region for the problem introduced in the previous section. We first describe the main challenges of our converse proof.
\begin{itemize}
    \item The side information of each authorized or unauthorized set of users is a vector Gaussian random variable, and each component of this vector accounts for the side information of one user of this set. In our study, we use sufficient statistics \cite[Section~2.9]{Cover_Book} to convert this vector Gaussian side information to a scalar random variable and facilitate the analysis of our setting. 
    For the converse proof, this conversion allows us to reduce the problem to two cases. A first case (respectively second case), in which the source is more (respectively less) correlated, in a sense that we make precise in the sequel, with the side information of any set of authorized users than with the side information of any unauthorized set of users.
    \item Another key step in the proof of our converse is the proof of the sufficiency of a single auxiliary random variable in the outer region that we derive, to achieve minimum information leakage at the unauthorized users for each of the two cases discussed above.
    \item A particularly challenging aspect of our setting is the compound structure of the problem, which arises as a consequence of having multiple authorized and multiple unauthorized sets of users. Specifically, in our achievability region, it leads to, first, an optimization over the distribution of the involved auxiliary random variables and, then, to an optimization over the sets of authorized users and unauthorized users, whereas the order of these two optimizations are reversed in our outer region. In general, such a mismatch between the inner and outer regions leads to a gap between the achievability and the converse, e.g., as in \cite{liang2009compound} for compound wiretap channels. In our setting, we obtain a capacity result by proving the existence of a saddle point, which proves that the order of the optimizations is irrelevant. 
    \end{itemize}
We now discuss the main challenges of our achievability proof.
    \begin{itemize}
    \item The achievability is first proved for discrete random variables and then extended to continuous random variables through quantization. Note that one cannot consider a specific quantization strategy at the unauthorized users to ensure the leakage requirement in an information-theoretic manner; therefore, a key step in this extension is to prove that the leakage constraint holds for continuous random variables.
    \item For the achievability proof, the use of sufficient statistics also facilitates the evaluation of the achievable rate region, in particular, the computation of the conditional covariance of vector Gaussian sources.
    \end{itemize}

\subsection{Related works}
Of particular relevance to this paper, \cite{Prabhakaran07} have established the first characterization of the rate at which an encoder may compress a source such that an authorized user can recover the source in a lossless manner while guaranteeing a minimum information leakage at an unauthorized user who observes the encoded source. Other variations of this problem are studied in \cite{GunduzITW08,GunduzISIT08,VillardPianta,Tandon13,Chia13,Kittichokechai16}. This problem is generalized to a scenario, in which the authorized user may recover the compressed source with some predefined distortion in \cite{VillardPianta}. Specifically, \cite{VillardPianta} characterized the optimal tradeoff between the rate, the desired distortion, and the information leakage when both the authorized and unauthorized users observe different \ac{iid} side information sequences that are correlated with the compressed source. The secure lossy compression of a vector Gaussian source when both the authorized user and the unauthorized user have vector Gaussian side information have been studied in \cite{EkremUlukus13_Lossy}, which derives inner and outer regions on the optimal trade-off between the rate, the desired distortion, and the information leakage. \cite{AllertonCuffSong,SchielerCuff} study this problem in the case where the fidelity of the communication to the authorized user is measured by a distortion metric and the secrecy performance of the system is also evaluated under a distortion metric, a line of study that was first initiated in \cite{Yamamoto83,Yamamoto88}. Secure source coding when there is a shared secret key between the legitimate terminals has also been studied in \cite{Yamamoto94,Yamamoto97,Merhav06,Merhav08,SatpathyCuff15}. Note that all these previous works do not consider access structures and deal with a source coding problem. The problem studied in this paper subsumes the secure lossy compression of a scalar Gaussian source when both the authorized user and the unauthorized user have scalar Gaussian side information as well as the secure lossy compression of a scalar Gaussian source when both the authorized user and the unauthorized user have vector Gaussian side information. 

In \cite[Section~V.A]{VillardPianta} and \cite[Example~1]{EkremUlukus13_Lossy}, the authors study a single-user and single-eavesdropper Gaussian secure source coding problem, which is a special case of the problem studied in our paper. Indeed, the problem studied in this paper involves multiple sets of authorized users and multiple sets of unauthorized users (eavesdroppers). Specifically, in our setting, we upper-bound the information leakage over all possible sets of unauthorized users, and for the reconstruction of the source we require that any set of authorized users can recover the source with some fixed distortion level. As discussed above, this creates additional challenges compared to a single authorized and single unauthorized user. Note also that the single-user single-eavesdropper case had not been fully solved, as \cite[Example~1]{EkremUlukus13_Lossy} establishes the capacity region when the compression rate is infinity, and \cite[Section~V.A]{VillardPianta} establishes the capacity when the side information at the eavesdropper is a degraded version of the legitimate receiver's side information. We note that the authors in \cite[Remark~8]{VillardPianta} conjecture that their achievability is optimal in the non-degraded case, however they do not provide a converse proof.

In the context of secret sharing, another related work is \cite{VidhiRemi22}, where a function of a Gaussian source must be reconstructed in a lossless manner by authorized sets of users and must be kept secret from unauthorized sets of users, who all own side information about the source. Finally, note that, in our results, the length of the compressed data stored in the public database and the source observation at the users scale linearly with the number of source observations $n$ and does not depend on the number of participants but only on the access structure. Specifically, the compressed data stored in the public database must allow the reconstruction of the source for the group of authorized participants that has the least amount of information about the source in their side information. This contrasts with traditional problems that involve access structures, e.g., secret-sharing model \cite{Secret_Sharing_2}, for which the best known coding schemes require the share size to scale exponentially with the number of participants for some access structures \cite{beimel2011secret}.

\subsection{Paper Organization}
The remainder of the paper is organized as follows. We define the notation in Section~\ref{sec:perliminari} and formally define the problem in Section~\ref{sec:Prob_Defi}. We present our main results in Section~\ref{sec:Main_Result} and provide the proofs in Section~\ref{sec:Converse}. We provide concluding remarks in Section~\ref{sec:Conclusion}.

\section{Notation}
\label{sec:perliminari}
Let $\bbN^+$ be the set of positive natural numbers, $\bbR$ be the set of real numbers, and define $\bbR_+\triangleq\{x\in\bbR|x\geq 0\}$ and $\bbR_{++}\triangleq\mathbb{R}_+\backslash\{0\}$. For any $a, b\in\bbR$, define $\intseq{a}{b}\triangleq [\lfloor a\rfloor, \lceil b\rceil]\cap\bbN^+$ and $[a]^+\triangleq\max\{0,a\}$. Random variables are denoted by capital letters and their realizations by lower case letters. Vectors are denoted by boldface letters, e.g., $\mathbf{X}$ denotes a random vector and $\mathbf{x}$ denotes a realization of $\mathbf{X}$. $\bbE_{X} (\cdot)$ is the expectation with respect to the random variable $X$, for brevity, we sometimes omit the subscripts in the expectation if it is clear from the context. 
The set of $\epsilon-$strongly jointly typical sequences of length $n$, according to $P_{XY}$, is denoted by $\calT_{\epsilon}^{(n)}({P_{XY}})$~\cite{ElGamalKim}. Superscripts denote the dimension of a vector, e.g., $X^n$. $X_i^j$ denotes $(X_i,X_{i+1},\dots,X_j)$, and $X_{\sim i}^{n}$ denotes the vector $X^n$ except $X_i$. The cardinality of a set is denoted by $|\cdot|$. The entropy of the discrete random variable $X$ is denoted by $\ent(X)$, the differential entropy of the random variable $X$ is denoted by $\dent(X)$, and the mutual information between the random variables $X$ and $Y$ is denoted by $\MI(X;Y)$. The support of a probability distribution $P$ is denoted by supp$(P)$. The $n$-fold product distribution constructed from the same distribution $P$ is denoted by $P^{\otimes n}$. Throughout the paper, $\log$ denotes the base $2$~logarithm.

\begin{figure*}
\centering
\includegraphics[width=6.2in]{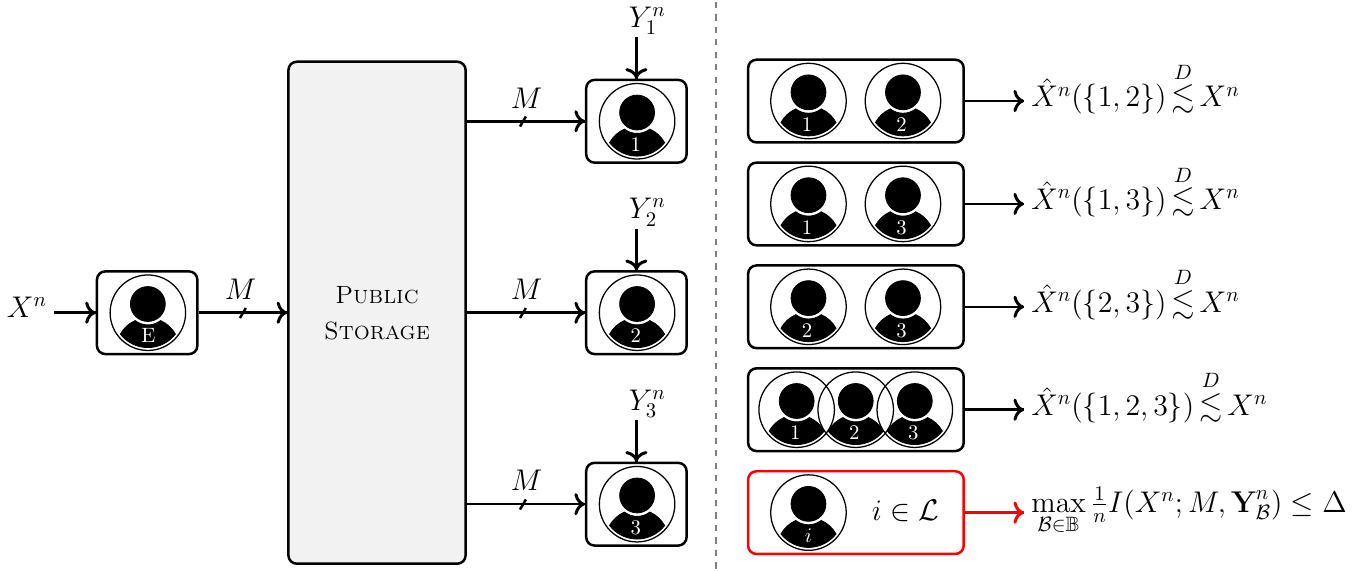}
\caption{Secure source coding with three users, i.e., $\calL=\{1,2,3\}$, when any single user must not learn more than $n\Delta$ bits of information about the source $X^n$, i.e., we set $\bbA = \{\{1,2\},\{1,3\},\{2,3\},\{1,2,3\}\}$, and
$\bbB = \{\{1\},\{2\},\{3\}\}$. $\hat{X}^n(\{i,j\})\mathop\lesssim \limits^{D}X^n$, for $i,j\in\{1,2,3\}$ and $i\ne j$, means that the distortion between the reconstructed source by the users $i$ and $j$ together and the source sequence $X^n$ must be less than $D$.}
\label{fig:System_Model}
\vspace{-0.5cm}
\end{figure*}
\section{Problem Statement}
\label{sec:Prob_Defi}
Consider a memoryless source $\big(\calX\times\boldsymbol{\calY}_\calL, P_{X\mathbf{Y}_\calL}\big)$, where $\calL\triangleq\intseq{1}{L}$ and $\mathbf{Y}_\calL\triangleq(Y_\ell)_{\ell\in\calL}$, that consists of $L+1$ alphabets $\calX\times\boldsymbol{\calY}_\calL$ and a joint distribution $P_{X\mathbf{Y}_\calL}$ over $\calX\times\boldsymbol{\calY}_\calL$. Let $\bbA$ be a set of subsets of $\calL$ such that for any $\calS\subseteq\calL$, if $\calS$ has a subset that belongs to $\bbA$, then $\calS\in\bbA$, i.e., $\bbA$ has a monotone access structure \cite{Monotone_Property}. 
Then, define $\bbB\triangleq2^\calL\backslash\bbA$ to be the set of all colluding subsets of users for which the information leakage about the source $X^n$ must be minimized (see~Fig.~\ref{fig:System_Model}). Henceforth, for any $\calA\in\bbA$ and for any $\calB\in\bbB$, $\mathbf{Y}_\calA$ and $\mathbf{Y}_\calB$ denote $(Y_\ell)_{\ell\in\calA}$ and $(Y_\ell)_{\ell\in\calB}$, respectively. 
Let $d:\calX\times\boldsymbol{\calY}_\calA\to\intseq{0}{d_{\max}}$ be a distortion measure such that $0\le d_{\max}<\infty$.
\begin{definition}\label{eq:Defi_Source_Code}
A $(2^{nR},n)$ source code for the memoryless source $\big(\calX\times\boldsymbol{\calY}_\calL, p_{X\mathbf{Y}_\calL}\big)$ consists of
\begin{itemize}
    \item an encoding function $f:x^n\mapsto m$, which assigns an index $m\in\intseq{1}{2^{nR}}$ to each $x^n\in\calX^n$. As depicted in Fig.~\ref{fig:System_Model}, $M$ is stored in a public database;
    \item decoding functions $\hat{x}_\calA:m\times \mathbf{y}_\calA^n\mapsto\hat{x}^n(\calA)\cup\{\mathfrak{e}\}$, where $\calA\in\bbA$, which assigns an estimate $\hat{x}^n(\calA)\in\calX^n$ or an error $\mathfrak{e}$ to each $m\in\intseq{1}{2^{nR}}$ and $\mathbf{y}_\calA^n\in\boldsymbol{\calY}_\calA^n$.
\end{itemize}
\end{definition}
\begin{definition}
\label{eq:Defi_Achi_Rate}
Let $D>0$. A pair $(R,\Delta)\in\bbR_+^2$ is achievable if there exists a sequence of $(2^{nR},n)$ source codes, such that,
\begin{subequations}\label{eq:Achi_Defi}
\begin{align}
    \max\limits_{\calA\in\bbA}\limsup\limits_{n\rightarrow\infty}\bbE\big[d\big(X^n,\hat{X}^n(\calA)\big)\big]&\le D,\label{eq:Error_Prob}\\
    \max\limits_{\calB\in\bbB}\lim\limits_{n\rightarrow\infty}\frac{1}{n}\MI(X^n;M,\mathbf{Y}_\calB^n)&\le\Delta,\label{eq:Equiv}
\end{align}
where the distortion between the sequences $x^n$ and $\hat{x}^n(\calA)$ is defined by
\begin{align}
    d\big(x^n,\hat{x}^n(\calA)\big)\triangleq\frac{1}{n}\sum\limits_{i=1}^nd\big(x_i,\hat{x}_i(\calA)\big).\label{eq:Disto_Pro_1}
\end{align}
\end{subequations}
The set of all achievable pairs is referred to as the rate-leakage~region and denoted by $\calR(D,\bbA)$.
\end{definition}

\eqref{eq:Error_Prob} means that any set of authorized users $\calA\in\bbA$ can reconstruct the source $X^n$ within the distortion $D$ from the observation $\mathbf{Y}_\calA^n$ and the public data $M$, and \eqref{eq:Equiv} means that any colluding set of unauthorized users $\calB\in\bbB$ cannot learn more than $n\Delta$ bits about the source $X^n$ from the observation $\mathbf{Y}_\calB^n$ and $M$. 
In this paper, we consider $P_{X\mathbf{Y}_{\calL}}$ the joint distribution of zero-mean jointly Gaussian random variables with a non-singular covariance matrix. We denote the variance of $X$ by $\sigma_X^2$. Without loss of generality, for every $\calA\in\bbA$ and $\calB\in\bbB$, by \cite[Theorem~3.5.2]{Gallager_Stochastic_Book}, one can write
\begin{subequations}\label{eq:Y_X_Relation_General}
    \begin{align}
    \mathbf{Y}_\calA=\mathbf{h}_\calA X+\mathbf{N}_\calA,\label{eq:Y_X_Relation_General_1}\\
    \mathbf{Y}_\calB=\mathbf{h}_\calB X+\mathbf{N}_\calB,\label{eq:Y_X_Relation_General_2}
\end{align}
\end{subequations}where $\mathbf{h}_\calA\in\bbR_{++}^{\abs{\calA}}$ and $\mathbf{h}_\calB\in\bbR_{++}^{\abs{\calB}}$ and $\mathbf{N}_\calA$ and $\mathbf{N}_\calB$ are zero-mean Gaussian random vectors with identity covariance matrix and independent of $X$. \eqref{eq:Y_X_Relation_General} is proved in Appendix~\ref{proof:System_Model}. Then, still without loss of generality, by normalizing \eqref{eq:Y_X_Relation_General}, one can consider the following source model
\begin{subequations}\label{eq:System_Model_2}
\begin{align}
\mathbf{Y}_\calA&=\mathbf{1}_\calA X+\mathbf{N}_\calA,\quad\forall\calA\in\bbA\label{eq:System_Model_21}\\
\mathbf{Y}_\calB&=\mathbf{1}_\calB X+\mathbf{N}_\calB,\,\quad\forall\calB\in\bbB\label{eq:System_Model_22}
\end{align}
\end{subequations}where $\mathbf{N}_\calA$ and $\mathbf{N}_\calB$ are zero-mean Gaussian random vectors with covariance matrices $\boldsymbol{\Sigma}_\calA\succ 0$ and $\boldsymbol{\Sigma}_\calB\succ 0$, respectively, that are independent of $X$ and $\mathbf{1}_\calA$ is the all-ones vector with size $\abs{\calA}$. Without loss of generality, we can also consider $\mathbf{N}_\calA$ and $\mathbf{N}_\calB$ independent, for $\calA\in\bbA$ and $\calB\in\bbB$, since \eqref{eq:Achi_Defi} only depends on the marginal distributions $(P_{X\mathbf{Y}_\calA})_{\calA\in\bbA}$ and $(P_{X\mathbf{Y}_\calB})_{\calB\in\bbB}$. 
In this paper, the distortion of the reconstructed sequence $\big(\hat{X}_i(\calA)\big)_{i=1}^n$ in Definition~\ref{eq:Defi_Achi_Rate} is measured by the mean square error as,
\begin{align}
    \frac{1}{n}\bbE\big[\big(X_i-\hat{X}_i(\calA)\big)^2\big]\le D.\label{eq:Distortion_Cons}
\end{align}

Since the minimizer of the mean square error is the \ac{MMSE} estimator, which is given by the conditional mean, we assume that the authorized users choose this optimal estimator, i.e., the authorized users in $\calA\in\bbA$ form $\big(\hat{X}_i(\calA)\big)_{i=1}^n$ as $\hat{X}_i(\calA)\triangleq\bbE\big[X_i|\mathbf{Y}_\calA^n,f(X^n)\big]$.

\section{Main Results}
\label{sec:Main_Result}
Henceforth, for some $\calA\in\bbA$, we assume $0\le D\le\sigma_{X|\mathbf{Y}_\calA}^2$, where $\sigma_{X|\mathbf{Y}_\calA}^2$ is the conditional variance of $X$ given $\mathbf{Y}_\calA$, $\sigma_{X|\mathbf{Y}_\calA}^2=\bbE\big[\big(X-\bbE[X|\mathbf{Y}_\calA]\big)^2\big|\mathbf{Y}_\calA\big]$. If $D\ge\sigma_{X|\calY_\calA}^2$ for all $\calA$ in $\bbA$, then $\calR(D,\bbA)=\{(R,\Delta):R\ge0,\Delta\ge\max_{\calB\in\bbB}\{I(X;Y_\calB)\}\}$, 
because the achievability scheme that consists in setting $M\triangleq\emptyset$ implies 
\begin{align}
    \frac{1}{n}\bbE\big[\big(X_i-\hat{X}_i(\calA)\big)^2\big]=\sigma_{X|\mathbf{Y}_\calA}^2.\nonumber
\end{align}
\subsection{Results for General Access Structures}
The main result of this paper is a closed-form expression for the optimal trade-off between the compression rate and the leakage rate of the source, which is provided in the following theorem.
\begin{theorem}\label{thm:Capacity}
Let $D>0$. For any access structure $\bbA$,
\begin{align}\nonumber
 &\calR(D,\bbA)=\nonumber\\
 &\left.\begin{cases}(R,\Delta):\\
R\ge\left[\frac{1}{2}\log\frac{\sigma_X^2}{D}-\frac{1}{2}\log\left(1+\frac{\sigma_X^2}{\tr\big(\boldsymbol{\Sigma}_{\calA^\star}^{-1}\big)^{-1}}\right)\right]^+\\
\Delta\ge \left\{ \begin{array}{l}
g_1\big(\calA^\star,\calB^\star\big)\qquad\text{when}\,\,\tr\big(\boldsymbol{\Sigma}_{\calA^\star}^{-1}\big) \ge\tr\big(\boldsymbol{\Sigma}_{\calB^\star}^{-1}\big)\\
g_2\big(\calA^\star,\calB^\star\big)\qquad\text{when}\,\,\tr\big(\boldsymbol{\Sigma}_{\calA^\star}^{-1}\big)\le\tr\big(\boldsymbol{\Sigma}_{\calB^\star}^{-1}\big)
\end{array} \right.
\end{cases}\hspace{-0.5cm}\right\},
\end{align}where
\begin{align}
    g_1\big(\calA^\star,\calB^\star\big)&\triangleq
\left[\frac{1}{2}\log\frac{\sigma_X^2}{D}-\frac{1}{2}\log\left(1+\frac{\sigma_X^2}{\tr\big(\boldsymbol{\Sigma}_{\calA^\star}^{-1}\big)^{-1}}\right)\right]^++\nonumber\\
&\quad\frac{1}{2}\log\left(1+\frac{\sigma_X^2}{\tr\big(\boldsymbol{\Sigma}_{\calB^\star}^{-1}\big)^{-1}}\right),\nonumber\\
g_2\big(\calA^\star,\calB^\star\big)&\triangleq\frac{1}{2}\log\left(\left[\frac{\sigma_X^2}{D}-\left(1+\frac{\sigma_X^2}{\tr\big(\boldsymbol{\Sigma}_{\calA^\star}^{-1}\big)^{-1}}\right)\right]^++\right.\nonumber\\
&\quad\left.1+\frac{\sigma_X^2}{\tr\big(\boldsymbol{\Sigma}_{\calB^\star}^{-1}\big)^{-1}}\right),\nonumber
\end{align}
$\calA^\star\in \argmin\limits_{\calA\in\bbA}\{\tr\left(\boldsymbol{\Sigma}_\calA^{-1}\right)\}$, and $\calB^\star\in \argmax\limits_{\calB\in\bbB}\{\tr\left(\boldsymbol{\Sigma}_\calB^{-1}\right)\}$.
\end{theorem}
The converse of Theorem~\ref{thm:Capacity} is provided in Section~\ref{sec:Converse}. The achievability proof of Theorem~\ref{thm:Capacity} is similar to that of \cite[Theorem~3]{VillardPianta} and is provided in Appendix~\ref{sec:Achie_Proof}. 
\begin{remark}[Comparison with {\cite[Example~1]{EkremUlukus13_Lossy} and \cite[Section~V.A]{VillardPianta}}]
\label{rem:Simple_Case}
When there is only one authorized and one unauthorized user, the problem setup in \eqref{eq:Achi_Defi} with $\bbA \triangleq \{1\}$ and $\bbB\triangleq \{2\}$ reduces to the problem setup in \cite[Example~1]{EkremUlukus13_Lossy} and \cite[Section~V.A]{VillardPianta} and Theorem~\ref{thm:Capacity} yields the capacity region
\begin{align}\nonumber
 &\calR(D,\bbA)=\left.\begin{cases}(R,\Delta):\\
R\ge\left[\frac{1}{2}\log\frac{\sigma_X^2}{D}-\frac{1}{2}\log\left(1+\frac{\sigma_X^2}{\sigma_{1}^2}\right)\right]^+\\
\Delta\ge \left\{ \begin{array}{l}
g_1\qquad\text{when}\,\,\,\,\sigma_{1}^2 \le\sigma_{2}^2\\
g_2
\qquad\text{when}\,\,\,\,\sigma_{1}^2\ge\sigma_{2}^2
\end{array} \right.
\end{cases}\right\},
\end{align}where
\begin{align}
    &g_1\triangleq\left[\frac{1}{2}\log\frac{\sigma_X^2}{D}-\frac{1}{2}\log\left(1+\frac{\sigma_X^2}{\sigma_{1}^2}\right)\right]^++\frac{1}{2}\log\left(1+\frac{\sigma_X^2}{\sigma_{2}^2}\right),\nonumber\\
&g_2\triangleq\frac{1}{2}\log\left(\left[\frac{\sigma_X^2}{D}-\left(1+\frac{\sigma_X^2}{\sigma_{1}^2}\right)\right]^++1+\frac{\sigma_X^2}{\sigma_{2}^2}\right).\nonumber
\end{align}
\end{remark}
In the lower bound of the compression rate $R$ in Theorem~\ref{thm:Capacity}, the term $\frac{1}{2}\log\frac{\sigma_X^2}{D}$ is the source coding capacity in the absence of side information \cite[Theorem~3.6]{ElGamalKim}, and the term $\frac{1}{2}\log\left(1+\frac{\sigma_X^2}{\tr\big(\boldsymbol{\Sigma}_{\calA^\star}^{-1}\big)^{-1}}\right)$ is the gain provided by the side information at the authorized users. In the lower bound on the information leakage $\Delta$ in Theorem~\ref{thm:Capacity}, the term $\left(1+\frac{\sigma_X^2}{\tr\big(\boldsymbol{\Sigma}_{\calB^\star}^{-1}\big)^{-1}}\right)$ represents a penalty coming from the side information at the unauthorized users. When $\sigma_X^2=2$, $D=0.1$, and $\tr\big(\boldsymbol{\Sigma}_{\calB^\star}^{-1}\big)=3.5$, the leakage rate $\Delta^\star$ is depicted in Fig.~\ref{fig:General_Example} with respect to the storage rate $R^\star$, where $(R^\star,\Delta^\star)$ represents the corner points of the region $\calR(D,\bbA)$ characterized in Theorem~\ref{thm:Capacity}. As seen in Fig.~\ref{fig:General_Example}, the leakage does not grow linearly with the storage rate $R^\star$, when $\text{tr}\big(\boldsymbol{\Sigma}_{\calA^\star}^{-1}\big)\le\text{tr}\big(\boldsymbol{\Sigma}_{\calB^\star}^{-1}\big)$. Intuitively, in this regime, the storage rate $R^*$ decreases as $\text{tr}\big(\boldsymbol{\Sigma}_{\calA^\star}^{-1}\big)$ grows but, since the unauthorized users in $\calB^\star$ have a ``less noisy" side information about the source than the authorized users in $\calA^\star$ have, the information leakage $\Delta^*$ does not increase with the storage rate $R^\star$ as fast as it does when the authorized sets of users $\calA^\star$ have a ``less noisy" side information about the source than the authorized set of users $\calB^\star$. 
In Fig.~\ref{fig:General_Example}, the corner point $C_1=\left(\frac{1}{2}\log\frac{\sigma_X^2}{D},\frac{1}{2}\log\frac{\sigma_X^2}{D}+\frac{1}{2}\log\left(1+\frac{\sigma_X^2}{\text{tr}\big(\boldsymbol{\Sigma}_{\calB^\star}^{-1}\big)^{-1}}\right)\right)$ corresponds to the case in which the side information at the authorized set of users is not correlated with the source, i.e., $\text{tr}\big(\boldsymbol{\Sigma}_{\calA^\star}^{-1}\big)\to0$, and therefore the communication rate is maximal. On the other hand, the corner point $C_2=\left(0,\frac{1}{2}\log\left(1+\frac{\sigma_X^2}{\text{tr}\big(\boldsymbol{\Sigma}_{\calB^\star}^{-1}\big)^{-1}}\right)\right)$ corresponds to the case in which the distortion between the side information at the authorized sets of users and the source is less than $D$, meaning that the encoder does not need to generate $M$. 
Note that, in this case, from \eqref{eq:SXgivenYA}, $D\ge\sigma_{X|\mathbf{Y}_{\calA^\star}}^2$ translates~to, $\text{tr}\big(\boldsymbol{\Sigma}_{\calA^\star}^{-1}\big)\ge D^{-1}-\sigma_X^{-2}$.
\begin{figure}
\centering
\includegraphics[width=\Figwidth]{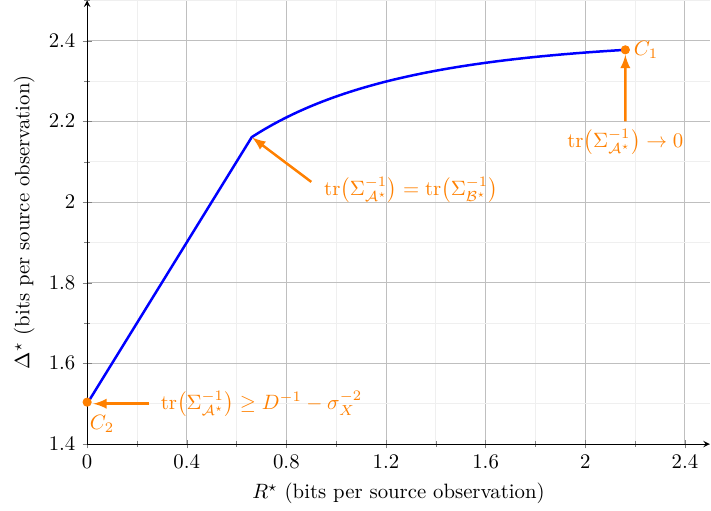}
\caption{$(R^\star,\Delta^\star)$ represents the corner points of the rate-leakage region $\calR(D,\bbA)$ characterized in Theorem~\ref{thm:Capacity}, for fixed noise variances, when $\sigma_X^2=2$, $D=0.1$, and ${\tr\big(\boldsymbol{\Sigma}_{\calB^\star}^{-1}\big)}=3.5$.}
\label{fig:General_Example}
\vspace{-0.5cm}
\end{figure}
\subsection{Results for Threshold Access Structures}
In this section, we consider a special type of access structure, which is known as the threshold access structure \cite{Secret_Sharing_2} and defined, for a threshold $t\in\intseq{1}{L}$, as
\begin{align}
    \bbA_t\triangleq\{\calA\subseteq\calL:\abs{\calA}\ge t\}.\nonumber%
\end{align}In other words, the threshold access structure is such that any set of $t$ users is able to reconstruct the compressed source with some predefined distortion. Similar to the general case, the complement of the set $\bbA_t$ is defined as $\bbB_t\triangleq2^\calL\backslash\bbA_t=\{\calB\subseteq\calL:\abs{\calB}< t\}$, and we consider $\calA_t^\star\in\argmin\limits_{\calA\in\bbA_t}\tr\left(\boldsymbol{\Sigma}_{\calA}^{-1}\right)$ and $\calB_t^\star\in\argmax\limits_{\calB\in\bbB_t}\{\tr\left(\boldsymbol{\Sigma}_{\calB}^{-1}\}\right)$. The following result presents necessary and sufficient conditions to determine whether the optimal trade-off between the compression rate and the leakage rate is decreasing or increasing with the threshold $t$.
\begin{theorem}
\label{thm:Threshold_Capacity}
Let $t\in\intseq{1}{L}$, suppose that $\tr\big(\boldsymbol{\Sigma}_{\calA_{t}^\star}^{-1}\big)< D^{-1}-\sigma_X^{-2}$, which means that the source $X^n$ needs to be encoded to satisfy \eqref{eq:Distortion_Cons} as discussed after Remark~\ref{rem:Simple_Case}, and define $R(D,\bbA_t)\triangleq\min\big\{R:(R,\Delta)\in\calR(D,\bbA_t)\big\}$. Then, we have,
\begin{itemize}
\begin{subequations}
\item $\calR(D,\bbA_L)\supseteq\calR(D,\bbA_t)\Leftrightarrow$\\$\frac{\sigma_X^{-2}+\tr\big(\boldsymbol{\Sigma}_{\calA_t^\star}^{-1}\big)}{\sigma_X^{-2}+\tr\big(\boldsymbol{\Sigma}_{\calB_t^\star}^{-1}\big)}\le\frac{\sigma_X^{-2}+\tr\big(\boldsymbol{\Sigma}_{\calA_L^\star}^{-1}\big)}{\sigma_X^{-2}+\tr\big(\boldsymbol{\Sigma}_{\calB_L^\star}^{-1}\big)}$;
\item $R(D,\bbA_t)\ge R(D,\bbA_{t+i})$, for $i\in\intseq{1}{L-t}$.
\end{subequations}
\end{itemize}
\end{theorem}

\begin{theorem}
\label{thm:Threshold_Capacity_2}
For any $t\in\intseq{1}{L}$, let  $\Delta(D,\bbA_t)\triangleq\min\big\{\Delta:(R,\Delta)\in\calR(D,\bbA_t)\big\}$, and suppose that $\tr\big(\boldsymbol{\Sigma}_{\calA_{t}^\star}^{-1}\big)< D^{-1}-\sigma_X^{-2}$, which means that the source $X^n$ needs to be encoded to satisfy \eqref{eq:Distortion_Cons} as discussed after Remark~\ref{rem:Simple_Case}. Then, for any $t\in\intseq{1}{L}$, and $i\in\intseq{1}{L-t}$,
\begin{itemize}[leftmargin=*]
    \item when $\tr\big(\boldsymbol{\Sigma}_{\calA_t^\star}^{-1}\big)\le\tr\big(\boldsymbol{\Sigma}_{\calB_t^\star}^{-1}\big)$ and $\tr\big(\boldsymbol{\Sigma}_{\calA_{t+i}^\star}^{-1}\big)\le\tr\big(\boldsymbol{\Sigma}_{\calB_{t+i}^\star}^{-1}\big)$,
    \begin{align}
        &\Delta(D,\bbA_t)\ge\Delta(D,\bbA_{t+i})\Leftrightarrow\nonumber\\
        &\tr\big(\boldsymbol{\Sigma}_{\calA_{t+i}^\star}^{-1}\big)-\tr\big(\boldsymbol{\Sigma}_{\calA_t^\star}^{-1}\big)\ge\tr\big(\boldsymbol{\Sigma}_{\calB_{t+i}^\star}^{-1}\big)-\tr\big(\boldsymbol{\Sigma}_{\calB_t^\star}^{-1}\big);\nonumber
    \end{align}
    \item when $\tr\big(\boldsymbol{\Sigma}_{\calA_t^\star}^{-1}\big)\ge\tr\big(\boldsymbol{\Sigma}_{\calB_t^\star}^{-1}\big)$ and $\tr\big(\boldsymbol{\Sigma}_{\calA_{t+i}^\star}^{-1}\big)\ge\tr\big(\boldsymbol{\Sigma}_{\calB_{t+i}^\star}^{-1}\big)$,
        \begin{align}
        &\Delta(D,\bbA_t)\ge\Delta(D,\bbA_{t+i})\Leftrightarrow\nonumber\\
        &\frac{\sigma_X^{-2}+\tr(\boldsymbol{\Sigma}_{\calB_t^\star}^{-1})}{\sigma_X^{-2}+\tr(\boldsymbol{\Sigma}_{\calA_t^\star}^{-1})}\ge\frac{\sigma_X^{-2}+\tr(\boldsymbol{\Sigma}_{\calB_{t+i}^\star}^{-1})}{\sigma_X^{-2}+\tr(\boldsymbol{\Sigma}_{\calA_{t+i}^\star}^{-1})};\nonumber
    \end{align}
    \item when $\tr\big(\boldsymbol{\Sigma}_{\calA_t^\star}^{-1}\big)\ge\tr\big(\boldsymbol{\Sigma}_{\calB_t^\star}^{-1}\big)$ and $\tr\big(\boldsymbol{\Sigma}_{\calA_{t+i}^\star}^{-1}\big)\le\tr\big(\boldsymbol{\Sigma}_{\calB_{t+i}^\star}^{-1}\big)$,
    \begin{align}
        \Delta(D,\bbA_t)\le\Delta(D,\bbA_{t+i});\nonumber
    \end{align}
    \item when $\tr\big(\boldsymbol{\Sigma}_{\calA_t^\star}^{-1}\big)\le\tr\big(\boldsymbol{\Sigma}_{\calB_t^\star}^{-1}\big)$ and $\tr\big(\boldsymbol{\Sigma}_{\calA_{t+i}^\star}^{-1}\big)\ge\tr\big(\boldsymbol{\Sigma}_{\calB_{t+i}^\star}^{-1}\big)$,
    \begin{align}
        \Delta(D,\bbA_t)\ge\Delta(D,\bbA_{t+i}).\nonumber
    \end{align}
\end{itemize}
\end{theorem}
The proofs of Theorem~\ref{thm:Threshold_Capacity} and Theorem~\ref{thm:Threshold_Capacity_2} are available in Appendix~\ref{proof:Threshold_Capacity}.
\begin{example}
\begin{figure}
\centering
\includegraphics[width=\Figwidth]{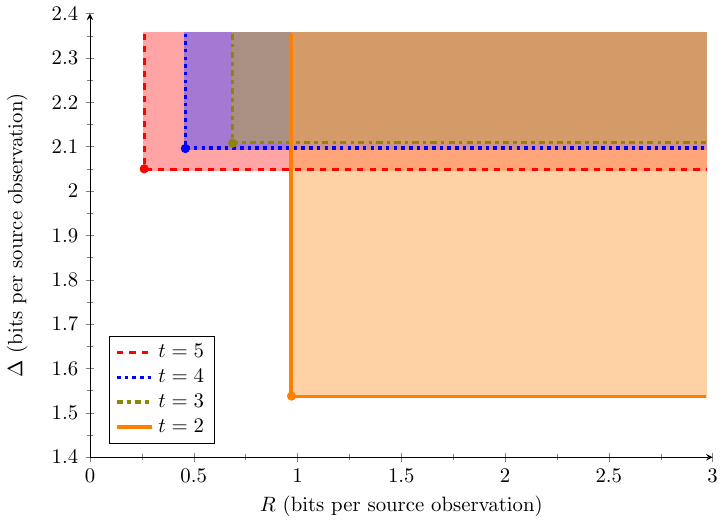}
\caption{The rate-leakage region for threshold access structures when $D=0.1$, $\sigma_X^2=2$, and $\boldsymbol{\Sigma}_\calL=\big[\begin{array}{ccccc}
    1 & 0.8 & 0.9 & 0.7 & 0.6
   \end{array}\big]^\intercal$.}
\label{fig:Example_1}
\vspace{-0.5cm}
\end{figure}
\sloppy Consider an encoder and five users. Let $D=0.1$, $\sigma_X^2=2$, and $\boldsymbol{\Sigma}_\calL=\big[\begin{array}{ccccc}
    1 & 0.8 & 0.9 & 0.7 & 0.6
   \end{array}\big]^\intercal$.

From the definitions of $\boldsymbol{\Sigma}_{\calA_t^\star}$ and $\boldsymbol{\Sigma}_{\calB_t^\star}$, we have $\boldsymbol{\Sigma}_{\calA_5^\star}=\diag(1,0.8,0.9,0.7,0.6)$ and $\boldsymbol{\Sigma}_{\calB_5^\star}=\diag(0.8,0.9,0.7,0.6)$. Hence,  $\tr\big(\boldsymbol{\Sigma}_{\calA_5^\star}^{-1}\big)=6.4563$ and $\tr\big(\boldsymbol{\Sigma}_{\calB_5^\star}^{-1}\big)=5.4563$. Plugging these in Theorem~\ref{thm:Capacity} results to $R\ge0.2618$ and $\Delta\ge 2.085$.

When $t=4$, $\boldsymbol{\Sigma}_{\calA_4^\star}=\diag(1,0.8,0.9,0.7)$ and $\boldsymbol{\Sigma}_{\calB_4^\star}=\diag(0.8,0.7,0.6)$. Hence,  $\tr\left(\boldsymbol{\Sigma}_{\calA_4^\star}^{-1}\right)=4.7897$ and $\tr\big(\boldsymbol{\Sigma}_{\calB_4^\star}^{-1}\big)=4.3452$, and by Theorem~\ref{thm:Capacity}, $R\ge 0.4594$ and $\Delta\ge 2.1282$.

When $t=3$, $\boldsymbol{\Sigma}_{\calA_3^\star}=\diag(1,0.9,0.8)$ and $\boldsymbol{\Sigma}_{\calB_3^\star}=\diag(0.7,0.6)$. Hence,  $\tr\big(\boldsymbol{\Sigma}_{\calA_3^\star}^{-1}\big)=3.3611$ and $\tr\big(\boldsymbol{\Sigma}_{\calB_3^\star}^{-1}\big)=3.0952$, and by Theorem~\ref{thm:Capacity}, $R\ge 0.6865$ and $\Delta\ge 2.1415$.

Finally, when $t=2$, $\boldsymbol{\Sigma}_{\calA_2^\star}=\diag(1,0.9)$ and $\boldsymbol{\Sigma}_{\calB_2^\star}=0.6$. Therefore,  $\tr\big(\boldsymbol{\Sigma}_{\calA_2^\star}^{-1}\big)=2.1111$ and $\tr\big(\boldsymbol{\Sigma}_{\calB_2^\star}^{-1}\big)=0.6$, and by Theorem~\ref{thm:Capacity}, $R\ge 0.9686$ and $\Delta\ge 2.1282$. 

This example verifies the relationships between the compression rate and the leakage rate for different thresholds provided in Theorem~\ref{thm:Threshold_Capacity} and Theorem~\ref{thm:Threshold_Capacity_2}. For instance, in Theorem~\ref{thm:Threshold_Capacity_2} for $t=3$ and $i=1$, we operate in the second case and the condition $\frac{\sigma_X^{-2}+\tr(\boldsymbol{\Sigma}_{\calB_t^\star}^{-1})}{\sigma_X^{-2}+\tr(\boldsymbol{\Sigma}_{\calA_t^\star}^{-1})}\ge\frac{\sigma_X^{-2}+\tr(\boldsymbol{\Sigma}_{\calB_{t+i}^\star}^{-1})}{\sigma_X^{-2}+\tr(\boldsymbol{\Sigma}_{\calA_{t+i}^\star}^{-1})}$ is satisfied so that $\Delta(D,\bbA_t)\ge\Delta(D,\bbA_{t+i})$, which is consistent with Fig.~\ref{fig:Example_1}.
\end{example}
\section{Converse Proof of Theorem~\ref{thm:Capacity}}
\label{sec:Converse}
In Section~\ref{subsec:General_Outer}, we provide a general outer region on the rate-leakage region $\calR(D,\bbA)$. In Sections~\ref{sec:Conversion_Scalar}, \ref{subsec:Better_Authorized}, and \ref{subsec:Better_Unauthorized}, we show that this outer region reduces to the region in Theorem~\ref{thm:Capacity}. Specifically, we convert the problem to a scalar problem by using sufficient statistics in Section~\ref{sec:Conversion_Scalar}. Then, in Section~\ref{subsec:Better_Authorized}, we study the case when the side information of any set of authorized users is more correlated, in a sense that we make precise in the sequel, with the source than the side information of any unauthorized set of users. Finally, in Section~\ref{subsec:Better_Unauthorized}, we study the case when the side information of the unauthorized sets of users is more correlated with the source than with the side information of the authorized sets of users.
\subsection{A General Outer Region}
\label{subsec:General_Outer}
We first provide a general outer region on the secure rate-distortion region of the problem defined in Section~\ref{sec:Prob_Defi}, which is based on \cite[Section~III-B]{VillardPianta}.
\begin{theorem}\label{thm:Upper_Bound}
For every $(\bbA,\bbB)$, the region $\calR(D,\bbA)$ is included in $\bigcap\limits_{(\calA,\calB)\in(\bbA,\bbB)}\calR_G(\mathbf{Y}_\calA,\mathbf{Y}_\calB)$, where
\begin{align}
  &\calR_G(\mathbf{Y}_\calA,\mathbf{Y}_\calB)\triangleq\bigcup\limits_{\substack{U-V-X-(\mathbf{Y}_\calA,\mathbf{Y}_\calB) \\ \bbE\left[\sigma_{X|\mathbf{Y}_\calA,V}^2\right]\le D}} \nonumber\\
  &\left.\begin{cases}(R,\Delta):\\
R>\MI(V;X|\mathbf{Y}_\calA)\\
\Delta>\MI(V;X)-\MI(V;\mathbf{Y}_\calA|U)+\MI(X;\mathbf{Y}_\calB|U)\\
\end{cases}\right\}.\nonumber
\end{align}
\end{theorem}
\begin{proof}
Consider the secure source coding problem in \cite[Section~III-B]{VillardPianta}, which consists of a memoryless source $\big((\calX,\calY_1,\calY_2),P_{XY_1Y_2}\big)$ with three outputs $X^n$, at the encoder, $Y_1^n$, at the legitimate terminal, and $Y_2^n$, at the eavesdropper. In this setting, the encoder wishes to encode its observed sequence in such a way that the legitimate receiver can reconstruct the source sequence $X^n$ with distortion $D$, and the information leakage about $X^n$ at the eavesdropper is minimized. An $(n,R)$-code for source coding is defined by an encoding function $f:\calX^n\to\intseq{1}{2^{nR}}$ and a decoding function $g:\intseq{1}{2^{nR}}\times\calY_1^n\to \calX^n$. In this problem, a pair $(R,\Delta)\in\bbR_+^2$ is achievable if there exists a sequence of $(n,R)$-codes such that, 
\setcounter{equation}{5}
\begin{subequations}\label{eq:General_Upper_Dist_Sec}
\begin{align}
    &\limsup\limits_{n\to\infty}\bbE[d(x^n,\hat{X}(f(X^n),Y_1^n))]\le D,\label{eq:General_Upper_Dist_Sec_1}\\
    &\lim\limits_{n\to\infty}\frac{1}{n}\MI(X^n;V^n,Y_2^n)\le\Delta,\label{eq:General_Upper_Dist_Sec_2}
\end{align}
\end{subequations}where $V^n\triangleq f(X^n)$. The outer region derived for this problem in \cite[Section~III-B]{VillardPianta} is $\calR_G(Y_1,Y_2)$. Now, consider the secure source coding problem defined in Section~\ref{sec:Prob_Defi} and the rate pair $(R,\Delta)\in\calR(D,\bbA)$ such that \eqref{eq:Error_Prob} and \eqref{eq:Equiv} are satisfied. In particular, \eqref{eq:General_Upper_Dist_Sec} is satisfied for any $\calA\in\bbA$ and $\calB\in\bbB$ when $Y_1^n$ is replaced by $\mathbf{Y}_\calA^n$ and $Y_2^n$ is replaced by $\mathbf{Y}_\calB^n$. It means that for any $\calA\in\bbA$ and $\calB\in\bbB$, $(R,\Delta)\in\calR_G(\mathbf{Y}_\calA,\mathbf{Y}_\calB)$ and Theorem~\ref{thm:Upper_Bound} holds.

The proof that the distortion constraint in $\calR_G(\mathbf{Y}_\calA,\mathbf{Y}_\calB)$ reduces to $\bbE\left[\sigma_{X|\mathbf{Y}_\calA,V}^2\right]\le D$ for Gaussian sources can be found in Appendix~\ref{proof:Upper_Bound}.
\end{proof}
\begin{remark}[Leakage Measure]
In \cite[Theorem~3]{VillardPianta}, the equivocation is used as a measure of leakage but, since we consider continuous sources in our setting, to avoid a negative equivocation, we replace the equivocation with mutual information leakage (see Definition~\ref{eq:Defi_Achi_Rate}).
\end{remark}

\subsection{Conversion to a Scalar Problem}
\label{sec:Conversion_Scalar}
To prove the converse part of Theorem~\ref{thm:Capacity}, we use the notion of sufficient statistics \cite[Section~2.9]{Cover_Book} and the following lemma from \cite{SIMO_Fading_WTC} to convert the problem in \eqref{eq:System_Model_2} to a problem in which the encoder, the authorized users, and the unauthorized users observe a scalar Gaussian~source.
\begin{lemma}{(\hspace{-0.15mm}\cite[Lemma~3.1]{SIMO_Fading_WTC})}\label{lemma:Suff_Stats}
Consider the channel with input $X$ and output $\mathbf{Y}$ given by
\begin{align}
    \mathbf{Y}=\mathbf{h}X+\mathbf{N},\nonumber%
\end{align}where $\mathbf{N}$ is a zero-mean Gaussian vector with covariance matrix $\boldsymbol{\Sigma}$ and $\mathbf{h}\in\bbR^n$. A sufficient statistic to correctly determine $X$ from $\mathbf{Y}$ is the following scalar
\begin{align}
    \tilde{Y}\triangleq\mathbf{h}^\intercal\boldsymbol{\Sigma}^{-1}\mathbf{Y}.\nonumber%
\end{align}
\end{lemma}
Fix $\calA\in\bbA$ and $\calB\in\bbB$. By Lemma~\ref{lemma:Suff_Stats}, sufficient statistics to correctly determine $X$ from $\mathbf{Y}_\calA$ and $\mathbf{Y}_\calB$ in \eqref{eq:System_Model_2} are the following scalars,
\begin{align}
\tilde{Y}_\calA&=\mathbf{1}_\calA^\intercal\boldsymbol{\Sigma}_\calA^{-1} \mathbf{Y}_\calA,\quad\tilde{Y}_\calB=\mathbf{1}_\calB^\intercal\boldsymbol{\Sigma}_\calB^{-1} \mathbf{Y}_\calB\label{eq:Suff_Stats_Inp_Outp}.
\end{align}
Hence, we have
\begin{subequations}\label{eq:Suff_Stats_Coeffies}
\begin{align}
&V^n-X^n-\mathbf{Y}_\calA^n-\tilde{Y}_\calA^n\label{eq:Suff_Stats_Markov_1},\\
&V^n-X^n-\tilde{Y}_\calA^n-\mathbf{Y}_\calA^n\label{eq:Suff_Stats_Markov_2},\\
&V^n-X^n-\mathbf{Y}_\calB^n-\tilde{Y}_\calB^n\label{eq:Suff_Stats_Markov_3},\\
&V^n-X^n-\tilde{Y}_\calB^n-\mathbf{Y}_\calB^n\label{eq:Suff_Stats_Markov_4},
\end{align}where
\begin{itemize}
    \item the Markov chain in \eqref{eq:Suff_Stats_Markov_1} follows since $V^n$ is a function of $X^n$ and $\tilde{Y}_\calA^n$ is a function of $\mathbf{Y}_\calA^n$;
    \item the Markov chain in \eqref{eq:Suff_Stats_Markov_2} follows since $V^n$ is a function of $X^n$ and from \cite[Section~2.9]{Cover_Book}  $X^n-\tilde{Y}_\calA^n-\mathbf{Y}_\calA^n$;
    \item the Markov chains in \eqref{eq:Suff_Stats_Markov_3} and \eqref{eq:Suff_Stats_Markov_4} are obtained similarly.
\end{itemize}
Next we rewrite $\eqref{eq:System_Model_2}$ as
\begin{align}
\tilde{Y}_\calA&=h_\calA X+\tilde{N}_\calA,\quad\tilde{Y}_\calB=h_\calB X+\tilde{N}_\calB,\label{eq:Suff_Stats_Inp_Outp_B}
\end{align}
where
\begin{align}
h_\calA&\triangleq\mathbf{1}_\calA^\intercal\boldsymbol{\Sigma}_\calA^{-1}\mathbf{1}_\calA=\tr\left(\boldsymbol{\Sigma}_\calA^{-1}\right),\label{eq:Suff_Stats_Coeffie_1}\\ h_\calB&\triangleq\mathbf{1}_\calB^\intercal\boldsymbol{\Sigma}_\calB^{-1}\mathbf{1}_\calB=\tr\left(\boldsymbol{\Sigma}_\calB^{-1}\right),\label{eq:Suff_Stats_Coeffie_2}\\
\tilde{N}_\calA&=\mathbf{1}_\calA^\intercal\boldsymbol{\Sigma}_\calA^{-1}\mathbf{N}_\calA,\quad\tilde{N}_\calB=\mathbf{1}_\calB^\intercal\boldsymbol{\Sigma}_\calB^{-1}\mathbf{N}_\calB,\label{eq:Suff_Stats_Coeffies_4}
\end{align}where \eqref{eq:Suff_Stats_Coeffie_2} follows since $\boldsymbol{\Sigma}_\calA$ and $\boldsymbol{\Sigma}_\calB$ are diagonal matrices.
\end{subequations}

We now show that the sufficient statistics in \eqref{eq:Suff_Stats_Inp_Outp} preserves the distortion constraint and the leakage constraint in Definition~\ref{eq:Defi_Achi_Rate}. 
By the distortion constraint in Theorem~\ref{thm:Upper_Bound}, we have 
\begin{align}
&D\ge\frac{1}{n}\sum\limits_{i=1}^n\bbE\Big[\Big(X_i-\bbE[X_i|V^n,\mathbf{Y}_\calA^n]\Big)^2\Big]. \label{eq:Disto_Const_3}
\end{align}Expanding the \ac{RHS} of \eqref{eq:Disto_Const_3} results to,
\begin{align}
&\mathbb{E}\Big[\Big(X_i-
\bbE[X_i|V^n,\mathbf{Y}_\calA^n]\Big)^2\Big]\nonumber\\
&=\int\limits_{x_i}\int\limits_{v^n}\int\limits_{\mathbf{y}_\calA^n}\Big(x_i-\bbE[X_i|v^n,\mathbf{y}_\calA^n]\Big)^2P(x_i,v^n,\mathbf{y}_\calA^n)dx_idv^nd\mathbf{y}_\calA^n\nonumber\\
&=\int\limits_{x_i}\int\limits_{v^n}\int\limits_{\mathbf{y}_\calA^n}\int\limits_{\tilde{y}_\calA^n}\Big(x_i-\bbE[X_i|v^n,\mathbf{y}_\calA^n]\Big)^2\times\nonumber\\ &\quad P(x_i,v^n,\mathbf{y}_\calA^n,\tilde{y}_\calA^n)dx_idv^nd\mathbf{y}_\calA^nd\tilde{y}_\calA^n\nonumber\\
&\mathop = \limits^{(a)}\int\limits_{x_i}\int\limits_{v^n}\int\limits_{\mathbf{y}_\calA^n}\int\limits_{\tilde{y}_\calA^n}\Big(x_i-\int\limits_{\tilde{x}_i}\tilde{x}_iP(\tilde{x}_i|v^n,\mathbf{y}_\calA^n)d\tilde{x}_i\Big)^2\times\nonumber\\ 
&\quad P(x_i,v^n,\mathbf{y}_\calA^n,\tilde{y}_\calA^n)dx_idv^nd\mathbf{y}_\calA^nd\tilde{y}_\calA^n\nonumber\\
&\mathop = \limits^{(b)}\int\limits_{x_i}\int\limits_{v^n}\int\limits_{\mathbf{y}_\calA^n}\int\limits_{\tilde{y}_\calA^n}\Big(x_i-\int\limits_{\tilde{x}_i}\tilde{x}_iP(\tilde{x}_i|v^n,\tilde{y}_\calA^n)d\tilde{x}_i\Big)^2\times\nonumber\\ &\quad P(x_i,v^n,\mathbf{y}_\calA^n,\tilde{y}_\calA^n)dx_idv^nd\mathbf{y}_\calA^nd\tilde{y}_\calA^n\nonumber\\
&=\int\limits_{x_i}\int\limits_{v^n}\int\limits_{\tilde{y}_\calA^n}\Big(x_i-\int\limits_{\tilde{x}_i}\tilde{x}_iP(\tilde{x}_i|v^n,\tilde{y}_\calA^n)d\tilde{x}_i\Big)^2\times\nonumber\\ &\quad P(x_i,v^n,\tilde{y}_\calA^n)dx_idv^nd\tilde{y}_\calA^n\nonumber\\
&=\bbE\Big[\Big(X_i-
E[X_i|V^n,\tilde{Y}_\calA^n]\Big)^2\Big],\nonumber%
\end{align}where
\begin{itemize}
    \item[$(a)$] follows since $\mathbb{E}[X_i|v^n,\mathbf{y}_\calA^n]=\int\limits_{\tilde{x}_i}\tilde{x}_iP(\tilde{x}_i|v^n,\mathbf{y}_\calA^n)d\tilde{x}_i$;
    \item[$(b)$] follows from $P(\tilde{x}_i|v^n,\mathbf{y}_\calA^n)=P(\tilde{x}_i|v^n,\mathbf{y}_\calA^n,\tilde{y}_\calA^n)=P(\tilde{x}_i|v^n,\tilde{y}_\calA^n)$ by the Markov chains in \eqref{eq:Suff_Stats_Markov_1} and \eqref{eq:Suff_Stats_Markov_2}.
\end{itemize}
Therefore, rewriting \eqref{eq:System_Model_2} as \eqref{eq:Suff_Stats_Inp_Outp} preserves the distortion constraint. We now show that it also preserves the information leakage,
\begin{align}
\MI(X^n;M,\mathbf{Y}_\calB^n)&=\MI(X^n;V^n,\mathbf{Y}_\calB^n)\nonumber\\
&=\MI(X^n;V^n)+\MI(X^n;\mathbf{Y}_\calB^n|V^n)\nonumber\\
&\mathop = \limits^{(a)}\MI(X^n;V^n)+\MI(X^n;\mathbf{Y}_\calB^n,\tilde{Y}_\calB^n|V^n)\nonumber\\
&\mathop = \limits^{(b)}\MI(X^n;V^n)+\MI(X^n;\tilde{Y}_\calB^n|V^n)\nonumber\\
&= \MI(X^n;V^n,\tilde{Y}_\calB^n),\nonumber%
\end{align}where $(a)$ and $(b)$ follow from \eqref{eq:Suff_Stats_Markov_3} and \eqref{eq:Suff_Stats_Markov_4}, respectively.  
Next, when $\tr(\boldsymbol{\Sigma}_\calB^{-1})\le\tr(\boldsymbol{\Sigma}_\calA^{-1})$, we redefine $\tilde{Y}_\calB$ as,
\begin{align}
    \tilde{Y}_\calB&=\frac{h_\calB}{h_\calA}\tilde{Y}_\calA+N',\nonumber
\end{align}where $N'\sim\calN\left(0,\tr(\boldsymbol{\Sigma}_\calB^{-1})\left(1-\frac{\tr(\boldsymbol{\Sigma}_\calB^{-1})}{\tr(\boldsymbol{\Sigma}_\calA^{-1})}\right)\right)$. Note that after redefining $\tilde{Y}_\calB$,  the joint distribution between $X$ and $\tilde{Y}_\calA$ and the joint distribution between $X$ and $\tilde{Y}_\calB$ are preserved, therefore the constraints in Definition~\ref{eq:Defi_Achi_Rate} are preserved. As a result, we have the Markov chain $X-\tilde{Y}_\calA-\tilde{Y}_\calB$.
Hence, when $\tr(\boldsymbol{\Sigma}_\calB^{-1})\le\tr(\boldsymbol{\Sigma}_\calA^{-1})$, without loss of generality, we can suppose that $X-\tilde{Y}_\calA-\tilde{Y}_\calB$. In this first case, we informally say that the authorized set of users have better side information. Similarly, when $\tr(\boldsymbol{\Sigma}_\calA^{-1})<\tr(\boldsymbol{\Sigma}_\calB^{-1})$, we can suppose that $X-\tilde{Y}_\calB-\tilde{Y}_\calA$. In this second case, we informally say that the unauthorized set of users have a better side information. 

We now study each of these two cases separately in Section~\ref{subsec:Better_Authorized} and Section~\ref{subsec:Better_Unauthorized}.
\subsection{When Authorized Users Have Better Side Information}
\label{subsec:Better_Authorized}
Fix $\calA\in\bbA$, and $\calB\in\bbB$. Suppose $\tr(\boldsymbol{\Sigma}_\calB^{-1})\le\tr(\boldsymbol{\Sigma}_\calA^{-1})$. In this case, as discussed in Section~\ref{sec:Conversion_Scalar}, the union in Theorem~\ref{thm:Upper_Bound} can be taken over $(U,V,X,\tilde{Y}_\calA,\tilde{Y}_\calB)$ such that $U-V-X-\tilde{Y}_\calA-\tilde{Y}_\calB$.~Then,
\begin{align}
\Delta&\ge\MI(V;X)-\MI(V;\tilde{Y}_\calA|U)+\MI(X;\tilde{Y}_\calB|U)\nonumber\\
&\mathop = \limits^{(a)}\MI(V;X)-\MI(V;\tilde{Y}_\calA)+\MI(U;\tilde{Y}_\calA)+\MI(X;\tilde{Y}_\calB)-I(U;\tilde{Y}_\calB)\nonumber\\
&\mathop \ge \limits^{(b)} \MI(V;X)-\MI(V;\tilde{Y}_\calA)+\MI(X;\tilde{Y}_\calB)\nonumber\\
&\mathop = \limits^{(c)} \MI(V;X|\tilde{Y}_\calA)+\MI(X;\tilde{Y}_\calB),\nonumber%
\end{align}where $(a)$, $(b)$, and $(c)$ follow since $U-V-X-\tilde{Y}_\calA-\tilde{Y}_\calB$. This implies that the region in Theorem~\ref{thm:Upper_Bound} is included in the following region
\begin{align}
  \bigcap_{(\calA,\calB)\in(\bbA,\bbB)}\bigcup\limits_{\substack{V-X-\tilde{Y}_\calA-\tilde{Y}_\calB \\ \bbE\big[\sigma_{X|\tilde{Y}_\calA,V}^2\big]\le D}} \left.\begin{cases}(R,\Delta):\\
R>\MI(V;X|\tilde{Y}_\calA)\\
\Delta>\MI(V;X|\tilde{Y}_\calA)+\MI(X;\tilde{Y}_\calB)
\end{cases}\hspace{-0.3cm}\right\}.\label{eq:Upper_Bound_2}
\end{align}
Optimizing the rate and the leakage constraints in \eqref{eq:Upper_Bound_2} separately results in a larger region, i.e., an outer region. As a result, the region in \eqref{eq:Upper_Bound_2} is included in the following region,
\begin{align}
  \hspace{-0.3cm}\bigcap_{(\calA,\calB)\in(\bbA,\bbB)} \left.\begin{cases}(R,\Delta):\\
R>\min\limits_{\substack{V-X-\tilde{Y}_\calA-\tilde{Y}_\calB \\ \bbE\big[\sigma_{X|\tilde{Y}_\calA,V}^2\big]\le D}}\MI(V;X|\tilde{Y}_\calA)\\
\Delta>\min\limits_{\substack{V-X-\tilde{Y}_\calA-\tilde{Y}_\calB \\ \bbE\big[\sigma_{X|\tilde{Y}_\calA,V}^2\big]\le D}}\big[\MI(V;X|\tilde{Y}_\calA)+\MI(X;\tilde{Y}_\calB)\big]
\end{cases}\hspace{-0.3cm}\right\}.\label{eq:Upper_Bound_22}
\end{align}
Since the source is Gaussian the term $\MI(X;\tilde{Y}_\calB)$ is fixed, and we know that the term $\MI(V;X|\tilde{Y}_\calA)=\dent(X|\tilde{Y}_\calA)-\dent(X|\tilde{Y}_\calA,V)$ is minimized by joint Gaussian $(V,X,\tilde{Y}_\calA)$ \cite[Lemma~1]{Thomas87}. Hence, the region in \eqref{eq:Upper_Bound_22} is again included in the intersection of all $(\calA,\calB)\in(\bbA,\bbB)$, i.e., $ \bigcap_{(\calA,\calB)\in(\bbA,\bbB)}$ of the following region,
\begin{align}
  \left.\begin{cases}(R,\Delta):\\
R>\min\limits_{\sigma_{X|\tilde{Y}_\calA,V}^2\le D}\frac{1}{2}\log\frac{\sigma_{X|\tilde{Y}_\calA}^2}{\sigma_{X|\tilde{Y}_\calA,V}^2}\\
\Delta>\min\limits_{\sigma_{X|\tilde{Y}_\calA,V}^2\le D}\left[\frac{1}{2}\log\frac{\sigma_{X|\tilde{Y}_\calA}^2}{\sigma_{X|\tilde{Y}_\calA,V}^2}+\frac{1}{2}\log\frac{\sigma_X^2}{\sigma_{X|\tilde{Y}_\calB}^2}\right]\\
\end{cases}\hspace{-0.3cm}\right\}.\nonumber
\end{align}
From the monotonicity of the log function, the region above is included,
\begin{align}
   \bigcap_{(\calA,\calB)\in(\bbA,\bbB)}\left.\begin{cases}(R,\Delta):\\
R>\frac{1}{2}\log\frac{\sigma_{X|\tilde{Y}_\calA}^2}{D}\\
\Delta>\frac{1}{2}\log\frac{\sigma_{X|\tilde{Y}_\calA}^2}{D}+\frac{1}{2}\log\frac{\sigma_X^2}{\sigma_{X|\tilde{Y}_\calB}^2}\\
\end{cases}\hspace{-0.3cm}\right\}.\label{eq:Upper_Bound_4}
\end{align}
Now we have,
\begin{subequations}\label{eq:SXgivenYAB}
\begin{align}
    \sigma_{X|\tilde{Y}_\calA}^2&=\sigma_X^2-\frac{\sigma_{X,\tilde{Y}_\calA}^2}{\sigma_{\tilde{Y}_\calA}^2}\nonumber\\
    &\mathop = \limits^{(a)}\sigma_X^2-\frac{h_\calA^2\sigma_X^4}{h_\calA^2\sigma_X^2+\tr\left(\boldsymbol{\Sigma}_\calA^{-1}\right)}\nonumber\\
    &\mathop = \limits^{(b)}\sigma_X^2-\frac{\tr\left(\boldsymbol{\Sigma}_\calA^{-1}\right)\sigma_X^4}{\tr\left(\boldsymbol{\Sigma}_\calA^{-1}\right)\sigma_X^2+1}\nonumber\\
    &=\frac{\sigma_X^2}{\tr\left(\boldsymbol{\Sigma}_\calA^{-1}\right)\sigma_X^2+1},\label{eq:SXgivenYA}
\end{align}where 
\begin{itemize}
    \item[$(a)$] follows by calculating $\sigma_{X,\tilde{Y}_\calA}^2$ and $\sigma_{\tilde{Y}_\calA}^2$ from \eqref{eq:Suff_Stats_Inp_Outp};
    \item[$(b)$] follows since from \eqref{eq:Suff_Stats_Inp_Outp} we have $h_\calA=\tr\left(\boldsymbol{\Sigma}_\calA^{-1}\right)$.
\end{itemize}Similarly, we have
\begin{align}
    \sigma_{X|\tilde{Y}_\calB}^2&=\frac{\sigma_X^2}{\tr\left(\boldsymbol{\Sigma}_\calB^{-1}\right)\sigma_X^2+1}.\label{eq:SXgivenYB}
\end{align}
\end{subequations}
Hence, the region in \eqref{eq:Upper_Bound_4} can be written as the intersection of all $(\calA,\calB)\in(\bbA,\bbB)$, i.e., $ \bigcap_{(\calA,\calB)\in(\bbA,\bbB)}$ of the following region,
\begin{align}
   \left.\begin{cases}(R,\Delta):\\
R>\frac{1}{2}\log\frac{\sigma_X^2}{D\big(\tr\left(\boldsymbol{\Sigma}_\calA^{-1}\right)\sigma_X^2+1\big)}\\
\Delta>\frac{1}{2}\log\frac{\sigma_X^2}{D\big(\tr\left(\boldsymbol{\Sigma}_\calA^{-1}\right)\sigma_X^2+1\big)}+\frac{1}{2}\log\big(\tr\left(\boldsymbol{\Sigma}_\calB^{-1}\right)\sigma_X^2+1\big)\\
\end{cases}\hspace{-0.4cm}\right\}.\label{eq:Upper_Bound_6}
\end{align}Since the arguments of the log functions are decreasing in $\tr\left(\boldsymbol{\Sigma}_\calA^{-1}\right)$ and increasing in $\tr\left(\boldsymbol{\Sigma}_\calB^{-1}\right)$ we can compute the intersection in \eqref{eq:Upper_Bound_6} and rewrite the region in \eqref{eq:Upper_Bound_6} as follows,
\begin{align}
   \left.\begin{cases}(R,\Delta):\\
R>\frac{1}{2}\log\frac{\sigma_X^2}{D\big(\tr\big(\boldsymbol{\Sigma}_{\calA^\star}^{-1}\big)\sigma_X^2+1\big)}\\
\Delta>\frac{1}{2}\log\frac{\sigma_X^2}{D\big(\tr\big(\boldsymbol{\Sigma}_{\calA^\star}^{-1}\big)\sigma_X^2+1\big)}+\frac{1}{2}\log\big(\tr\big(\boldsymbol{\Sigma}_{\calB^\star}^{-1}\big)\sigma_X^2+1\big)\\
\end{cases}\hspace{-0.4cm}\right\},\nonumber
\end{align}where $\calA^\star\in \argmin\limits_{\calA\in\bbA}\{\tr\left(\boldsymbol{\Sigma}_\calA^{-1}\right)\}$ and $\calB^\star\in \argmax\limits_{\calB\in\bbB}\{\tr\left(\boldsymbol{\Sigma}_\calB^{-1}\right)\}$. Note that we also have $\Delta\ge\MI(X;\mathbf{Y}_{\calB^\star})=\frac{1}{2}\log\left(1+\sigma_X^2\tr\big(\boldsymbol{\Sigma}_{\calB^\star}^{-1}\big)\right)$, whence the definition of $g_1$, in Theorem~\ref{thm:Capacity}.

\subsection{When Unauthorized Users Have Better Side Information}
\label{subsec:Better_Unauthorized}
Fix $\calA\in\bbA$, $\calB\in\bbB$, and suppose that $\tr(\boldsymbol{\Sigma}_\calA^{-1})<\tr(\boldsymbol{\Sigma}_\calB^{-1})$. We will need the following~lemma.
\begin{lemma}\label{lemma:Conversion_XgYV_to_XgV}
Consider $Y=hX+N$, where $h$ is a constant, and $X$ and $N$ are independent, zero-mean Gaussian random variables with variance $\sigma_X^2$ and $\sigma_N^2$, respectively. When $D\le\sigma_{X|Y}^2$ and $V$ is Gaussian, we have
\begin{itemize}
    \item $\frac{\sigma_N^2}{h^2}-D>0$;
    \item $\sigma_{X|V,Y}^2\le D \Leftrightarrow \sigma_{X|V}^2\le(D^{-1}-h^2\sigma_N^{-2})^{-1}$;
    \item $\sigma_{X|V,Y}^2= D \Leftrightarrow \sigma_{X|V}^2= (D^{-1}-h^2\sigma_N^{-2})^{-1}$.
\end{itemize} 
\end{lemma}The proof of Lemma~\ref{lemma:Conversion_XgYV_to_XgV} is provided in Appendix~\ref{proof:Conversion_XgYV_to_XgV}.

As discussed in Section~\ref{sec:Conversion_Scalar}, the union in Theorem~\ref{thm:Upper_Bound} can be taken over $(U,V,X,\tilde{Y}_\calA,\tilde{Y}_\calB)$ such that $U-V-X-\tilde{Y}_\calB-\tilde{Y}_\calA$. Similar to Section~\ref{subsec:Better_Authorized}, optimizing the rate and equivocation constraints in Theorem~\ref{thm:Upper_Bound} separately results in a larger region, i.e., an outer region. Optimizing the rate in Theorem~\ref{thm:Upper_Bound} yields,
\begin{align}
R&\ge \min\limits_{\substack{V-X-\tilde{Y}_\calA \\ \bbE\big[\sigma_{X|\tilde{Y}_\calA,V}^2\big]\le D}}\MI(V;X|\tilde{Y}_\calA)\nonumber\\
&\ge\frac{1}{2}\log\frac{\sigma_X^2}{D\big(\tr\left(\boldsymbol{\Sigma}_\calA^{-1}\right)\sigma_X^2+1\big)},\label{eq:Second_Case_First_Term}
\end{align}where the last inequality holds as in the derivation of \eqref{eq:Upper_Bound_6}.

Now optimizing the information leakage constraint in Theorem~\ref{thm:Upper_Bound} yields,
\begin{align}
\Delta&>\min\limits_{\substack{U-V-X-\tilde{Y}_\calB-\tilde{Y}_\calA \\ \bbE\big[\sigma_{X|\tilde{Y}_\calA,V}^2\big]\le D}}\big[\MI(V;X)-\MI(V;\tilde{Y}_\calA|U)+\MI(X;\tilde{Y}_\calB|U)\big]\nonumber\\
&\mathop = \limits^{(a)}\min\limits_{\substack{U-V-X-\tilde{Y}_\calB-\tilde{Y}_\calA \\ \bbE\big[\sigma_{X|\tilde{Y}_\calA,V}^2\big]\le D}}\big[\MI(V;X)-I(V;\tilde{Y}_\calA)+\MI(U;\tilde{Y}_\calA)+\nonumber\\
&\qquad\MI(X;\tilde{Y}_\calB)-\MI(U;\tilde{Y}_\calB)\big]\nonumber\\
&\mathop = \limits^{(b)}\MI(X;\tilde{Y}_\calB)+\min\limits_{\substack{U-V-X-\tilde{Y}_\calB-\tilde{Y}_\calA \\ \bbE\big[\sigma_{X|\tilde{Y}_\calA,V}^2\big]\le D}}\big[\MI(V;X)-\MI(V;\tilde{Y}_\calA)-\nonumber\\
&\qquad\MI(U;\tilde{Y}_\calB|\tilde{Y}_\calA)\big]\nonumber\\
&\mathop\ge \limits^{(c)} \MI(X;\tilde{Y}_\calB)+\min\limits_{\substack{U-V-X-\tilde{Y}_\calB-\tilde{Y}_\calA \\ \bbE\big[\sigma_{X|\tilde{Y}_\calA,V}^2\big]\le D}}\big[\MI(V;X)-\MI(V;\tilde{Y}_\calA)-\nonumber\\
&\qquad\MI(V;\tilde{Y}_\calB|\tilde{Y}_\calA)\big]\nonumber\\
&= \MI(X;\tilde{Y}_\calB)+\min\limits_{\substack{V-X-\tilde{Y}_\calB-\tilde{Y}_\calA \\ \bbE\big[\sigma_{X|\tilde{Y}_\calA,V}^2\big]\le D}}\big[\MI(V;X)-I(V;\tilde{Y}_\calA,\tilde{Y}_\calB)\big]\nonumber\\
&\mathop = \limits^{(d)} \MI(X;\tilde{Y}_\calB)+\min\limits_{\substack{V-X-\tilde{Y}_\calB-\tilde{Y}_\calA \\ \bbE\big[\sigma_{X|\tilde{Y}_\calA,V}^2\big]\le D}}\big[\MI(V;X)-\MI(V;\tilde{Y}_\calB)\big]\nonumber\\
&\mathop = \limits^{(e)} \min\limits_{\substack{V-X-\tilde{Y}_\calB-\tilde{Y}_\calA \\ \bbE\big[\sigma_{X|\tilde{Y}_\calA,V}^2\big]\le D}}C(V),\label{eq:Cost_Func}
\end{align}where
\begin{itemize}
    \item[$(a)$] follows since $U-V-X-\tilde{Y}_\calB-\tilde{Y}_\calA$ forms a Markov Chain;
    \item[$(b)$] follows since $U-\tilde{Y}_\calB-\tilde{Y}_\calA$ forms a Markov Chain;
    \item[$(c)$] follows since $U-V-\tilde{Y}_\calB-\tilde{Y}_\calA$ forms a Markov Chain;
    \item[$(d)$] follows since $V-\tilde{Y}_\calB-\tilde{Y}_\calA$ forms a Markov Chain;
    \item[$(e)$] follows by defining,
    \begin{align}
    C(V)&\triangleq\MI(X;\tilde{Y}_\calB)+\MI(V;X)-\MI(V;\tilde{Y}_\calB)\nonumber\\
    &=\dent(\tilde{Y}_\calB|V)-\dent(X|V)+k_1\nonumber\\
        &=\dent\left(\frac{1}{h_\calB}\tilde{Y}_\calB\Big|V\right)-\dent(X|V)+\log\abs{h_\calB}+k_1,\nonumber\\
        &=\dent\left(\frac{1}{h_\calB}\tilde{Y}_\calB\Big|V\right)-\dent(X|V)+k_2,\label{eq:cost_function}%
    \end{align}where $k_1\triangleq\MI(X;\tilde{Y}_\calB)+\dent(\tilde{Y}_\calB)-\dent(X)$ is a constant which is independent of $V$ and $k_2\triangleq \log\abs{h_\calB}+k_1$.
\end{itemize}
\begin{lemma}\label{lemma:Minimizer}
    When $X$ and $\tilde{Y}_\calB$ are Gaussian random variables, as defined in \eqref{eq:Suff_Stats_Coeffies}, and $V-X-\tilde{Y}_\calB$ forms a Markov chain
    \begin{align}
        C(V)
        &\triangleq\dent\left(\frac{1}{h_\calB}\tilde{Y}_\calB\Big|V\right)-\dent(X|V),\nonumber
    \end{align}
    is minimized when the auxiliary random variable $V$ is a Gaussian random variable.
\end{lemma}
\begin{proof}
    The proof follows from the extremal inequality \cite{Liu07} and \cite[Theorem~1]{WangChen13}. For completeness, we prove Lemma~\ref{lemma:Minimizer} in Appendix~\ref{proof:Optimal_is_gaussian}.
\end{proof} 
Hence, we can rewrite the \ac{RHS} of \eqref{eq:Cost_Func} as follows,
\begin{align}
   &\min\limits_{\substack{V-X-\tilde{Y}_\calB-\tilde{Y}_\calA \\ \bbE\big[\sigma_{X|\tilde{Y}_\calA,V}^2\big]\le D}} C(V)\nonumber\\
   &\mathop = \limits^{(a)}\min\limits_{\substack{V-X-\tilde{Y}_\calB-\check{Y}_\calA \nonumber\\
V\,\,\text{is Gaussian}\\\sigma_{X|V,\check{Y}_\calA}^2\le D}}
\Big[\MI(X;\tilde{Y}_\calB)+\MI(V;X)-\MI(V;\tilde{Y}_\calB)\Big]\nonumber\\
&\mathop = \limits^{(b)}\min\limits_{\substack{V-X-\tilde{Y}_\calB-\check{Y}_\calA \nonumber\\
V\,\,\text{is Gaussian}\\\sigma_{X|V}^2\le F_\calA(D)}}
\Big[\MI(X;\tilde{Y}_\calB)+\MI(V;X)-\MI(V;\tilde{Y}_\calB)\Big]\nonumber\\
&\mathop = \limits^{(c)}\min\limits_{\substack{\sigma_{X|V}^2\le F_\calA(D)}}
\left[\frac{1}{2}\log\frac{h_\calB^2\sigma_X^2+\tilde{\sigma}_\calB^2}{\tilde{\sigma}_\calB^2}+\frac{1}{2}\log\frac{\sigma_X^2}{\sigma_{X|V}^2}-\right.\nonumber\\
&\qquad\left.\frac{1}{2}\log\frac{h_\calB^2\sigma_X^2+\tilde{\sigma}_\calB^2}{h_\calB^2\sigma_{X|V}^2+\tilde{\sigma}_\calB^2}\right]\nonumber\\
&\mathop \ge \limits^{(d)}\frac{1}{2}\log\frac{h_\calB^2\sigma_X^2+\tilde{\sigma}_\calB^2}{\tilde{\sigma}_\calB^2}+\frac{1}{2}\log\frac{\sigma_X^2}{F_\calA(D)}-\nonumber\\
&\qquad\frac{1}{2}\log\frac{h_\calB^2\sigma_X^2+\tilde{\sigma}_\calB^2}{h_\calB^2F_\calA(D)+\tilde{\sigma}_\calB^2}\nonumber\\
&\mathop = \limits^{(e)} \frac{1}{2}\log\left(\tr\left(\boldsymbol{\Sigma}_\calB^{-1}\right)F_\calA(D)+1\right)+\frac{1}{2}\log\frac{\sigma_X^2}{F_\calA(D)},\label{eq:Second_Case_Second_Term}
\end{align}where 
\begin{itemize}
    \item[$(a)$] follows from Lemma~\ref{lemma:Minimizer};
    \item[$(b)$] follows from Lemma~\ref{lemma:Conversion_XgYV_to_XgV} with \begin{align}
        F_\calA(D)&\triangleq\frac{\tilde{\sigma}_\calA^2D}{\tilde{\sigma}_\calA^2-h_\calA^2D}=\frac{D}{1-\tr\big(\boldsymbol{\Sigma}_{\calA^\star}^{-1}\big)D};
    \end{align}
    \item[$(c)$] follows by calculating the mutual information of the random variables defined in \eqref{eq:Suff_Stats_Inp_Outp} using the fact that $\tilde{N}_\calB$ is independent of $(X,V)$, specifically,
    \begin{align}
        &\MI(V;\tilde{Y}_\calB)=\dent(\tilde{Y}_\calB)-\dent(\tilde{Y}_\calB|V)\nonumber\\
        &=\frac{1}{2}\log\left(2\pi e(h_\calB^2\sigma_X^2+\tilde{\sigma}_\calB^2)\right)-\nonumber\\
        &\qquad\frac{1}{2}\log\left(2\pi e(h_\calB^2\sigma_{X|V}^2+\tilde{\sigma}_\calB^2)\right)\nonumber\\
        &=\frac{1}{2}\log\frac{h_\calB^2\sigma_X^2+\tilde{\sigma}_\calB^2}{h_\calB^2\sigma_{X|V}^2+\tilde{\sigma}_\calB^2};\nonumber
    \end{align}
    \item[$(d)$] follows since $\frac{h_\calB^2\sigma_{X|V}^2+\tilde{\sigma}_\calB^2}{\sigma_{X|V}^2}$ is a monotonically decreasing function in $\sigma_{X|V}^2$;
    \item[$(e)$] follows since from \eqref{eq:Suff_Stats_Inp_Outp}, $h_\calB=\tilde{\sigma}_\calB^2=\tr\left(\boldsymbol{\Sigma}_\calB^{-1}\right)$.
\end{itemize}Therefore, by \eqref{eq:Second_Case_First_Term}, \eqref{eq:Cost_Func}, and \eqref{eq:Second_Case_Second_Term}, the region in Theorem~\ref{thm:Upper_Bound} is included in the intersection of all $(\calA,\calB)\in(\bbA,\bbB)$, i.e., $ \bigcap_{(\calA,\calB)\in(\bbA,\bbB)}$ of the following region,
\begin{align}
   \left.\begin{cases}(R,\Delta):\\
R>\frac{1}{2}\log\frac{\sigma_X^2}{D\big(\tr\left(\boldsymbol{\Sigma}_\calA^{-1}\right)\sigma_X^2+1\big)}\\
\Delta>\frac{1}{2}\log\left(\tr\left(\boldsymbol{\Sigma}_\calB^{-1}\right)F_\calA(D)+1\right)+\frac{1}{2}\log\frac{\sigma_X^2}{F_\calA(D)}\\
\end{cases}\right\}.\label{eq:Upper_Bound_Sec_Case}
\end{align}
Then, the region in \eqref{eq:Upper_Bound_Sec_Case} is included in the following region, 
\begin{align}
   \left.\begin{cases}(R,\Delta):\\
R>\max\limits_{\calA\in\bbA}\frac{1}{2}\log\frac{\sigma_X^2}{D\big(\tr\left(\boldsymbol{\Sigma}_\calA^{-1}\right)\sigma_X^2+1\big)}\\
\Delta>\max\limits_{(\calA,\calB)\in(\bbA,\bbB)}\frac{1}{2}\log\left(\tr\left(\boldsymbol{\Sigma}_\calB^{-1}\right)\sigma_X^2+\frac{\sigma_X^2}{F_\calA(D)}\right)\\
\end{cases}\right\}.\nonumber
\end{align}
Since the arguments of the $\log$ function for the bound on $R$ is decreasing in $\tr\left(\boldsymbol{\Sigma}_\calA^{-1}\right)$ and the argument of the $\log$ function for the bound on $\Delta$ is decreasing in $\tr\left(\boldsymbol{\Sigma}_\calA^{-1}\right)$ and increasing in $\tr\left(\boldsymbol{\Sigma}_\calB^{-1}\right)$, we can compute rewrite the region in \eqref{eq:Upper_Bound_Sec_Case} as follows,
\begin{align}
   \left.\begin{cases}(R,\Delta):\\
R>\frac{1}{2}\log\frac{\sigma_X^2}{D\big(\tr\big(\boldsymbol{\Sigma}_{\calA^\star}^{-1}\big)\sigma_X^2+1\big)}\\
\Delta>\frac{1}{2}\log\left(\tr\left(\boldsymbol{\Sigma}_{\calB^\star}^{-1}\right)\sigma_X^2+\frac{\sigma_X^2}{F_{\calA^\star}(D)}\right)
\end{cases}\right\},\nonumber
\end{align}where $\calA^\star\in \argmin\limits_{\calA\in\bbA}\{\tr\left(\boldsymbol{\Sigma}_\calA^{-1}\right)\}$, $\calB^\star\in \argmax\limits_{\calB\in\bbB}\{\tr\left(\boldsymbol{\Sigma}_\calB^{-1}\right)\}$, and $F_{\calA^\star}(D)=\frac{D}{1-\tr\big(\boldsymbol{\Sigma}_{\calA^\star}^{-1}\big)D}$.

Note that, from \eqref{eq:Equiv}, we have $\Delta\ge\MI(X;\mathbf{Y}_\calB)=\frac{1}{2}\log\left(1+\frac{\sigma_X^2}{\tr\big(\boldsymbol{\Sigma}_{\calB^\star}^{-1}\big)^{-1}}\right)$.
Therefore, we have the following bound on the leakage rate,
\begin{align}
&\Delta\ge\max\left\{\frac{1}{2}\log\left(1+\frac{\sigma_X^2}{\tr\big(\boldsymbol{\Sigma}_{\calB^\star}^{-1}\big)^{-1}}\right),\right.\nonumber\\
&\left.\frac{1}{2}\log\left(\frac{\sigma_X^2}{D}+\frac{\sigma_X^2}{\tr\big(\boldsymbol{\Sigma}_{\calB^\star}^{-1}\big)^{-1}}-\frac{\sigma_X^2}{\tr\big(\boldsymbol{\Sigma}_{\calA^\star}^{-1}\big)^{-1}}\right)\right\},  \nonumber%
\end{align}which can be written as $g_2\big(\calA^\star,\calB^\star\big)$ defined in Theorem~\ref{thm:Capacity}.

\section{Conclusion}
\label{sec:Conclusion}
In this paper, we study a secure source coding problem with multiple users when the encoder and the users observe copies of correlated scalar Gaussian random variables. Specifically, the objective is to guarantee a given distortion level of recovery of the source for some sets of authorized users and simultaneously minimize information leakage about the source for some other sets of users. This can be seen as a secret-sharing problem where perfect reconstruction of the secret is relaxed to an approximate reconstruction, and the perfect security requirement is relaxed to controlled information leakage. 
Our main result is the characterization of the optimal trade-off between the source compression rate, the desired distortion, and the information leakage for this problem. We note that characterizing this trade-off when the source is a vector Gaussian random variable is an open problem. 

\begin{appendices}

\section{Proof of Equation~\ref{eq:Y_X_Relation_General}}
\label{proof:System_Model}
For the sake of completeness, we present the following theorem from \cite[Theorem~3.5.2]{Gallager_Stochastic_Book}, which is essential in our proof.
\begin{theorem}[\cite{Gallager_Stochastic_Book}]
\label{thm:Gallager}
Let $\mathbf{X}$ and $\mathbf{Y}$ be zero-mean, jointly Gaussian, and jointly non-singular. Then $\mathbf{X}$ can be expressed as $\mathbf{X} = \mathbf{G}\mathbf{Y} + \mathbf{V}$, where $\mathbf{V}$ is statistically independent of $\mathbf{Y}$ and
\begin{align}
    \mathbf{G} &= \mathbf{K}_{XY}\mathbf{K}_Y^{-1}\nonumber\\
    \mathbf{K}_V &= \mathbf{K}_X - \mathbf{K}_{XY}\mathbf{K}_Y^{-1}\mathbf{K}_{XY}^\intercal.\nonumber
\end{align}
\end{theorem}
For every $\calA\in\bbA$, by Theorem~\ref{thm:Gallager}, we have
\begin{align}
    \mathbf{Y}_\calA=\boldsymbol{\Sigma}_{X\mathbf{Y}_\calA}\sigma_X^{-2} X+\mathbf{N}'_\calA,\label{eq:eqution_3}
\end{align}where $\boldsymbol{\Sigma}_{\mathbf{N}'_\calA}\triangleq\boldsymbol{\Sigma}_{\mathbf{Y}_\calA}-\boldsymbol{\Sigma}_{X\mathbf{Y}_\calA}\sigma_X^{-2}\boldsymbol{\Sigma}_{X\mathbf{Y}_\calA}^\intercal$ and $\boldsymbol{\Sigma}_{\mathbf{N}'_\calA}\succ 0$ from \cite[Eq.~(3.43)]{Gallager_Stochastic_Book}, since the covariance matrix $\boldsymbol{\Sigma}_{\mathbf{N}'_\calA}$ is invertible. Next, we normalize \eqref{eq:eqution_3} as follows. From Cholesky decomposition, there exists an invertible matrix $\mathbf{C}\in\bbR^{|\calA|\times|\calA|}$ such that $\boldsymbol{\Sigma}_{\mathbf{N}'_\calA}=\mathbf{C}\mathbf{C}^\intercal$, therefore, we can rewrite \eqref{eq:eqution_3} as follows
\begin{align}
    \mathbf{Y}'_\calA=\mathbf{h}_\calA X+\mathbf{N}''_\calA,\nonumber
\end{align}where $\mathbf{Y}'_\calA\triangleq\mathbf{C}^{-1}\mathbf{Y}_\calA$, $\mathbf{h}_\calA\triangleq\mathbf{C}^{-1}\boldsymbol{\Sigma}_{X\mathbf{Y}_\calA}\sigma_X^{-2}$, and $\mathbf{N}''_\calA\sim\calN(0,\mathbf{I}_{|\calA|})$. Similarly, for every $\calB\in\bbB$, one can show that
\begin{align}
    \mathbf{Y}'_\calB=\mathbf{h}_\calB X+\mathbf{N}''_\calB,\nonumber
\end{align}where $\mathbf{Y}'_\calB\triangleq\mathbf{C}^{-1}\mathbf{Y}_\calB$, $\mathbf{h}_\calB\triangleq\mathbf{C}^{-1}\boldsymbol{\Sigma}_{X\mathbf{Y}_\calB}\sigma_X^{-2}$, and $\mathbf{N}''_\calB\sim\calN(0,\mathbf{I}_{|\calB|})$.

\section{Proof of Theorem~\ref{thm:Threshold_Capacity} and Theorem~\ref{thm:Threshold_Capacity_2}}
\label{proof:Threshold_Capacity}
We first show that, there exist sets of authorized users  $\calA_t^\star\in\argmin\limits_{\calA\in\bbA_t}\tr\left(\boldsymbol{\Sigma}_{\calA}^{-1}\right)$ and unauthorized users $\calB_t^\star\in\argmax\limits_{\calB\in\bbB_t}\tr\left(\boldsymbol{\Sigma}_{\calB}^{-1}\right)$ such that for any $t\in\intseq{1}{L-1}$, $\calA_t^\star\subset\calA_{t+1}^\star$ and $\calB_t^\star\subset\calB_{t+1}^\star$. Then, by using Theorem~\ref{thm:Capacity}, we remark that $\calA_t^\star$ and $\calB_t^\star$ also correspond to the sets that appear in the expression of the optimal compression and the leakage rates for the threshold access structure $\bbA_t$. Ultimately, using the monotonicity of the sets $\big(\calA_t^\star\big)_{t\in\intseq{1}{L}}$ and $\big(\calB_t^\star\big)_{t\in\intseq{1}{L}}$ and Theorem~\ref{thm:Capacity}, we compute necessary and sufficient conditions to determine whether the optimal compression and leakage rates increase or decrease with the threshold $t$. 
\begin{lemma}
\label{lemma:Threshold}
There exist sets $\big(\calA_t^\star\big)_{t\in\intseq{1}{L}}$ and $\big(\calB_t^\star\big)_{t\in\intseq{1}{L}}$ such that, for any $t\in\intseq{1}{L-1}$, we have $\calA_t^\star\subset\calA_{t+1}^\star$ and $\calB_t^\star\subset\calB_{t+1}^\star$, and for any $t\in\intseq{1}{L}$, $\calA_t^\star\in\argmin\limits_{\calA\in\bbA_t}\tr\left(\boldsymbol{\Sigma}_{\calA}^{-1}\right)$ and $\calB_t^\star\in\argmax\limits_{\calB\in\bbB_t}\tr\left(\boldsymbol{\Sigma}_{\calB}^{-1}\right)$.
\end{lemma}
\begin{proof}
We write $\boldsymbol{\Sigma}_{\calL}$ as $\diag(\sigma_1^2,\dots,\sigma_L^2)$. Without loss of generality, suppose that $\sigma_1^2\ge\sigma_2^2\ge\dots\ge\sigma_L^2$. For $t\in\intseq{1}{L-1}$, let $\calA_t^\star\triangleq\intseq{1}{t}$ and $\calB_t^\star\triangleq\intseq{L-t+2}{L}$, since $\calA_t^\star\in\argmin\limits_{\calA:\abs{\calA}=t}\sum\limits_{i\in\calA}\sigma_i^{-2}$ and $\calB_t^\star\in\argmax\limits_{\calB:\abs{\calB}<t}\sum\limits_{i\in\calB}\sigma_i^{-2}$ we have $\calA_t^\star\subset\calA_{t+1}^\star$ and $\calB_t^\star\subset\calB_{t+1}^\star$. 
\end{proof}
When $t=L$, from Lemma~\ref{lemma:Threshold} we have $\tr\left(\boldsymbol{\Sigma}_{\calA_L^\star}^{-1}\right)>\tr\left(\boldsymbol{\Sigma}_{\calB_L^\star}^{-1}\right)$, therefore, the optimal leakage rate in Theorem~\ref{thm:Capacity} is,
\begin{subequations}\label{eq:D_1D_t}
\begin{align}
\Delta(D,\bbA_L)&=\frac{1}{2}\log\frac{\sigma_X^2}{D}-\frac{1}{2}\log\left(1+\frac{\sigma_X^2}{\tr\big(\boldsymbol{\Sigma}_{\calA_L^\star}^{-1}\big)^{-1}}\right)+\nonumber\\
&\qquad\frac{1}{2}\log\left(1+\frac{\sigma_X^2}{\tr\big(\boldsymbol{\Sigma}_{\calB_L^\star}^{-1}\big)^{-1}}\right)
\label{eq:D_1}
\end{align}and for $t\in\intseq{1}{L-1}$ and $\tr\big(\boldsymbol{\Sigma}_{\calA_t^\star}^{-1}\big)\ge\tr\big(\boldsymbol{\Sigma}_{\calB_t^\star}^{-1}\big)$,
\begin{align}
\Delta(D,\bbA_t)&=\frac{1}{2}\log\frac{\sigma_X^2}{D}-\frac{1}{2}\log\left(1+\frac{\sigma_X^2}{\tr\big(\boldsymbol{\Sigma}_{\calA_t^\star}^{-1}\big)^{-1}}\right)+\nonumber\\
&\qquad\frac{1}{2}\log\left(1+\frac{\sigma_X^2}{\tr\big(\boldsymbol{\Sigma}_{\calB_t^\star}^{-1}\big)^{-1}}\right).
\label{eq:D_t}
\end{align}By \eqref{eq:D_1}, \eqref{eq:D_t}, and monotonicity of the $\log$ function, we have
\begin{align}
    &\Delta(D,\bbA_L)\le \Delta(D,\bbA_t)\Leftrightarrow\nonumber\\
&\frac{\sigma_X^{-2}+\tr\big(\boldsymbol{\Sigma}_{\calA_t^\star}^{-1}\big)}{\sigma_X^{-2}+\tr\big(\boldsymbol{\Sigma}_{\calB_t^\star}^{-1}\big)}\le\frac{\sigma_X^{-2}+\tr\big(\boldsymbol{\Sigma}_{\calA_L^\star}^{-1}\big)}{\sigma_X^{-2}+\tr\big(\boldsymbol{\Sigma}_{\calB_L^\star}^{-1}\big)}.\label{eq:DLbDt}
\end{align}
Next, for $t\in\intseq{1}{L}$, and $\tr\big(\boldsymbol{\Sigma}_{\calA_t^\star}^{-1}\big)\le\tr\big(\boldsymbol{\Sigma}_{\calB_t^\star}^{-1}\big)$,
\begin{align}
\Delta(D,\bbA_t)=\frac{1}{2}\log\left(\frac{\sigma_X^2}{D}+\frac{\tr\big(\boldsymbol{\Sigma}_{\calB_t^\star}^{-1}\big)}{\sigma_X^{-2}}-\frac{\tr\big(\boldsymbol{\Sigma}_{\calA_t^\star}^{-1}\big)}{\sigma_X^{-2}}\right).\label{eq:D_t_2}
\end{align}Using \eqref{eq:D_1}, \eqref{eq:D_t_2}, and monotonicity of the $\log$ function, we have
\begin{align}
    &\Delta(D,\bbA_L)\le \Delta(D,\bbA_t)\Leftrightarrow\nonumber\\
    &\frac{1}{2}\log\left[\frac{\sigma_X^2\left(\sigma_X^{-2}+\tr\big(\boldsymbol{\Sigma}_{\calB_L^\star}^{-1}\big)\right)}{D\left(\sigma_X^{-2}+\tr\big(\boldsymbol{\Sigma}_{\calA_L^\star}^{-1}\big)\right)}\right]\le\nonumber\\
    &\frac{1}{2}\log\left(\frac{\sigma_X^2}{D}+\frac{\tr\big(\boldsymbol{\Sigma}_{\calB_t^\star}^{-1}\big)}{\sigma_X^{-2}}-\frac{\tr\big(\boldsymbol{\Sigma}_{\calA_t^\star}^{-1}\big)}{\sigma_X^{-2}}\right)\Leftrightarrow\nonumber\\
    & \tr\big(\boldsymbol{\Sigma}_{\calB_L^\star}^{-1}\big)-\tr\big(\boldsymbol{\Sigma}_{\calA_L^\star}^{-1}\big)\le\nonumber\\ &D\left(\sigma_X^{-2}+\tr\big(\boldsymbol{\Sigma}_{\calA_L^\star}^{-1}\big)\right)\left(\tr\big(\boldsymbol{\Sigma}_{\calB_t^\star}^{-1}\big)-\tr\big(\boldsymbol{\Sigma}_{\calA_t^\star}^{-1}\big)\right),
    \label{eq:DLbDt_2}
\end{align}which is always true since the \ac{RHS} \eqref{eq:DLbDt_2} is positive and the left-hand side of \eqref{eq:DLbDt_2} is negative because $\tr\big(\boldsymbol{\Sigma}_{\calB_L^\star}^{-1}\big)\le\tr\big(\boldsymbol{\Sigma}_{\calA_L^\star}^{-1}\big)$.  
\end{subequations}
Next, for $i\in\intseq{1}{L-t}$, by Theorem~\ref{thm:Capacity} and using that the $\log$ function is increasing, we have
\begin{align}
    R(D,\bbA_t)\ge R(D,\bbA_{t+i})
    &\Leftrightarrow\tr\big(\boldsymbol{\Sigma}_{\calA_t^\star}^{-1}\big)\le\tr\big(\boldsymbol{\Sigma}_{\calA_{t+i}^\star}^{-1}\big),\label{eq:Threshold_Third_Prove_2}
\end{align}and \eqref{eq:Threshold_Third_Prove_2} is always true by Lemma~\ref{lemma:Threshold}. Note that \eqref{eq:Threshold_Third_Prove_2} also proves that for any $t\in\intseq{1}{L}$, $R(D,\bbA_t)\ge R(D,\bbA_L)$.

We now prove Theorem~\ref{thm:Threshold_Capacity_2}. Let $i\in\intseq{1}{L-t}$. We consider four cases. First, when $\tr\big(\boldsymbol{\Sigma}_{\calA_t^\star}^{-1}\big)\le\tr\big(\boldsymbol{\Sigma}_{\calB_t^\star}^{-1}\big)$ and $\tr\big(\boldsymbol{\Sigma}_{\calA_{t+i}^\star}^{-1}\big)\le\tr\big(\boldsymbol{\Sigma}_{\calB_{t+i}^\star}^{-1}\big)$, we have
\begin{align}
    &\Delta(D,\bbA_t)\ge\Delta(D,\bbA_{t+i})\Leftrightarrow\nonumber\\
    &\frac{1}{2}\log\left(\frac{\sigma_X^2}{D}+\frac{\tr\big(\boldsymbol{\Sigma}_{\calB_t^\star}^{-1}\big)}{\sigma_X^{-2}}-\frac{\tr\big(\boldsymbol{\Sigma}_{\calA_t^\star}^{-1}\big)}{\sigma_X^{-2}}\right)\ge\nonumber\\
    &\frac{1}{2}\log\left(\frac{\sigma_X^2}{D}+\frac{\tr\big(\boldsymbol{\Sigma}_{\calB_{t+i}^\star}^{-1}\big)}{\sigma_X^{-2}}-\frac{\tr\big(\boldsymbol{\Sigma}_{\calA_{t+i}^\star}^{-1}\big)}{\sigma_X^{-2}}\right)\Leftrightarrow\nonumber\\
    &\tr\big(\boldsymbol{\Sigma}_{\calA_{t+i}^\star}^{-1}\big)-\tr\big(\boldsymbol{\Sigma}_{\calA_t^\star}^{-1}\big)\ge\tr\big(\boldsymbol{\Sigma}_{\calB_{t+i}^\star}^{-1}\big)-\tr\big(\boldsymbol{\Sigma}_{\calB_t^\star}^{-1}\big).\label{eq:Threshold_Sec_Prove_2}
\end{align}Second, when $\tr\big(\boldsymbol{\Sigma}_{\calA_t^\star}^{-1}\big)\ge\tr\big(\boldsymbol{\Sigma}_{\calB_t^\star}^{-1}\big)$ and $\tr\big(\boldsymbol{\Sigma}_{\calA_{t+i}^\star}^{-1}\big)\ge\tr\big(\boldsymbol{\Sigma}_{\calB_{t+i}^\star}^{-1}\big)$, we have
\begin{align}
    &\Delta(D,\bbA_t)\ge\Delta(D,\bbA_{t+i})\Leftrightarrow\nonumber\\
    &\frac{1}{2}\log\left(\frac{\sigma_X^{-2}+\tr\big(\boldsymbol{\Sigma}_{\calB_t^\star}^{-1}\big)}{\sigma_X^{-2}+\tr\big(\boldsymbol{\Sigma}_{\calA_t^\star}^{-1}\big)}\right)\ge\frac{1}{2}\log\left(\frac{\sigma_X^{-2}+\tr\big(\boldsymbol{\Sigma}_{\calB_{t+i}^\star}^{-1}\big)}{\sigma_X^{-2}+\tr\big(\boldsymbol{\Sigma}_{\calA_{t+i}^\star}^{-1}\big)}\right)\Leftrightarrow\nonumber\\
    &\frac{\sigma_X^{-2}+\tr\big(\boldsymbol{\Sigma}_{\calB_t^\star}^{-1}\big)}{\sigma_X^{-2}+\tr\big(\boldsymbol{\Sigma}_{\calA_t^\star}^{-1}\big)}\ge\frac{\sigma_X^{-2}+\tr\big(\boldsymbol{\Sigma}_{\calB_{t+i}^\star}^{-1}\big)}{\sigma_X^{-2}+\tr\big(\boldsymbol{\Sigma}_{\calA_{t+i}^\star}^{-1}\big)}.\label{eq:Threshold_Sec_Prove}
\end{align}Third, when $\tr\big(\boldsymbol{\Sigma}_{\calA_t^\star}^{-1}\big)\ge\tr\big(\boldsymbol{\Sigma}_{\calB_t^\star}^{-1}\big)$ and $\tr\big(\boldsymbol{\Sigma}_{\calA_{t+i}^\star}^{-1}\big)\le\tr\big(\boldsymbol{\Sigma}_{\calB_{t+i}^\star}^{-1}\big)$, we have,
\begin{align}
    &\Delta(D,\bbA_t)\ge\Delta(D,\bbA_{t+i})\Leftrightarrow\nonumber\\
    & \frac{1}{2}\log\left[\frac{\sigma_X^2\left(\sigma_X^{-2}+\tr\big(\boldsymbol{\Sigma}_{\calB_t^\star}^{-1}\big)\right)}{D\left(\sigma_X^{-2}+\tr\big(\boldsymbol{\Sigma}_{\calA_t^\star}^{-1}\big)\right)}\right]\ge\nonumber\\
    &\frac{1}{2}\log\left(\frac{\sigma_X^2}{D}+\frac{\tr\big(\boldsymbol{\Sigma}_{\calB_{t+i}^\star}^{-1}\big)}{\sigma_X^{-2}}-\frac{\tr\big(\boldsymbol{\Sigma}_{\calA_{t+i}^\star}^{-1}\big)}{\sigma_X^{-2}}\right)\Leftrightarrow\nonumber\\
    & \tr\big(\boldsymbol{\Sigma}_{\calB_t^\star}^{-1}\big)-\tr\big(\boldsymbol{\Sigma}_{\calA_t^\star}^{-1}\big)\ge\nonumber\\
    &D\left(\sigma_X^{-2}+\tr\big(\boldsymbol{\Sigma}_{\calA_t^\star}^{-1}\big)\right)\left(\tr\big(\boldsymbol{\Sigma}_{\calB_{t+i}^\star}^{-1}\big)-\tr\big(\boldsymbol{\Sigma}_{\calA_{t+i}^\star}^{-1}\big)\right),\label{eq:Threshold_Sec_Prove_4}
\end{align}since $\tr\big(\boldsymbol{\Sigma}_{\calB_t^\star}^{-1}\big)-\tr\big(\boldsymbol{\Sigma}_{\calA_t^\star}^{-1}\big)\le 0$ and $\tr\big(\boldsymbol{\Sigma}_{\calB_{t+i}^\star}^{-1}\big)-\tr\big(\boldsymbol{\Sigma}_{\calA_{t+i}^\star}^{-1}\big)\ge 0$ the inequality in the \ac{RHS}~of~\eqref{eq:Threshold_Sec_Prove_4} is never satisfied and $\Delta(D,\bbA_t)\le\Delta(D,\bbA_{t+i})$.
Fourth, when $\tr\big(\boldsymbol{\Sigma}_{\calA_t^\star}^{-1}\big)\le\tr\big(\boldsymbol{\Sigma}_{\calB_t^\star}^{-1}\big)$ and $\tr\big(\boldsymbol{\Sigma}_{\calA_{t+i}^\star}^{-1}\big)\ge\tr\big(\boldsymbol{\Sigma}_{\calB_{t+i}^\star}^{-1}\big)$, we have,
\begin{align}
    &\Delta(D,\bbA_t)\ge\Delta(D,\bbA_{t+i})\Leftrightarrow\nonumber\\
    &\frac{1}{2}\log\left(\frac{\sigma_X^2}{D}+\frac{\tr\big(\boldsymbol{\Sigma}_{\calB_t^\star}^{-1}\big)}{\sigma_X^{-2}}-\frac{\tr\big(\boldsymbol{\Sigma}_{\calA_t^\star}^{-1}\big)}{\sigma_X^{-2}}\right)\ge\nonumber\\
    &\frac{1}{2}\log\left[\frac{\sigma_X^2\left(\sigma_X^{-2}+\tr\big(\boldsymbol{\Sigma}_{\calB_{t+i}^\star}^{-1}\big)\right)}{D\left(\sigma_X^{-2}+\tr\big(\boldsymbol{\Sigma}_{\calA_{t+i}^\star}^{-1}\big)\right)}\right]\Leftrightarrow\nonumber\\
    & D\left(\sigma_X^{-2}+\tr\big(\boldsymbol{\Sigma}_{\calA_{t+i}^\star}^{-1}\big)\right)\left(\tr\big(\boldsymbol{\Sigma}_{\calB_t^\star}^{-1}\big)-\tr\big(\boldsymbol{\Sigma}_{\calA_t^\star}^{-1}\big)\right)\ge\nonumber\\
    &\tr\big(\boldsymbol{\Sigma}_{\calB_{t+i}^\star}^{-1}\big)-\tr\big(\boldsymbol{\Sigma}_{\calA_{t+i}^\star}^{-1}\big),\label{eq:Threshold_Sec_Prove_5}
\end{align}since $\tr\big(\boldsymbol{\Sigma}_{\calB_t^\star}^{-1}\big)-\tr\big(\boldsymbol{\Sigma}_{\calA_t^\star}^{-1}\big)\ge 0$ and $\tr\big(\boldsymbol{\Sigma}_{\calB_{t+i}^\star}^{-1}\big)-\tr\big(\boldsymbol{\Sigma}_{\calA_{t+i}^\star}^{-1}\big)\le 0$, the inequality in the \ac{RHS}~of~\eqref{eq:Threshold_Sec_Prove_5} is always satisfied and $\Delta(D,\bbA_t)\ge\Delta(D,\bbA_{t+i})$.

\begin{figure*}[b!]
\hrulefill
\begin{align}
    \begin{bmatrix}
    \sigma_V^2 & h\sigma_{XV}\\
    h\sigma_{XV} & \sigma_{Y}^2
   \end{bmatrix}^{-1}&=\big(\sigma_V^2\sigma_{Y}^2-h^2\sigma_{XV}^2\big)^{-1}
    \begin{bmatrix}
    \sigma_{Y}^2 & -h\sigma_{XV}\\
    -h\sigma_{XV} & \sigma_V^2
   \end{bmatrix}\nonumber\\
   &\mathop = \limits^{(a)} \begin{bmatrix}
    \sigma_V^{-2}\big(h^2\sigma_{X|V}^2+\sigma_N^2\big)^{-1}\big(h^2\sigma_X^2+\sigma_N^2\big) & -h\sigma_V^{-2}\sigma_{XV}\big(h^2\sigma_{X|V}^2+\sigma_N^2\big)^{-1}\\
    -h\sigma_V^{-2}\sigma_{XV}\big(h^2\sigma_{X|V}^2+\sigma_N^2\big)^{-1} & \big(h^2\sigma_{X|V}^2+\sigma_N^2\big)^{-1}
   \end{bmatrix}\nonumber%
\end{align}
\end{figure*}

\section{Distortion Constraint in Theorem~\ref{thm:Upper_Bound}}
\label{proof:Upper_Bound}
We have,
\begin{subequations}\label{eq:Distortion_Upper}
\begin{align}
D&\ge\frac{1}{n}\sum\limits_{i=1}^n\bbE\Big(X_i-\bbE\big[X_i\big|\mathbf{Y}_\calA^n,f(X^n)\big]\Big)^2\label{eq:Distortion_Upper_0}\\
&=\frac{1}{n}\sum\limits_{i=1}^n\bbE\left[\sigma_{X_i|\mathbf{Y}_\calA^n,f(X^n)}^2\right]\label{eq:Distortion_Upper_1}\\
&\mathop\ge\limits^{(a)}\frac{1}{n}\sum\limits_{i=1}^n\bbE\left[\sigma_{X_i|\mathbf{Y}_{\calA,i},\mathbf{Y}_{\calA\sim i}^n,f(X^n),\mathbf{Y}_\calB^{i-1},X^{i-1}}^2\right]\label{eq:Distortion_Upper_3}\\
&\mathop=\limits^{(b)}\frac{1}{n}\sum\limits_{i=1}^n\bbE_{\mathbf{Y}_{\calA,i},V_i}\left[\sigma_{X_i|\mathbf{Y}_{\calA,i},V_i}^2\right]\label{eq:Distortion_Upper_4}\\
&\mathop=\limits^{(c)}\sum\limits_{i=1}^n\bbP(Q=i)\bbE_{\mathbf{Y}_{\calA,Q},V_Q|Q=i}\left[\sigma_{X_Q|\mathbf{Y}_{\calA,Q},V_Q,Q=i}^2\right]\\
&=\bbE_{\mathbf{Y}_{\calA,Q},V_Q,Q}\left[\sigma_{X_Q|\mathbf{Y}_{\calA,Q},V_Q,Q}^2\right]\label{eq:Distortion_Upper_5}\\
&\mathop=\limits^{(d)}\bbE_{\mathbf{Y}_{\calA},V}\left[\sigma_{X|\mathbf{Y}_\calA,V}^2\right],\label{eq:Distortion_Upper_6}
\end{align}where
\begin{itemize}
    \item[$(a)$] follows since conditioning reduces \ac{MMSE} \cite[Lemma~13]{EkremUlukus13_Lossy};
    \item[$(b)$] follows by defining $V_i\triangleq\big(\mathbf{Y}_{\calA\sim i}^n,f(X^n),\mathbf{Y}_\calB^{i-1},X^{i-1}\big)$, which is consistent with the definition of $V_i$ in the proof of \cite[Theorem~3]{VillardPianta};
    \item[$(c)$] holds with $Q$ uniformly distributed over $\intseq{1}{n}$;
    \item[$(d)$] follows by defining $X\triangleq X_Q$, $\mathbf{Y}_\calA\triangleq\mathbf{Y}_{\calA,Q}$, and $V\triangleq(V_Q,Q)$ which is consistent with the definitions of these random variables in the proof of \cite[Theorem~3]{VillardPianta}.
\end{itemize}
\end{subequations}

\section{Proof of Lemma~\ref{lemma:Conversion_XgYV_to_XgV}}
\label{proof:Conversion_XgYV_to_XgV}
We prove the first statement of Lemma~\ref{lemma:Conversion_XgYV_to_XgV} as follows. By \cite[Chapter~3]{Gallager_Stochastic_Book}, we have
\begin{align}
    \sigma_{X|Y}^2&=\sigma_X^2-\frac{\sigma_{XY}^2}{\sigma_Y^2}\nonumber\\
    &\mathop=\limits^{(a)}\sigma_X^2-\frac{h^2\sigma_X^4}{h^2\sigma_X^2+\sigma_N^2}\nonumber\\
    &=\frac{\sigma_N^2}{h^2}-\sigma_N^4\big(h^4\sigma_X^2+h^2\sigma_N^2\big)^{-1}.\label{eq:Poor_Equality}
\end{align}where $(a)$ follows by calculating $\sigma_{XY}=h\sigma_X^2$ and $\sigma_Y^2=h^2\sigma_X^2+\sigma_N^2$. Hence, by \eqref{eq:Poor_Equality}, the constraint $D\le\sigma_{X|Y}^2$ can be expressed as
\begin{align}
    D\le\frac{\sigma_N^2}{h^2}-\sigma_N^4\big(h^4\sigma_X^2+h^2\sigma_N^2\big)^{-1},\nonumber%
\end{align}which can be rewritten as
\begin{align}
    0<\sigma_N^4\big(h^4\sigma_X^2+h^2\sigma_N^2\big)^{-1}\le\frac{\sigma_N^2}{h^2}-D.\label{eq:Lemma2_First_Statement}%
\end{align} Next, we prove the second statement of Lemma~\ref{lemma:Conversion_XgYV_to_XgV}. Since $(V,X,Y)$ are jointly Gaussian and $V$ is independent of $N$, $\sigma_{X|V,Y}^2$ is given by \cite[Chapter~3]{Gallager_Stochastic_Book}
\begin{align}
    &\sigma_{X|V,Y}^2=\sigma_X^2-\nonumber\\
    &\quad\big[\sigma_{XV}\quad h\sigma_X^2\big]\begin{bmatrix}
    \sigma_V^2 & h\sigma_{XV}\\
    h\sigma_{XV} & \sigma_{Y}^2
   \end{bmatrix}^{-1}\big[\sigma_{XV}\quad h\sigma_X^2\big]^\intercal,\nonumber%
\end{align} 
where 
$\left[ \begin{smallmatrix}
       \sigma_V^2 & h\sigma_{XV}\\[0.3em]
       h\sigma_{XV} & \sigma_{Y}^2 
\end{smallmatrix}\right]^{-1}$ is provided at the bottom of this page, in which $(a)$ follows since
\begin{align}
    \sigma_V^{-2}\big(\sigma_V^2\sigma_{Y}^2-h^2\sigma_{XV}^2\big)&=\sigma_{Y}^2-h^2\sigma_V^{-2}\sigma_{XV}^2\nonumber\\
    &=h^2\sigma_X^2-h^2\sigma_V^{-2}\sigma_{XV}^2+\sigma_N^2\nonumber\\
    &=h^2\sigma_{X|V}^2+\sigma_N^2,\nonumber%
\end{align}where the last equality holds since $\sigma_{X|V}^2=\sigma_X^2-\sigma_V^{-2}\sigma_{XV}^2$. 
Hence,
\begin{align}
     \sigma_{X|V,Y}^2&=\frac{\sigma_N^2\sigma_{X|V}^2}{h^2\sigma_{X|V}^2+\sigma_N^2}.\label{eq:Cov_XgVY_1}
\end{align}Then, by \eqref{eq:Cov_XgVY_1}, the constraint $\sigma_{X|V,Y}^2\le D$ can be expressed as follows:
\begin{align}
    &\frac{\sigma_N^2\sigma_{X|V}^2}{h^2\sigma_{X|V}^2+\sigma_N^2}\le D,\nonumber%
\end{align}and, since $\frac{\sigma_N^2}{h^2}-D>0$ by \eqref{eq:Lemma2_First_Statement}, we have
\begin{align}
    \sigma_{X|V}^2&\le\frac{\sigma_N^2D}{\sigma_N^2-h^2D}.\nonumber
\end{align}
The last statement of the lemma holds because if $D=\sigma_{X|V,Y}^2$, then
\begin{align}
    D&=\frac{\sigma_N^2\sigma_{X|V}^2}{h^2\sigma_{X|V}^2+\sigma_N^2},\label{eq:DsigmaXgV}
\end{align}where \eqref{eq:DsigmaXgV} follows from \eqref{eq:Cov_XgVY_1}, and solving \eqref{eq:DsigmaXgV} for $\sigma_{X|V}^2$ results to $\sigma_{X|V}^2=\frac{\sigma_N^2D}{\sigma_N^2-h^2D}$.

\section{Proof of Lemma~\ref{lemma:Minimizer}}
\label{proof:Optimal_is_gaussian}
Here, we show that for a given feasible $V$, we can construct a feasible Gaussian random variable $\bar{V}$ such that $C(V)=C(\bar{V})$. Thus, this implies that restricting $V$ to be Gaussian does not change the optimum value of the optimization problem. 
Next, to study the difference of the two differential entropy in \eqref{eq:cost_function}, we need the following properties of the Fisher information and the differential entropy.
\begin{definition}(\hspace{-0.15mm}\cite[Definition~1]{EkremUlukus13_MIMO})\label{defi:FI}
Let $X$ and $U$ be random variables with well-defined densities, and $f_{X|U}$ be the corresponding conditional density. The conditional Fisher information of $X$ is defined by
$$\FI(X|U)=\bbE\left[\left(\frac{\partial\log f_{X|U}(x|u)}{\partial x}\right)^2\right],$$ where the expectation is over~$(U,X)$.
\end{definition}

\begin{lemma}{(\cite[Lemma~18]{EkremUlukus13_Lossy})}
\label{lemma:FI_Variance}
Let $(V,X,G_1,G_2)$ be random variables such that $(V,X)$ and $(G_1,G_2)$
are independent, and let $G_1$ and $G_2$ be Gaussian random variables with variance $0<\sigma_1^2\le\sigma_2^2$. Then, we have
\begin{align}
    \FI(X+G_2|V)^{-1}-\sigma_2^2\ge\FI(X+G_1|V)^{-1}-\sigma_1^2.\nonumber
\end{align}
\end{lemma}
From \cite[Lemma~3]{EkremUlukus13_MIMO}, we have
\begin{align}
    \dent\left(\frac{1}{h_\calB}\tilde{Y}_\calB\Big|V\right)-\dent(X|V)=\frac{1}{2}\int_0^{\frac{\tilde{\sigma}_\calB^2}{h_\calB^2}}\FI\left(X+N|V\right)d\sigma_N^2,\label{eq:Conv_to_FI}
\end{align}where $N$ is a zero-mean Gaussian random variable with covariance $\sigma_N^2\ge 0$. 
We now outer region \eqref{eq:Conv_to_FI} by substituting $G_1\leftarrow \emptyset$ and $G_2\leftarrow N$ in Lemma~\ref{lemma:FI_Variance} as follows,
\begin{align}
    \FI\left(X+N|V\right)\le\left(\FI(X|V)^{-1}+\sigma_N^2\right)^{-1}.\label{eq:Upper_Bounding_FI}
\end{align}Substituting \eqref{eq:Upper_Bounding_FI} in \eqref{eq:Conv_to_FI}, we have,
\begin{align}
    \dent\left(\frac{1}{h_\calB}\tilde{Y}_\calB\Big|V\right)-\dent(X|V)\le\frac{1}{2}\log\frac{\FI(X|V)^{-1}+\frac{\tilde{\sigma}_\calB^2}{h_\calB^2}}{\FI(X|V)^{-1}}.\label{eq:Upper_Bound_FI}
\end{align}We can also lower bound the Fisher information in \eqref{eq:Conv_to_FI} by setting $G_2\leftarrow \frac{1}{h_\calB}\tilde{N}_\calB$ and $G_1\leftarrow N$ in Lemma~\ref{lemma:FI_Variance} as follows,
\begin{align}
    \FI\left(X+\frac{1}{h_\calB}\tilde{N}_\calB\Big|V\right)^{-1}-\frac{\tilde{\sigma}_\calB^2}{h_\calB^2}\ge\FI(X+N|V)^{-1}-\sigma_N^2,\nonumber
\end{align}for all $\sigma_N^2\le\frac{\tilde{\sigma}_\calB^2}{h_\calB^2}$, which can be rewritten as
\begin{align}
    \FI\left(X+N|V\right)\ge\left(\FI\left(X+\frac{1}{h_\calB}\tilde{N}_\calB\Big|V\right)^{-1}-\frac{\tilde{\sigma}_\calB^2}{h_\calB^2}+\sigma_N^2\right)^{-1}.\label{eq:Lower_Bounding_FI}
\end{align}Substituting \eqref{eq:Lower_Bounding_FI} in \eqref{eq:Conv_to_FI}, we have,
\begin{align}
    &\dent\left(\frac{1}{h_\calB}\tilde{Y}_\calB\Big|V\right)-\dent(X|V)\ge\nonumber\\
    &\qquad\frac{1}{2}\log\frac{\FI\left(X+\frac{1}{h_\calB}\tilde{N}_\calB\Big|V\right)^{-1}}{\FI\left(X+\frac{1}{h_\calB}\tilde{N}_\calB\Big|V\right)^{-1}-\frac{\tilde{\sigma}_\calB^2}{h_\calB^2}}.\label{eq:Lower_Bound_FI}
\end{align}Next, define, for $0\le t\le 1$,
\begin{align}
g(t)&\triangleq t\FI(X|V)^{-1}+\nonumber\\
&\qquad(1-t)\left[\FI\left(X+\frac{1}{h_\calB}\tilde{N}_\calB\Big|V\right)^{-1}-\frac{\tilde{\sigma}_\calB^2}{h_\calB^2}\right],\label{eq:gt}\\
    f(t)&\triangleq\frac{1}{2}\log\frac{g(t)+\frac{\tilde{\sigma}_\calB^2}{h_\calB^2}}{g(t)}.\label{eq:ft_Defi}
\end{align}
Hence, \eqref{eq:Upper_Bound_FI} and \eqref{eq:Lower_Bound_FI} can be expressed as,
\begin{align}
     f(0)\le\dent\left(\frac{1}{h_\calB}\tilde{Y}_\calB\Big|V\right)-\dent(X|V)\le f(1).\nonumber
\end{align}Since $f$ is continuous, from the intermediate value theorem, there exists a $t^\star\in[0,1]$ such that,
\begin{align}
    \dent\left(\frac{1}{h_\calB}\tilde{Y}_\calB\Big|V\right)-\dent(X|V)&=f(t^\star) =\frac{1}{2}\log\frac{g(t^\star)+\frac{\tilde{\sigma}_\calB^2}{h_\calB^2}}{g(t^\star)},\label{eq:ft_star_Defi}
\end{align}where $g(t^\star)$ is bounded as follows,
\begin{subequations}
\begin{align}
    \FI(X|V)^{-1}\mathop \le \limits^{(a)} g(t^\star)&\mathop \le \limits^{(b)}\FI\left( X+\frac{1}{h_\calB}\tilde{N}_\calB\Big|V\right)^{-1}-\frac{\tilde{\sigma}_\calB^2}{h_\calB^2}\label{eq:gStar_Bound_1}\\
    &\mathop \le \limits^{(c)}\FI\left(X+\frac{1}{h_\calA}\tilde{N}_\calA\Big|V\right)^{-1}-\frac{\tilde{\sigma}_\calA^2}{h_\calA^2},\label{eq:gStar_Bound_2}
\end{align}
where
\begin{itemize}
    \item[$(a)$] and $(b)$ follow by \eqref{eq:gt} and therefore by substituting $G_1\leftarrow \emptyset$ and $G_2\leftarrow \frac{1}{h_\calB}\tilde{N}_\calB$, in Lemma~\ref{lemma:FI_Variance}, we have,
    \begin{align}
        \FI\left(X+\frac{1}{h_\calB}\tilde{N}_\calB\Big|V\right)^{-1}-\frac{\tilde{\sigma}_\calB^2}{h_\calB^2}\ge\FI(X|V)^{-1};\label{eq:Btilde}
    \end{align}
    \item[$(c)$] follows by Lemma~\ref{lemma:FI_Variance} with $G_1\leftarrow \frac{1}{h_\calB}\tilde{N}_\calB$ and $G_2\leftarrow \frac{1}{h_\calA}\tilde{N}_\calA$.
\end{itemize}
\end{subequations}Hence, \eqref{eq:ft_star_Defi} implies that if we choose $\bar{V}$ to be a Gaussian random variable which satisfies $\sigma_{X|\bar{V}}^2=g(t^\star)$, then $C(\bar{V})=C(V)$. Now, we show that the Gaussian random variable $\bar{V}$ is feasible, i.e., $\sigma_{X|\bar{V},\tilde{Y}_\calA}^2\le D$. The following lemma connects the conditional covariance with Fisher information.
\begin{lemma}\label{lemma:Cond_Cov_FI}
Let $(V,X)$ be two arbitrary random variables with finite second moments, and $N$ be a zero-mean Gaussian random variable with variance $\sigma_N^2$. Let, $Y=X+N$ and assume that $V$ and $X$ are independent of $N$. Then we have,
\begin{align}
    \bbE_{Y,V}\big[\sigma_{X|Y,V}^2\big]=\sigma_N^2-\sigma_N^4\FI(X+N|V).\nonumber
\end{align}
\end{lemma}A vector version of this lemma is stated in \cite[Lemma~21]{EkremUlukus13_Lossy} without proof. For completeness, we prove Lemma~\ref{lemma:Cond_Cov_FI} in Appendix~\ref{proof:Cond_Cov_FI}. Using Lemma~\ref{lemma:Cond_Cov_FI}, we have,
\begin{align}
    \sigma_{X|\bar{V},\tilde{Y}_\calA}^2=\tilde{\sigma}_\calA^2-\tilde{\sigma}_\calA^4\FI(\tilde{Y}_\calA|\bar{V}).\label{eq:Cond_Cov_FI_Bound_2}
\end{align}
Then, we have,
\begin{align}
    \FI\left(\frac{1}{h_\calA}\tilde{Y}_\calA\Big|\bar{V}\right)&\mathop = \limits^{(a)}\sigma_{\frac{1}{h_\calA}\tilde{Y}_\calA|\bar{V}}^{-2}\nonumber\\
    &\mathop = \limits^{(b)}\left(\sigma_{X|\bar{V}}^{2}+\frac{1}{h_\calA^2}\tilde{\sigma}_\calA^2\right)^{-1}\nonumber\\
    &\mathop \ge \limits^{(c)} \left(\FI\left(X+\frac{1}{h_\calA}\tilde{N}_\calA\Big|V\right)^{-1}-\frac{\tilde{\sigma}_\calA^2}{h_\calA^2}+\frac{\tilde{\sigma}_\calA^2}{h_\calA^2}\right)^{-1}\nonumber\\
    &=\FI\left(\frac{1}{h_\calA}\tilde{Y}_\calA\Big|V\right),\label{eq:FI_Scaled_Y}
\end{align}where 
\begin{itemize}
    \item[$(a)$] follows from \cite[Lemma~2]{EkremUlukus13_MIMO}, which states that for Gaussian random variables $\FI(X|U)=\sigma_{X|U}^{-2}$; 
    \item[$(b)$] follows since $(\bar{V},X)$ and $\tilde{N}_\calA$ are independent;
    \item[$(c)$] follows from \eqref{eq:gStar_Bound_2} since $\sigma_{X|\bar{V}}^2=g(t^\star)$.
\end{itemize}
Next, using Lemma~\ref{lemma:Scaling_FI} below, \eqref{eq:FI_Scaled_Y} becomes
\begin{align}
    \FI\big(\tilde{Y}_\calA|\bar{V}\big)&\ge\FI\big(\tilde{Y}_\calA|V\big)\label{eq:FI_VB_V},
\end{align}
\begin{lemma}\label{lemma:Scaling_FI}
Let $X$ and $U$ be arbitrarily correlated random variables with well-defined densities. For any $a\in\bbR_{++}$,
\begin{align}
    \FI\left(aX|U\right)=\frac{1}{a^2}\FI(X|U).\nonumber
\end{align}
\end{lemma}
The proof of Lemma~\ref{lemma:Scaling_FI} is available in Appendix~\ref{proof:Scaling_FI}.

Therefore, combining \eqref{eq:Cond_Cov_FI_Bound_2} and \eqref{eq:FI_VB_V} results to
\begin{align}
    \sigma_{X|\bar{V},\tilde{Y}_\calA}^2&\le\tilde{\sigma}_\calA^2-\tilde{\sigma}_\calA^4\FI(\tilde{Y}_\calA|V)\nonumber\\
    &\mathop = \limits^{(a)}\sigma_{X|V,\tilde{Y}_\calA}^2\nonumber\\
    &\mathop\le\limits^{(b)} D,\label{eq:Cond_Cov_FI_Bound_3}
\end{align}where 
\begin{itemize}
    \item[$(a)$] follows from Lemma~\ref{lemma:Cond_Cov_FI};
    \item[$(b)$] follows since we assumed that $V$ is feasible, i.e., $\sigma_{X|V,\tilde{Y}_\calA}^2\le D$.
\end{itemize}\eqref{eq:Cond_Cov_FI_Bound_3} means that the constructed Gaussian random variable $\bar{V}$ is feasible, which means that for each feasible $V$, there exists a feasible Gaussian $\bar{V}$ such that $C(V)=C(\bar{V})$.

\section{Proof of Lemma~\ref{lemma:Cond_Cov_FI}}
\label{proof:Cond_Cov_FI}
We have,
\begin{align}
&\frac{1}{2}\FI(X+N|V=v)\mathop=\limits^{(a)}\frac{\partial}{\partial\sigma_N^2}\dent(X+N|V=v)\nonumber\\
&=\frac{\partial}{\partial\sigma_N^2}\big[\MI(X;X+N|V=v)+\dent(X+N|X,V=v)\big]\nonumber\\
&\mathop=\limits^{(b)}\frac{\partial}{\partial\sigma_N^2}\big[\MI(X;X+N|V=v)+\dent(N)\big]\nonumber\\
&=\frac{\partial}{\partial\sigma_N^2}\big[\MI(X;X+N|V=v)\big]+\frac{1}{2\sigma_N^2}\nonumber\\
&\mathop=\limits^{(c)}\frac{\partial}{\partial\sigma_N^2}\big[\MI\left(X;\sigma_N^{-1}X+N'|V=v\right)\big]+\frac{1}{2\sigma_N^2}\nonumber\\
&\mathop=\limits^{(d)}-\sigma_N^{-4}\frac{\partial}{\partial u}\big[\MI\left(X;\sqrt{u}X+N'|V=v\right)\big]+\frac{1}{2\sigma_N^2}\nonumber\\
&\mathop=\limits^{(e)}-\frac{1}{2}\sigma_N^{-4}\bbE_{Y|V=v}\left[\sigma_{X|Y,V=v}^2\right]+\frac{1}{2}\sigma_N^{-2},\label{eq:FI_Cov_Lemma}
\end{align}where
\begin{itemize}
    \item[$(a)$] follows from \cite[Lemma~3]{EkremUlukus13_MIMO};
    \item[$(b)$] follows since $V$ and $X$ are independent of $N$;
    \item[$(c)$] holds with $N'$ a standard  Gaussian random variable;
    \item[$(d)$] follows by defining $u\triangleq\sigma_N^{-2}$;
    \item[$(e)$] follows from \cite[Theorem~1]{MI_MMSE}, with the input distribution $P_{X|V=v}$.
\end{itemize}From \eqref{eq:FI_Cov_Lemma}, we have
\begin{align}
    \bbE_{Y|V=v}\left[\sigma_{X|Y,V=v}^2\right]=\sigma_N^2-\sigma_N^4\FI(X+N|V=v).\label{eq:FI_Cov_Lemma_2}
\end{align}Now let $F_{V,Y}$ be the cumulative distribution function of $(V,Y)$, therefore, from \eqref{eq:FI_Cov_Lemma_2} we have
\begin{align}
\bbE_{Y,V}[\sigma_{X|Y,V}^2]&=\int\limits_{v}\int\limits_{y}\sigma_{X|Y=y,V=v}^2dF_{V,Y}\nonumber\\
&= \int\limits_{v}\big(\sigma_N^2-\sigma_N^4J(X+N|V=v)\big)dF_{V}\nonumber\\
&=\sigma_N^2-\sigma_N^4\int\limits_{v} J(X+N|V=v)dF_V\nonumber\\
&=\sigma_N^2-\sigma_N^4J(X+N|V).\nonumber%
\end{align}

\section{Proof of Lemma~\ref{lemma:Scaling_FI}}
\label{proof:Scaling_FI}
Define $Y\triangleq aX$, then $f_{Y|U}(y|u)=\frac{1}{\abs{a}}f_{X|U}\left(\frac{y}{a}|u\right)$, and from Definition~\ref{defi:FI} we have
\begin{align}
    \FI(aX|U)&=\FI(Y|U)\nonumber\\
    &=\bbE\left[\left(\frac{\partial\log f_{Y|U}(y|u)}{\partial y}\right)^2\right]\nonumber\\
    &=\bbE\left[\left(\frac{\partial}{\partial y}\log\left( \frac{1}{\abs{a}}f_{X|U}\left(\frac{y}{a}\Big|u\right)\right)\right)^2\right]\nonumber\\
    &=\bbE\left[\left(\frac{\partial}{\partial y}\log\left( f_{X|U}\left(\frac{y}{a}\Big|u\right)\right)\right)^2\right]\nonumber\\
    &\mathop=\limits^{(a)}\bbE\left[\frac{1}{a^2}\left(\frac{\partial}{\partial z}\log\left( f_{X|U}(z|u)\right)\right)^2\right]\nonumber\\
    &=\frac{1}{a^2}\FI(X|U)\nonumber,
\end{align}where $(a)$ follows by defining $z\triangleq\frac{y}{a}$.

\section{Achievability Proof of Theorem~\ref{thm:Capacity}}
\label{sec:Achie_Proof}
We first establish an achievability region for discrete sources in Section~\ref{subsec:Discrete}. Then, we extend this achievability result to continuous sources in Section~\ref{subsec:Continuous}. Note that to obtain this extension, it is critical to precisely characterize the convergence of all the vanishing terms in the proof for discrete sources in Section~\ref{subsec:Discrete}. Finally, we establish the achievability region by Theorem~\ref{thm:Capacity} for Gaussian sources in Section~\ref{subsec:Continuous_Thm1}.
\subsection{Discrete Sources}
\label{subsec:Discrete}
In this section, we prove the achievability of the following region for discrete sources.
\begin{theorem}
\label{thm:discrete_achievable}
For any access structure $\bbA$, the following region is achievable for discrete sources,
\begin{align}
   &\bigcup\limits_{\substack{U-V-X-\tilde{Y}_\calL \\ \max\limits_{\calA\in\bbA}\bbE\big[d\big(X,\hat{X}_\calA(V,\mathbf{Y}_\calA)\big)\big]\le D}}\nonumber\\
   &\left.\begin{cases}(R,\Delta):\\
R>\max\limits_{\calA\in\bbA}\{\MI(V;X|\mathbf{Y}_\calA)\}\\
\Delta>\max\limits_{(\calA,\calB)\in(\bbA,\bbB)}\big\{\MI(V;X)-\MI(V;\mathbf{Y}_\calA|U)+\MI(X;\mathbf{Y}_\calB|U)\big\}\\
\end{cases}\hspace{-0.4cm}\right\}.\nonumber
\end{align}
\end{theorem}
\sloppy Fix $0<\epsilon<\epsilon_1<\epsilon_2$, and $(\calA,\calB)\in(\bbA,\bbB)$. Let $U$ and $V$ be two auxiliary random variables defined over the alphabets, $\calU$ and $\calV$, respectively, such that $P_{\mathbf{Y}_\calL XUV}=P_{\mathbf{Y}_\calL X}P_{V|X}P_{U|V}$.

\subsubsection{Random Codebook Generation}
\begin{itemize}
\item 
Let $C_U^{(n)}\triangleq\left(U^n(m_p,m'_p)\right)_{(m_p,m'_p)\in\calM_p\times\calM'_p}$, where $\calM_p\triangleq\blr{1}{2^{nR_p}}$ and $\calM'_p\triangleq\blr{1}{2^{nR'_p}}$, be a random codebook consisting of independent random sequences each generated according to $P_U^\on$. A realization of $C_U^{(n)}$ is denoted by $\calC_U^{(n)}\triangleq\left(u^n(m_p,m'_p)\right)_{(m_p,m'_p)\in\calM_p\times\calM'_p}$. The index $m_p$ is interpreted as a bin index $\calB_1(m_p)$ and $m'_p$ is the index of the codeword in the bin number $\calB_1(m_p)$.

\item
Fix $\calC_U^{(n)}$. For every $(m_p,m'_p)\in(\calM_p,\calM'_p)$, let $C_V^{(n)}(m_p,m'_p)\triangleq\big(V^n(m_p,m'_p,m_s,m'_s)\big)_{(m_s,m'_s)\in\calM_s\times\calM'_s}$, where $\calM_s\triangleq\blr{1}{2^{nR_s}}$ and $\calM'_s\triangleq\blr{1}{2^{nR'_s}}$, be a random codebook consisting of independent random sequences each generated according to $P_{V|U}^\on\big(\cdot|u^n(m_p,m'_p)\big)$. We denote a realization of $C_V^{(n)}(m_p,m'_p)$ by $\calC_V^{(n)}(m_p,m'_p)\triangleq\big(v^n(m_p,m'_p,m_s,m'_s)\big)_{(m_s,m'_s)\in\calM_s\times\calM'_s}$. The index $m_s$ is interpreted as a bin index $\calB_2(m_p,m'_p,m_s)$ and $m'_s$ is the index of the codeword in the bin number $\calB_2(m_p,m'_p,m_s)$.
\end{itemize}

\subsubsection{Encoding} Given $x^n$, the encoder finds the first indices $(m_p,m'_p)$ such that $\big(u^n(m_p,m'_p),x^n\big)\in\calT_{\epsilon_1}^{(n)}(U,X)$. Then, it finds the first indices $(m_s,m'_s)$ such that $\big(v^n(m_p,m'_p,m_s,m'_s),x^n\big)\in\calT_{\epsilon_1}^{(n)}\big(V,X|u^n(m_p,m'_p)\big)$. Finally, the encoder stores the index pair $m\triangleq(m_p,m_s)$ in the public database.

Note that the source $X$ is encoded into four messages $M_p$, $M'_p$, $M_s$, and $M'_s$, where $M_p$ and $M_s$ are publicly stored. $M'_p$ and $M'_s$ are not stored, but the authorized users in $\calA\in\bbA$ retrieve these indices, by the help of $\mathbf{Y}_\calA$, with small probability of error. 
The encoder wants the unauthorized sets of users to learn as little information as possible about the codeword $V^n(M_p,M'_p,M_s,M'_s)$, whereas the authorized sets of users must be able to reconstruct it.

\subsubsection{Decoding} Given $\mathbf{y}_\calA^n$ and $(m_p,m_s)$, the users in $\calA$ look for a unique pair $(\hat{m}'_p,\hat{m}'_s)$ such that $\big(u^n(m_p,\hat{m}'_p),v^n(m_p,\hat{m}'_p,m_s,\hat{m}'_s),x^n\big)\in\calT_{\epsilon_1}^{(n)}\big(U,V,X\big)$. Then, they compute the sequence $\hat{x}_i\triangleq \hat{x}_\calA\big(v_i(m_p,\hat{m}'_p,m_s,\hat{m}'_s),\mathbf{y}_{\calA,i}\big)$, for $i\in\intseq{1}{n}$, where the function $\hat{x}_\calA$ is defined such that $\bbE[d(X,\hat{X}_\calA)]\le D$.

\subsubsection{Error Analysis at the Authorized Users} 
Define the following events,
\begin{subequations}\label{eq:error_events}
\begin{align}
    \calE_1&\triangleq\big\{\big(X^n,\mathbf{Y}_\calL^n\big)\notin\calT_\epsilon^{(n)}(X,\mathbf{Y}_\calL)\big\},\label{eq:error_events_1}\\
    \calE_2&\triangleq\big\{\nexists (m_p,m'_p):\big(u^n(m_p,\hat{m}'_p),X^n\big)\in\calT_{\epsilon_1}^{(n)}(U,X)\big\},\label{eq:error_events_2}\\
    \calE_3&\triangleq\big\{\nexists (m_s,m'_s):\big(v^n(m_p,\hat{m}'_p,m_s,\hat{m}'_s),X^n\big)\in\nonumber\\
    &\qquad\calT_{\epsilon_1}^{(n)}\big(V,X|u^n(m_p,\hat{m}'_p)\big)\big\},\label{eq:error_events_3}\\
    \calE_4&\triangleq\big\{\exists \hat{m}'_p\ne m'_p,\hat{m}'_s:\big(u^n(m_p,\hat{m}'_p),v^n(m_p,\hat{m}'_p,m_s,\hat{m}'_s),\nonumber\\
    &\qquad\mathbf{y}_\calA^n\big)\in\big(\calB_1(m_p)\times\calB_2(m_p,m'_p,m_s)\times\boldsymbol{\calY}_\calA^n\big)\cap\nonumber\\
    &\qquad\calT_{\epsilon_2}^{(n)}\big(U,V,\mathbf{Y}_\calA\big)\big\},\label{eq:error_events_4}\\
    \calE_5&\triangleq\big\{\exists \hat{m}'_p,\hat{m}'_s\ne m'_s:\big(u^n(m_p,\hat{m}'_p),v^n(m_p,\hat{m}'_p,m_s,\hat{m}'_s),\nonumber\\
    &\qquad\mathbf{y}_\calA^n\big)\in\big(\calB_1(m_p)\times\calB_2(m_p,m'_p,m_s)\times\boldsymbol{\calY}_\calA^n\big)\cap\nonumber\\
    &\qquad\calT_{\epsilon_2}^{(n)}\big(U,V,\mathbf{Y}_\calA\big)\big\}.\label{eq:error_events_5}
\end{align}
\end{subequations}By the union bound, the probability of the event $\calE$ that an error occurs at the encoder or decoder satisfies 
\begin{align}
    \bbP\{\calE\}&\le\bbP\{\calE_1\}+\bbP\{\calE_1^c\cap\calE_2\}+\bbP\{\calE_1^c\cap\calE_2^c\cap\calE_3\}+\nonumber\\
    &\qquad\bbP\{\calE_4\}+\bbP\{\calE_5\cap\calE_4^c\}.\label{eq:Union_Bound}
\end{align}By \cite[Theorem~1.1]{Kramer_Book}, the first term on the \ac{RHS} of \eqref{eq:Union_Bound} vanishes as $n$ grows. The second term on the \ac{RHS} of \eqref{eq:Union_Bound}  is bounded as follows,
\begin{align}
    &\bbP\{\calE_1^c\cap\calE_2\}\nonumber\\
    &\le\Big(\bbP\left\{\big(U^n,X^n\big)\notin\calT_{\epsilon_1}^{(n)}(U,X)\big|\calE_1^c\right\}\Big)^{2^{n(R_p+R'_p)}}\nonumber\\
    &=\Big(1-\bbP\left\{\big(U^n,X^n\big)\in\calT_{\epsilon_1}^{(n)}(U,X)\big|\calE_1^c\right\}\Big)^{2^{n(R_p+R'_p)}}\nonumber\\
    &\le 2^{-2^{n(R_p+R'_p)}\bbP\left\{\big(U^n,X^n\big)\in\calT_{\epsilon_1}^{(n)}(U,X)\big|\calE_1^c\right\}}\nonumber\\
    &\le 2^{-\left(1-\eta_{\epsilon,\epsilon_1}^{(n)}\right)2^{n(R_p+R'_p)}2^{-n(\MI(U;X)+\epsilon_1\MI(U;X))}},\label{eq:error_2}
\end{align}where \eqref{eq:error_2} follows from \cite[Theorem~1.3]{Kramer_Book} with, $ \eta_{\epsilon,\epsilon_1}^{(n)}\triangleq2\abs{\calU}\abs{\calX}e^{-n\frac{(\epsilon_1-\epsilon)^2}{1+\epsilon}\mu_{U,X}}$ where $\mu_{U,X}\triangleq\min\limits_{(a,b)\in\text{supp}(P_{UX})}P_{UX}(a,b)$, the bound in \eqref{eq:error_2} vanishes when $n$ grows if
    \begin{align}
    R_p+R'_p>\MI(U;X)+\epsilon_1\MI(U;X).\label{eq:enc_error_2}
    \end{align}

The third term on the \ac{RHS} of \eqref{eq:Union_Bound} is bounded as follows,
\begin{align}
    &\bbP\{\calE_1^c\cap\calE_2^c\cap\calE_3\}\nonumber\\
    &\le\Big(\bbP\left\{\big(V^n,X^n\big)\notin\calT_{\epsilon_1}^{(n)}\big(V,X|u^n(m_p,m'_p)\big)\big|\calE_1^c\cap\calE_2^c\right\}\Big)^{2^{n(R_s+R'_s)}}\nonumber\\
    &\le 2^{-\left(1-\delta_{\epsilon,\epsilon_1}^{(n)}\right)2^{n\left(R_s+R'_s\right)}2^{-n\left(\MI(V;X|U)+\epsilon_1\MI(V;X|U)\right)}},\label{eq:error_3}
\end{align}where \eqref{eq:error_3} holds similar to \eqref{eq:error_2} with $\delta_{\epsilon,\epsilon_1}^{(n)}\triangleq 2\abs{\calU}\abs{\calV}\abs{\calX}e^{-n\frac{(\epsilon_1-\epsilon)^2}{1+\epsilon}\mu_{U,V,X}}$ and $\mu_{U,V,X}\triangleq\min\limits_{(a,b,c)\in\text{supp}(P_{UVX})}P_{UVX}(a,b,c)$, and the bound in \eqref{eq:error_3} vanishes when $n$ grows if,
\begin{align}
    R_s+R'_s>\MI(V;X|U)+\epsilon_1\MI(V;X|U).\label{eq:enc_error_3}
\end{align}
Next, the fourth term on the \ac{RHS} of \eqref{eq:Union_Bound} is bounded as,
\begin{align}
    &\bbP\{\calE_4\}\nonumber\\
    &=\bbP\big\{\exists \hat{m}'_p\ne m'_p,m'_s:\big(u^n(m_p,\hat{m}'_p),v^n(m_p,\hat{m}'_p,m_s,m'_s),\mathbf{y}_\calA^n\big)\nonumber\\
    &\qquad\in\big(\calB_1(m_p)\times\calB_2(m_p,\hat{m}'_p,m_s)\times\boldsymbol{\calY}_\calA^n\big)\cap\calT_{\epsilon_2}^{(n)}\big(U,V,\mathbf{Y}_\calA\big)\big\},\nonumber\\
    &\le 2^{n(R'_p+R'_s)}\bbP\big\{\big(U^n,V^n,\mathbf{Y}_\calA^n\big)\in\calT_{\epsilon_2}^{(n)}\big(U,V,\mathbf{Y}_\calA\big)\big|\big(U^n,V^n\big)\in\calT_{\epsilon_1}^{(n)}\big\}\nonumber\\
    &\le2^{n(R'_p+R'_s)}2^{-n[\MI(V;\mathbf{Y}_\calA)-{\epsilon_2}\MI(V;\mathbf{Y}_\calA)]},\label{eq:error_4}
\end{align}where \eqref{eq:error_4} follows from \cite[Theorem~1.3]{Kramer_Book}, and the \ac{RHS} of \eqref{eq:error_4} vanishes when $n$ grows if,
\begin{align}
    R'_p+R'_s<\MI(V;\mathbf{Y}_\calA)-{\epsilon_2}\MI(V;\mathbf{Y}_\calA).\nonumber
\end{align}Considering that the probability of error must vanish for all $\calA\in\bbA$, we need,
\begin{align}
    R'_p+R'_s<\min\limits_{\calA\in\bbA}\{\MI(V;\mathbf{Y}_\calA)-{\epsilon_2}\MI(V;\mathbf{Y}_\calA)\}.\label{eq:enc_error_4_C}
\end{align}

Finally, similar to \eqref{eq:error_4} ,the last term on the \ac{RHS} of \eqref{eq:Union_Bound} is bounded as,
\begin{align}
    \bbP\{\calE_5\cap\calE_4^c\}
    &\le2^{nR'_s}2^{-n[\MI(V;\mathbf{Y}_\calA|U)-{\epsilon_2}\MI(V;\mathbf{Y}_\calA|U)]},\label{eq:error_5}
\end{align}and the \ac{RHS} of \eqref{eq:error_5} vanishes when $n$ grows if,
\begin{align}
    R'_s<\MI(V;\mathbf{Y}_\calA|U)-\epsilon_2\MI(V;\mathbf{Y}_\calA|U).\nonumber
\end{align}Considering that the probability of error must vanish for all $\calA\in\bbA$, we need,
\begin{align}
    R'_s<\min\limits_{\calA\in\bbA}\{\MI(V;\mathbf{Y}_\calA|U)-\epsilon_2\MI(V;\mathbf{Y}_\calA|U)\}.\label{eq:enc_error_5_C}
\end{align}

\subsubsection{Analysis of Expected Distortion at the Authorized Users} We have
\begin{align}
    &\bbE\left[d\left(X^n,\hat{X}\big(f(X^n),\mathbf{Y}_\calA^n\big)\right)\right]\nonumber\\
    &\le d_{\max}\bbP\{\calE\}+(1-\bbP\{\calE\})\bbE\left[d\left(X^n,\hat{X}_\calA\big(f(X^n),\mathbf{Y}_\calA^n\big)\right)\Big|\calE^c\right]\nonumber\\
    &\mathop\le\limits^{(a)}\bbE\left[d\left(X,\hat{X}_\calA\big(V,\mathbf{Y}_\calA\big)\right)\right]+d_{\max}\bbP\{\calE\}+d_{\max}\abs{\calV}\abs{\calX}\abs{\boldsymbol{\calY}_\calA}\epsilon_2.\label{eq:Distortion_Analysis}%
\end{align}where $(a)$ follows since for every $(v^n,x^n,\mathbf{y}_\calA^n)\in\calT_{\epsilon_2}^{(n)}(V,X,\mathbf{Y}_\calA)$,
\begin{align}
    &d\left(x^n,\hat{X}_\calA\big(f(x^n),\mathbf{y}_\calA^n\big)\right)=\frac{1}{n}\sum\limits_{i=1}^nd\left(x_i,\hat{X}_\calA\big(v_i,\mathbf{y}_{\calA,i}\big)\right)\nonumber\\
    &\mathop=\limits^{(a)}\frac{1}{n}\sum\limits_{(v,x,\mathbf{y}_\calA)\in(\calV\times\calX\times\boldsymbol{\calY}_\calA)}d\left(x,\hat{X}_\calA\big(v,\mathbf{y}_\calA\big)\right)\times\nonumber\\
    &\quad N\left(v,x,\mathbf{y}_\calA\big|v^n,x^n,\mathbf{y}_\calA^n\right)\nonumber\\
    &=\bbE\left[d\left(X,\hat{X}_\calA\big(V,\mathbf{Y}_\calA\big)\right)\right]+\nonumber\\
    &\sum\limits_{(v,x,\mathbf{y}_\calA)\in(\calV\times\calX\times\boldsymbol{\calY}_\calA)}d\left(x,\hat{X}_\calA\big(v,\mathbf{y}_\calA\big)\right)\times\nonumber\\
    &\left(\frac{1}{n}N\left(v,x,\mathbf{y}_\calA\big|v^n,x^n,\mathbf{y}_\calA^n\right)-P(v,x,\mathbf{y}_\calA)\right)\nonumber\\
    &\mathop\le\limits^{(b)}\bbE\left[d\left(X,\hat{X}_\calA\big(V,\mathbf{Y}_\calA\big)\right)\right]+d_{\max}\abs{\calV}\abs{\calX}\abs{\boldsymbol{\calY}_\calA}\epsilon_2,\label{eq:disto_analysis}
\end{align}where $(a)$ follows by defining $N(x|x^n)\triangleq\sum_{i=1}^n\indic{1}_{\{x_i=x\}}$; and $(b)$ follows since $(v^n,x^n,\mathbf{y}_\calA^n)\in\calT_{\epsilon_2}^{(n)}(V,X,\mathbf{Y}_\calA)$.

\subsubsection{Information Leakage Analysis at Unauthorized Users}
We will need the following lemma.
\begin{lemma}
\label{lemma:Entropy_Upper_Bound}
For the coding scheme described above, the following inequality holds, 
\begin{align}
    \ent\big(X^n|M_p,M'_p,\mathbf{Y}_\calB^n\big)\ge n\big[\ent\big(X|U,\mathbf{Y}_\calB\big)-\eta_n\big],\nonumber%
\end{align}where 
    \begin{align}
        \eta_n&\triangleq\epsilon_2\MI(U;X)-\log\left(1-\delta_{\epsilon_1,\epsilon_2}^{(n)}\right)+\nonumber\\
        &\quad\frac{-\bbP\{\calE_1^c\cap\calE_2\}\log \bbP\{\calE_1^c\cap\calE_2\}}{n}+\nonumber\\
        &\quad\frac{1}{n(1-t_n)}+\epsilon_2\max\limits_{\calB\in\bbB}\ent(\mathbf{Y}_\calB|U)+\frac{1}{n}+\nonumber\\
        &\quad\bbP\{\calE_1\}\bbP\{\calE_1^c\cap\calE_2\}\log\max\limits_{\calB\in\bbB}\abs{\boldsymbol{\calY}_\calB},\label{eq:eta_defi}
    \end{align} where $\delta_{\epsilon_1,\epsilon_2}^{(n)}$ and $t_n$ are bounded in \eqref{eq:Kramer_Lemma_1} and $\bbP\{\calE_1\}$ and $\bbP\{\calE_1^c\cap\calE_2\}$ are defined in \eqref{eq:Union_Bound} and \eqref{eq:error_2},~respectively. 
\end{lemma}The proof of Lemma~\ref{lemma:Entropy_Upper_Bound} is similar to that of \cite[Lemma~6]{VillardPianta} and is available in Appendix~\ref{proof:Entropy_Upper_Bound}, in this paper, the precise characterization of $\eta_n$ is critical to extend our result to continuous~sources.

We now bound the information leakage at the unauthorized set of users $\calB\in\bbB$ as,
\begin{align}
&\frac{1}{n}\MI(X^n;M,\mathbf{Y}_\calB^n)\nonumber\\
&= \frac{1}{n}\big[\ent(X^n)-\ent\big(X^n|M,\mathbf{Y}_\calB^n\big)\big]\nonumber\\
&\mathop=\limits^{(a)}\ent(X)-\frac{1}{n}\ent\big(X^n|M_p,M_s,\mathbf{Y}_\calB^n\big)\nonumber\\
&=\ent(X)+\frac{1}{n}\big[\MI(X^n;M_s|M_p,\mathbf{Y}_\calB^n)-\ent\big(X^n|M_p,\mathbf{Y}_\calB^n\big)\big]\nonumber\\
&\mathop\le\limits^{(b)}\ent(X)+\frac{1}{n}\big[\ent(M_s)-\ent\big(X^n|M_p,M'_p,\mathbf{Y}_\calB^n\big)\big]\nonumber\\
&\mathop\le\limits^{(c)}\ent(X)+R_s-\ent\big(X|U,\mathbf{Y}_\calB\big)+\eta_n\nonumber\\
&=\MI\big(X;U,\mathbf{Y}_\calB\big)+R_s+\eta_n,\label{eq:IL_Ana_1}
\end{align}where
\begin{itemize}
    \item[$(a)$] follows since $X^n$ is \ac{iid}, and by the definition of $M$;
    \item[$(b)$] follows by bounding the conditional mutual information with the unconditional entropy and because conditioning does not increase the entropy;
    \item[$(c)$] follows from Lemma~\ref{lemma:Entropy_Upper_Bound}.
\end{itemize}
Since the leakage condition must hold for any $\calB\in\bbB$, the bound in \eqref{eq:IL_Ana_1} means that we need
\begin{align}
    \Delta-R_s>\max\limits_{\calB\in\bbB}\{\MI\big(X;U,\mathbf{Y}_\calB\big)\}+\eta_n.\label{eq:enc_error_6_C}
\end{align}Applying Fourier-Motzkin to \eqref{eq:enc_error_2}, \eqref{eq:enc_error_3} \eqref{eq:enc_error_4_C}, \eqref{eq:enc_error_5_C}, and  \eqref{eq:enc_error_6_C} results to the following inequalities,
\begin{subequations}
\begin{align}
R&>\MI(V;X|U)-(1-\epsilon_2)\min\limits_{\calA\in\bbA}\{\MI(V;\mathbf{Y}_\calA|U)\}+\nonumber\\
&\quad\epsilon_1\MI(V;X|U)\label{eq:Achievable_R_1},\\
R&>\MI(V;X)-(1-\epsilon_2)\min\limits_{\calA\in\bbA}\{\MI(V;\mathbf{Y}_\calA)\}+\epsilon_1\MI(U;X)+\nonumber\\
&\quad\epsilon_1\MI(V;X|U)\label{eq:Achievable_R_2},\\
\Delta&>\MI(V;X|U)-(1-\epsilon_2)\min\limits_{\calA\in\bbA}\{\MI(V;\mathbf{Y}_\calA)\}+\nonumber\\
&\quad\max\limits_{\calB\in\bbB}\{\MI(X;U,\mathbf{Y}_\calB)\}+\epsilon_1\MI(V;X|U)+\eta_n\label{eq:Achievable_D_3},\\
\Delta&>\MI(V;X|U)-(1-\epsilon_2)\min\limits_{\calA\in\bbA}\{\MI(V;\mathbf{Y}_\calA|U)\}+\nonumber\\
&\quad\max\limits_{\calB\in\bbB}\{\MI(X;U,\mathbf{Y}_\calB)\}+\epsilon_1\MI(V;X|U)+\eta_n\label{eq:Achievable_D_4}.
\end{align}
\end{subequations}Next, we show that \eqref{eq:Achievable_R_1} and \eqref{eq:Achievable_D_3} are redundant. \eqref{eq:Achievable_D_3} is redundant because for $\calA^\prime\in \argmin\limits_{\calA\in\bbA}\MI(V;\mathbf{Y}_\calA)$, we have
\begin{align}
\min\limits_{\calA\in\bbA}\MI(V;\mathbf{Y}_\calA)&=\MI(V;\mathbf{Y}_{\calA^\prime})\nonumber\\
&\mathop=\limits^{(a)}\MI(U,V;\mathbf{Y}_{\calA^\prime})\nonumber\\
&\geq\MI(V;\mathbf{Y}_{\calA^\prime}|U)\nonumber\\
&\geq \min\limits_{\calA\in\bbA}\{\MI(V;\mathbf{Y}_{\calA}|U)\},\nonumber%
\end{align}where $(a)$ follows since $U-V-\mathbf{Y}_\calA$, forms a Markov chain.  
\eqref{eq:Achievable_R_1} is also redundant, because we can rewrite \eqref{eq:Achievable_R_2} as,
\begin{align}
R&>\MI(V;X|U)+\MI(U;X)-(1-\epsilon_2)\min\limits_{\calA\in\bbA}\{\MI(U,V;\mathbf{Y}_\calA)\}+\nonumber\\
&\quad\epsilon_1\MI(U;X)+\epsilon_1\MI(V;X|U)\nonumber\\
&=\MI(V;X|U)-(1-\epsilon_2)\min\limits_{\calA\in\bbA}\{\MI(V;\mathbf{Y}_\calA|U)+\MI(U;\mathbf{Y}_\calA)-\nonumber\\
&\quad\MI(U;X)\}+(\epsilon_1+\epsilon_2)\MI(U;X)+\epsilon_1\MI(V;X|U),\nonumber%
\end{align}and since, by the Markov chain $U-X-\mathbf{Y}_\calA$, $\MI(U;\mathbf{Y}_\calA)-\MI(U;X)\le 0$, \eqref{eq:Achievable_R_2} implies \eqref{eq:Achievable_R_1}. Then, using the Markov chain $U-V-X-\mathbf{Y}_\calL$, the constraints \eqref{eq:Achievable_R_2} and \eqref{eq:Achievable_D_4} can be written~as,
\begin{subequations}\label{eq:Simplified_Region}
\begin{align}
R&>\max\limits_{\calA\in\bbA}\{\MI(V;X|\mathbf{Y}_\calA)\}+\max\limits_{\calA\in\bbA}\delta_{\epsilon,1}(\calA)+\epsilon',\label{eq:Simplified_R}\\
\Delta&>\max\limits_{(\calA,\calB)\in(\bbA,\bbB)}\big\{\MI(V;X)-\MI(V;\mathbf{Y}_\calA|U)+\MI(X;\mathbf{Y}_\calB|U)\big\}+\nonumber\\
&\quad\max\limits_{(\calA,\calB)\in(\bbA,\bbB)}\delta_{\epsilon,2}(\calA,\calB)+\epsilon',\label{eq:Simplified_D}
\end{align}
\end{subequations}where $\delta_{\epsilon,1}(\calA)\triangleq\epsilon_2\MI(V;\mathbf{Y}_\calA)$, $\epsilon'\triangleq\epsilon_1\left(\MI(U;X)+\MI(V;X|U)\right)$, $\delta_{\epsilon,2}(\calA,\calB)\triangleq\epsilon_2\MI(V;\mathbf{Y}_\calA|U)+\eta_n$. 
\subsection{Continuous Sources}
\label{subsec:Continuous}
We now prove that Section~\ref{subsec:Discrete} extends to continuous sources by quantizing the source observations. The main challenge of the quantization is that it can result in underestimating the information that the unauthorized sets of users can learn about the source. However, we will show that we can overcome this issue by making the quantization fine enough. We will use the following lemma.
\begin{lemma}[\hspace{-0.08mm}\cite{Cover_Book,Pinsker64,Fano61}]
\label{lemma:quantization} Let $X$ and $Y$ be two real-valued random variables with distributions $P_X$ and $P_Y$, respectively. Let $\calC_{\Phi_1}=\{C_i\}_{i\in\calI}$ and $\calK_{\Phi_2}=\{K_j\}_{j\in\calJ}$ be two partitions of the real line for $X$ and $Y$, respectively, such that for any $i\in\calI$, $P_X[C_i]=\Phi_1$ and for any $j\in\calJ$, $P_Y[K_j]=\Phi_2$, where $\Phi_1>0$ and $\Phi_2>0$. We denote the quantized versions of $X$ and $Y$ with  respect to the partitions $\calC_{\Phi_1}$ and $\calK_{\Phi_2}$ by $X_{\Phi_1}$ and $Y_{\Phi_2}$, respectively. Then, we have
\begin{align}
    \MI(X;Y)=\lim\limits_{\Phi_1,\Phi_2\to 0}\MI(X_{\Phi_1};Y_{\Phi_2}).\nonumber%
\end{align}
\end{lemma}Note that a quantization of $\mathbf{Y}_\calB^n$, $\calB\in\bbB$, can lead to the underestimation of $\MI(X^n;M,\mathbf{Y}_\calB^n)$. The following lemma shows that the leakage constraint in Definition~\ref{eq:Defi_Achi_Rate} is not affected, provided that the quantization is sufficiently fine.
\begin{lemma}
\label{lemma:quantization_2}
If the quantization $X_{\Phi_1}^n$ of $X^n$ and $\mathbf{Y}_{\calB,\Phi_2}^n$ of $\mathbf{Y}_\calB^n$, for every $\calB\in\bbB$, are fine enough, then for every $\epsilon>0$,
\begin{align}
    \max\limits_{\calB\in\bbB}\MI(X^n;M,\mathbf{Y}_\calB^n)&\le\max\limits_{\calB\in\bbB}\MI(X_{\Phi_1}^n;M,\mathbf{Y}_{\calB,\Phi_2}^n)+\epsilon.\nonumber%
\end{align}
\end{lemma}
\begin{proof}
For any $\epsilon>0$, and for any $\calB\in\bbB$, we have
\begin{align}
    &\MI(X^n;M,\mathbf{Y}_\calB^n)\nonumber\\
&\le\abs{\MI(X^n;M,\mathbf{Y}_\calB^n)-\MI(X_{\Phi_1}^n;M,\mathbf{Y}_{\calB,\Phi_2}^n)}+\nonumber\\
&\quad\MI(X_{\Phi_1}^n;M,\mathbf{Y}_{\calB,\Phi_2}^n)\nonumber\\
    &\le\max\limits_{\calB\in\bbB}\abs{\MI(X^n;M,\mathbf{Y}_\calB^n)-\MI(X_{\Phi_1}^n;M,\mathbf{Y}_{\calB,\Phi_2}^n)}+\nonumber\\
&\quad\max\limits_{\calB\in\bbB}\MI(X_{\Phi_1}^n;M,\mathbf{Y}_{\calB,\Phi_2}^n)\nonumber\\
    &\mathop\le\limits^{(a)}\epsilon+\max\limits_{\calB\in\bbB}\MI(X_{\Phi_1}^n;M,\mathbf{Y}_{\calB,\Phi_2}^n),\label{eq:fine_quantization}
\end{align}where $(a)$ follows from Lemma~\ref{lemma:quantization} when the quantization  $\mathbf{Y}_{\calB,\Phi_2}^n$ is fine enough, for any $\calB\in\bbB$. Note that \eqref{eq:fine_quantization} is valid for any $\calB\in\bbB$, therefore \eqref{eq:fine_quantization} results to the bound in Lemma~\ref{lemma:quantization_2}.
\end{proof}
For every $(\calA,\calB)\in(\bbA,\bbB)$, we quantize $U$, $V$, $X$, $\mathbf{Y}_\calA$, and $\mathbf{Y}_\calB$ as in Lemma~\ref{lemma:quantization} to form $U_\Phi$, $V_\Phi$, $X_\Phi$, $\mathbf{Y}_{\calA,\Phi}$, and $\mathbf{Y}_{\calB,\Phi}$ such that $\Phi=\ell^{-1}$ and $\abs{\calU_\Phi}=\abs{\calV_\Phi}=\abs{\calX_\Phi}=\abs{\boldsymbol{\calY}_{\calA,\Phi}}=\abs{\boldsymbol{\calY}_{\calB,\Phi}}=\ell$ with $\ell>0$. Then, we apply the proof for the discrete case to the random variables $U_\Phi$, $V_\Phi$, $X_\Phi$, $\mathbf{Y}_{\calA,\Phi}$, and $\mathbf{Y}_{\calB,\Phi}$. For fixed $\epsilon$, $\epsilon_1$, and $\epsilon_2$, using Lemma~\ref{lemma:quantization}, we fix $\ell$ to a number large enough such that, for every $(\calA,\calB)\in(\bbA,\bbB)$, $\abs{\MI(V_\Phi;\mathbf{Y}_{\calA,\Phi}|U_\Phi)-\MI(V;\mathbf{Y}_\calA|U)}\le\frac{\delta}{4}$, $\abs{\MI(V_\Phi;X_\Phi|\mathbf{Y}_{\calA,\Phi})-\MI(V;X|\mathbf{Y}_\calA)}\le\frac{\delta}{4}$,  $\abs{\MI(X_\Phi;\mathbf{Y}_{\calB,\Phi}|U_\Phi)-\MI(X;\mathbf{Y}_\calB|U)}\le\frac{\delta}{4}$, $\abs{\MI(U_\Phi;X_\Phi)-\MI(U;X)}\le\frac{\delta}{4}$, $\abs{\MI(V_\Phi;X_\Phi)-\MI(V;X)}\le\frac{\delta}{4}$, and $\abs{\MI(V_\Phi;X_\Phi|U_\Phi)-\MI(V;X|U)}\le\frac{\delta}{4}$. Therefore, when $n$ grows to $\infty$, \eqref{eq:Simplified_Region} and \eqref{eq:Distortion_Analysis} reduces~to,
\begin{align}
    &R>\max\limits_{\calA\in\bbA}\{\MI(V;X|\mathbf{Y}_\calA)\}+\max\limits_{\calA\in\bbA}\delta_{\epsilon,1}(\calA)+\epsilon'-\nonumber\\
&\quad(1+\epsilon_2)\frac{\delta}{4}-\epsilon_1\frac{3\delta}{4},\nonumber\\
&\Delta>\max\limits_{(\calA,\calB)\in(\bbA,\bbB)}\big\{\MI(V;X)-\MI(V;\mathbf{Y}_\calA|U)+\MI(X;\mathbf{Y}_\calB|U)\big\}+\nonumber\\
&\quad\max\limits_{(\calA,\calB)\in(\bbA,\bbB)}\delta_{\epsilon,2}(\calA,\calB)+\epsilon''-\frac{3\delta}{4}-(\epsilon_1+\epsilon_1)\frac{\delta}{4},\nonumber\\
  &\bbE\left[d\left(X^n,\hat{X}\big(f(X^n),\mathbf{Y}_\calA^n\big)\right)\right]\le\bbE\left[d\left(X,\hat{X}\big(V,\mathbf{Y}_\calA\big)\right)\right]+\nonumber\\
&\quad d_{\max}\bbP\{\calE\}+d_{\max}\abs{\calV_\Phi}\abs{\calX_\Phi}\abs{\boldsymbol{\calY}_{\calA,\Phi}}\epsilon_2.\nonumber
\end{align}
Note that, the term $\eta_n$ in $\delta_{\epsilon,2}(\calA,\calB)$, defined in \eqref{eq:eta_defi}, includes the entropy term $\epsilon_2\ent(\mathbf{Y}_{\calB,\Phi}|U_\Phi)$ which diverges when $\ell$ tends to infinity. The term $d_{\max}\abs{\calV_\Phi}\abs{\calX_\Phi}\abs{\boldsymbol{\calY}_{\calA,\Phi}}\epsilon_2$ also diverges when $\ell$ grows to $\infty$. Therefore, we choose $\epsilon_2\triangleq n^{-a}$, $a<\frac{1}{2}$ such that, when $n$ is large enough, the asymptotic source coding rate and the asymptotic leakage rate are as close as desired to the following rates,
\begin{subequations}\label{eq:Simplified_RD}
\begin{align}
    R&>\max\limits_{\calA\in\bbA}\big\{\MI(V;X|\mathbf{Y}_\calA)\big\},\label{eq:Simplified_D_31}\\
    \Delta&>\max\limits_{(\calA,\calB)\in(\bbA,\bbB)}\big\{\MI(V;X)-\MI(V;\mathbf{Y}_\calA|U)+\MI(X;\mathbf{Y}_\calB|U)\big\}.\label{eq:Simplified_D_3}
\end{align}
\end{subequations}
\subsection{Gaussian Sources}
\label{subsec:Continuous_Thm1}
For Gaussian sources, we use the mean square error to measure the distortion of the reconstructed sequence and \ac{MMSE} as the estimator, therefore the distortion constraint in Theorem~\ref{thm:discrete_achievable} reduces to
$\sigma_{X|\mathbf{Y}_\calA,V}^2\le D.$
\subsubsection{Mutual information Computation via Sufficient Statistics}
\label{subsubsec:MI_Calc}
In the following, we consider a random variable $V$ jointly Gaussian with $(X,\mathbf{Y}_\calL)$. 
Since the side information $\mathbf{Y}_\calA$ and $\mathbf{Y}_\calB$ are vectors and we aim to find the relationship between $\sigma_{X|V}^2$ and $\sigma_{X|\mathbf{Y}_\calA,V}^2$, it is easier to work with scalar random variables and use the sufficient statistics $\tilde{Y}_\calA$ and $\tilde{Y}_\calB$ defined in \eqref{eq:Suff_Stats_Inp_Outp_B}, respectively, to evaluate the mutual information expressions in the achievable rate region provided in \eqref{eq:Simplified_RD}.  
Hence,
\begin{align}
    \MI(V;X|\mathbf{Y}_\calA)&\mathop=\limits^{(a)}\MI(V;X|\mathbf{Y}_\calA,\tilde{Y}_\calA)\nonumber\\
    &\mathop=\limits^{(b)}\MI(V;X|\tilde{Y}_\calA),\nonumber%
\end{align}where $(a)$ follows from the Markov chain $X-\mathbf{Y}_\calA-\tilde{Y}_\calA$ and $(b)$ follows from the Markov chain $X-\tilde{Y}_\calA-\mathbf{Y}_\calA$. Similarly, one can also show that
\begin{align}
    \MI(V;\mathbf{Y}_\calA)&=\MI(V;\tilde{Y}_\calA),\nonumber\\
    \MI(X;\mathbf{Y}_\calB|U)&=\MI(X;\tilde{Y}_\calB),\nonumber\\
    \MI(V;X|\mathbf{Y}_\calA)&=\MI(V;X|\tilde{Y}_\calA),\nonumber\\
    \MI(X;\mathbf{Y}_\calB|V)&=\MI(X;\tilde{Y}_\calB|V),\nonumber\\
    \sigma_{X|\mathbf{Y}_\calA,V}^2&=\sigma_{X|\tilde{Y}_\calA,V}^2.\nonumber
\end{align}
Note that, for any $\calS\subseteq\calL$, $(V,X,\tilde{Y}_\calS)$ are jointly Gaussian because $\big[\begin{array}{ccc}
    V & X & \tilde{Y}_\calS
   \end{array}\big]=\big[\begin{array}{ccc}
    V & X & \mathbf{Y}_\calS
   \end{array}\big]\begin{bmatrix}
    I_2 & 0\\
    0 & h_\calS^\intercal\boldsymbol{\Sigma}_\calS^{-1}
\end{bmatrix}$, where $I_2$ is the two by two identity matrix.

Next, we consider two cases. First, the case $\tr\big(\boldsymbol{\Sigma}_{\calA^\star}^{-1}\big) \le\tr\big(\boldsymbol{\Sigma}_{\calB^\star}^{-1}\big)$, where $\calA^\star\in \argmin\limits_{\calA\in\bbA}\{\tr\left(\boldsymbol{\Sigma}_\calA^{-1}\right)\}$, $\calB^\star\in \argmax\limits_{\calB\in\bbB}\{\tr\left(\boldsymbol{\Sigma}_\calB^{-1}\right)\}$, which represents a situation in which authorized users have a better side information. Second, the case $\tr\big(\boldsymbol{\Sigma}_{\calA^\star}^{-1}\big) >\tr\big(\boldsymbol{\Sigma}_{\calB^\star}^{-1}\big)$, which represents a situation in which unauthorized users have a better side information. In the following, we assume $\tr\big(\boldsymbol{\Sigma}_{\calA^\star}^{-1}\big)\le\frac{1}{D}-\frac{1}{\sigma_X^2}$, since when this is not the case, $D\ge\sigma_{X|Y}^2$.
\subsubsection{When Authorized Users have a better side information}Suppose that $\tr\big(\boldsymbol{\Sigma}_{\calA^\star}^{-1}\big) \le\tr\big(\boldsymbol{\Sigma}_{\calB^\star}^{-1}\big)$. In this case, we choose the auxiliary random variable $V$ to be jointly Gaussian with $(X,\mathbf{Y}_\calL)$ and $U=\emptyset$. Then by Section~\ref{subsubsec:MI_Calc},
\begin{subequations}\label{eq:MI_Calc_Achi}
\begin{align}
    \MI(V;X)&=\dent(X)-\dent(X|V)\nonumber\\
&=\frac{1}{2}\log\frac{\sigma_X^2}{\sigma_{X|V}^2},\label{eq:IVX}\\
    \MI(V;\tilde{Y}_\calA)&=\dent(\tilde{Y}_\calA)-\dent(\tilde{Y}_\calA|V)\nonumber\\
&=\frac{1}{2}\log\frac{h_\calA^2\sigma_X^2+\tilde{\sigma}_\calA^2}{h_\calA^2\sigma_{X|V}^2+\tilde{\sigma}_\calA^2},\label{eq:IVYA}\\
    \MI(X;\mathbf{Y}_\calB)&=\MI(X;\tilde{Y}_\calB)=\dent(\tilde{Y}_\calB)-\dent(\tilde{Y}_\calB|X)\nonumber\\
&=\frac{1}{2}\log\frac{h_\calB^2\sigma_X^2+\tilde{\sigma}_\calB^2}{\tilde{\sigma}_\calB^2},\label{eq:IXYB}\\
    \MI(V;X|\mathbf{Y}_\calA)&=\MI(V;X|\tilde{Y}_\calA)\nonumber\\
&\quad\mathop=\limits^{(a)}\MI(V;X)-\MI(V;\tilde{Y}_\calA)\nonumber\\
&\quad\mathop=\limits^{(b)}\frac{1}{2}\log\frac{\sigma_X^2\big(h_\calA^2\sigma_{X|V}^2+\tilde{\sigma}_\calA^2\big)}{\sigma_{X|V}^2\big(h_\calA^2\sigma_X^2+\tilde{\sigma}_\calA^2\big)},\label{eq:IVXgYA_2}
\end{align}where 
\begin{itemize}
    \item[$(a)$] follows since $V-X-\tilde{Y}_\calA$ forms a Markov chain;
    \item[$(b)$] follows from \eqref{eq:IVX} and \eqref{eq:IVYA}.
\end{itemize}
\end{subequations}
Using Section~\ref{subsubsec:MI_Calc}, \eqref{eq:MI_Calc_Achi}, \eqref{eq:Cov_XgVY_1}, $h_\calA=\tilde{\sigma}_\calA^2=\tr\left(\boldsymbol{\Sigma}_\calA^{-1}\right)$, and $h_\calB=\tilde{\sigma}_\calB^2=\tr\left(\boldsymbol{\Sigma}_\calB^{-1}\right)$, as showed in~\eqref{eq:Suff_Stats_Inp_Outp}, we rewrite the achievable rate region in Theorem~\ref{thm:discrete_achievable} as 
\begin{align}
   &\bigcup\limits_{\substack{V-X-\mathbf{Y}_\calL \\ \max\limits_{\calA\in\bbA}\big\{\frac{\sigma_{X|V}^2}{\tr\left(\boldsymbol{\Sigma}_\calA^{-1}\right)\sigma_{X|V}^2+1}\big\}\le D}}\nonumber\\
&\left.\begin{cases}(R,\Delta):\\
R>\max\limits_{\calA\in\bbA}\left\{\frac{1}{2}\log\frac{\sigma_X^2\big(\tr\left(\boldsymbol{\Sigma}_\calA^{-1}\right)\sigma_{X|V}^2+1\big)}{\sigma_{X|V}^2\big(\tr\left(\boldsymbol{\Sigma}_\calA^{-1}\right)\sigma_X^2+1\big)}\right\}\\
\Delta>\max\limits_{(\calA,\calB)\in(\bbA,\bbB)}\left\{\frac{1}{2}\log\frac{\sigma_X^2\big(\tr\left(\boldsymbol{\Sigma}_\calA^{-1}\right)\sigma_{X|V}^2+1\big)}{\sigma_{X|V}^2\big(\tr\left(\boldsymbol{\Sigma}_\calA^{-1}\right)\sigma_X^2+1\big)}+\right.\\
\qquad\left.\frac{1}{2}\log\Big(\tr\left(\boldsymbol{\Sigma}_\calB^{-1}\right)\sigma_X^2+1\Big)\right\}
\end{cases}\right\}.\label{eq:Achievable_Region_SS}
\end{align}Since the arguments of the $\log$ functions and the distortion constraint in the region above are decreasing in $\tr\left(\boldsymbol{\Sigma}_\calA^{-1}\right)$ and increasing in $\tr\left(\boldsymbol{\Sigma}_\calB^{-1}\right)$, the maximization over $(\calA,\calB)$ in the above region is solved with $\calA^\star\in \argmin\limits_{\calA\in\bbA}\{\tr\left(\boldsymbol{\Sigma}_\calA^{-1}\right)\}$ and $\calB^\star\in \argmax\limits_{\calB\in\bbB}\{\tr\left(\boldsymbol{\Sigma}_\calB^{-1}\right)\}$ and \eqref{eq:Achievable_Region_SS} becomes
\begin{align}
\bigcup\limits_{\substack{V-X-\tilde{Y}_{\calL} \\ \sigma_{X|\tilde{Y}_{\calA^\star},V}^2\le D}}\left.\begin{cases}(R,\Delta):\\
R>\MI(V;X|\tilde{Y}_{\calA^\star})\\
\Delta>\MI(V;X|\tilde{Y}_{\calA^\star})+\MI(X;\tilde{Y}_{\calB^\star})
\end{cases}\right\}.\nonumber
\end{align}
Next, we rewrite the region above as
\begin{align}
 \bigcup\limits_{\substack{V-X-\tilde{Y}_{\calL} \\ \sigma_{X|\tilde{Y}_{\calA^\star},V}^2\le D}}\left.\begin{cases}(R,\Delta):\\
R>\frac{1}{2}\log\frac{\sigma_{X|\tilde{Y}_{\calA^\star}}^2}{\sigma_{X|\tilde{Y}_{\calA^\star},V}^2}\\
\Delta>\frac{1}{2}\log\frac{\sigma_{X|\tilde{Y}_{\calA^\star}}^2}{\sigma_{X|\tilde{Y}_{\calA^\star},V}^2}+\frac{1}{2}\log\frac{\sigma_X^2}{\sigma_{X|\tilde{Y}_{\calB^\star}}^2}
\end{cases}\right\}.\nonumber
\end{align}
Finally, by choosing $V$ such that $D=\sigma_{X|\tilde{Y}_{\calA^\star},V}^2$, and using \eqref{eq:SXgivenYA}, we see that the region above contains the following region
\begin{align}
  \left.\begin{cases}(R,\Delta):\\
R>\frac{1}{2}\log\frac{\sigma_X^2}{\tr\big(\boldsymbol{\Sigma}_{\calA^\star}^{-1}\big)\sigma_X^2D+D}\\
\Delta>\frac{1}{2}\log\frac{\sigma_X^2}{\tr\big(\boldsymbol{\Sigma}_{\calA^\star}^{-1}\big)\sigma_X^2D+D}+\frac{1}{2}\log\big(\tr\big(\boldsymbol{\Sigma}_{\calB^\star}^{-1}\big)\sigma_X^2+1\big)\\
\end{cases}\right\}.\nonumber
\end{align}

\subsubsection{When Unauthorized Users have a better side information}
Suppose that $\tr\big(\boldsymbol{\Sigma}_{\calA^\star}^{-1}\big) >\tr\big(\boldsymbol{\Sigma}_{\calB^\star}^{-1}\big)$. In this case, we choose the auxiliary random variable $V$ to be jointly Gaussian with $(X,\mathbf{Y}_\calL)$ and $U=V$. Hence, by Section~\ref{subsubsec:MI_Calc},
\begin{align}
\MI(X;\mathbf{Y}_\calB|V)&=\MI(X;\tilde{Y}_\calB|V)\nonumber\\
&\mathop=\limits^{(a)}\dent(\tilde{Y}_\calB|V)-\dent(\tilde{Y}_\calB|X)\nonumber\\
&=\frac{1}{2}\log\frac{h_\calB^2\sigma_{X|V}^2+\tilde{\sigma}_\calB^2}{\tilde{\sigma}_\calB^2},\label{eq:IXYB_BA}
\end{align}where $(a)$ follows since $V-X-\tilde{Y}_\calB$ forms a Markov chain. 
Using \eqref{eq:IVX}, \eqref{eq:IVXgYA_2}, \eqref{eq:IXYB_BA},  $h_\calA=\tilde{\sigma}_\calA^2=\tr\left(\boldsymbol{\Sigma}_\calA^{-1}\right)$,  and $h_\calB=\tilde{\sigma}_\calB^2=\tr\left(\boldsymbol{\Sigma}_\calB^{-1}\right)$, as showed in \eqref{eq:Suff_Stats_Inp_Outp}, we rewrite the achievable rate region in Theorem~\ref{thm:discrete_achievable} as 
\begin{align}
   &\bigcup\limits_{\substack{V-X-\tilde{Y}_\calL \\ \max\limits_{\calA\in\bbA}\sigma_{X|\tilde{Y}_\calA,V}^2\le D}}\nonumber\\
&\left.\begin{cases}(R,\Delta):\\
R>\max\limits_{\calA\in\bbA}\left\{\frac{1}{2}\log\frac{\sigma_X^2\big(\tr\left(\boldsymbol{\Sigma}_\calA^{-1}\right)\sigma_{X|V}^2+1\big)}{\sigma_{X|V}^2\big(\tr\left(\boldsymbol{\Sigma}_\calA^{-1}\right)\sigma_X^2+1\big)}\right\}\\
\Delta>\max\limits_{(\calA,\calB)\in(\bbA,\bbB)}\left\{\frac{1}{2}\log\frac{\sigma_X^2}{\sigma_{X|V}^2}+\right.\\
\quad\left.\frac{1}{2}\log\left({\tr\left(\boldsymbol{\Sigma}_\calB^{-1}\right)\sigma_{X|V}^2+1}\right)\right\}
\end{cases}\right\}.\label{eq:Achievable_Region_SS_BA}
\end{align}Since the arguments of the $\log$ functions and the distortion constraint in the region above are decreasing in $\tr\left(\boldsymbol{\Sigma}_\calA^{-1}\right)$ and increasing in $\tr\left(\boldsymbol{\Sigma}_\calB^{-1}\right)$, the maximization over $(\calA,\calB)$ in the above region is solved with $\calA^\star\in \argmin\limits_{\calA\in\bbA}\{\tr\left(\boldsymbol{\Sigma}_\calA^{-1}\right)\}$ and $\calB^\star\in \argmax\limits_{\calB\in\bbB}\{\tr\left(\boldsymbol{\Sigma}_\calB^{-1}\right)\}$ and \eqref{eq:Achievable_Region_SS_BA} becomes,
\begin{align}
\bigcup\limits_{\substack{V-X-\tilde{Y}_\calL \\ \sigma_{X|\tilde{Y}_{\calA^\star},V}^2\le D}}\left.\begin{cases}(R,\Delta):\\
R>\MI(V;X|\tilde{Y}_{\calA^\star})\\
\Delta>\MI(V;X)+\MI(X;\tilde{Y}_{\calB^\star}|V)\\
\end{cases}\right\},\nonumber
\end{align}
which we can write as
\begin{align}
 \bigcup\limits_{\substack{V-X-\tilde{Y}_\calL \\ \sigma_{X|\tilde{Y}_{\calA^\star},V}^2\le D}}\left.\begin{cases}(R,\Delta):\\
R>\frac{1}{2}\log\frac{\sigma_{X|\tilde{Y}_{\calA^\star}}^2}{\sigma_{X|\tilde{Y}_{\calA^\star},V}^2}\\
\Delta>\frac{1}{2}\log\frac{\sigma_X^2}{\sigma_{X|V}^2}+\frac{1}{2}\log\frac{h_{\calB^\star}^2\sigma_{X|V}^2+\tilde{\sigma}_{\calB^\star}^2}{\tilde{\sigma}_{\calB^\star}^2}\\
\end{cases}\right\}.\nonumber
\end{align}
Finally, by choosing $V$ such that $\sigma_{X|\tilde{Y}_{\calA^\star},V}^2=D$, which from Lemma~\ref{lemma:Conversion_XgYV_to_XgV} is equivalent to $\sigma_{X|V}^2=\frac{D}{1-\tr\big(\boldsymbol{\Sigma}_{\calA^\star}^{-1}\big)D}$, the following region is achievable,
\begin{align}
  \left.\begin{cases}(R,\Delta):\\
R>\frac{1}{2}\log\frac{\sigma_X^2}{D\left(\tr\big(\boldsymbol{\Sigma}_{\calA^\star}^{-1}\big)\sigma_X^2+1\right)}\\
\Delta>\frac{1}{2}\log\left(\frac{\sigma_X^2}{D}\left(1-\tr\big(\boldsymbol{\Sigma}_{\calA^\star}^{-1}\big)D\right)+\tr\big(\boldsymbol{\Sigma}_{\calB^\star}^{-1}\big)\sigma_X^2\right)\\
\end{cases}\right\}.\nonumber
\end{align}

\section{Proof of Lemma~\ref{lemma:Entropy_Upper_Bound}}
\label{proof:Entropy_Upper_Bound}
We have,
\begin{align}
    &\ent\big(X^n|M_p,M'_p,\mathbf{Y}_\calB^n\big)\nonumber\\
&=\ent\big(X^n,\mathbf{Y}_\calB^n|M_p,M'_p\big)-\ent\big(\mathbf{Y}_\calB^n|M_p,M'_p\big)\nonumber\\
    &\mathop=\limits^{(a)}\ent\big(X^n,\mathbf{Y}_\calB^n\big)-\ent(M_p,M'_p)-\ent\big(\mathbf{Y}_\calB^n|M_p,M'_p\big)\nonumber\\
    &\mathop=\limits^{(b)}n\ent\big(X,\mathbf{Y}_\calB\big)-\ent(M_p,M'_p)-\ent\big(\mathbf{Y}_\calB^n|M_p,M'_p\big),\label{eq:All_Terms}
\end{align}where
\begin{itemize}
    \item[$(a)$] follows since $(M_p,M'_p)$ is a deterministic function of $X^n$;
    \item[$(b)$] follows since $\big(X^n,\mathbf{Y}_\calB^n\big)$ is \ac{iid}.
\end{itemize}
We now bound the second term on the \ac{RHS} of \eqref{eq:All_Terms}. From the encoding scheme, described in Section~\ref{subsec:Discrete}, for every $j\in\intseq{1}{2^{n(R_p+R'_p)}}$, $\bar{M}_p\triangleq(M_p,M'_p)$, with rate $\bar{R}_p\triangleq R_p+R'_p$, has the following distribution,
\begin{align}
    &\bbP\big\{\bar{M}_p=j\big\}\nonumber\\
    &=\bbP\left\{\big(u^n(j),X^n\big)\in\calT_{\epsilon_1}^{(n)}(U,X)\bigcap\limits_{i=1}^{j-1}\right.\nonumber\\
    &\quad\left.\left(\big(u^n(j),X^n\big)\notin\calT_{\epsilon_1}^{(n)}(U,X)\right)\right\}\nonumber\\
    &=t_n(1-t_n)^{j-1},\nonumber
\end{align}where
$t_n\triangleq\bbP\left\{\big(U^n,X^n\big)\in\calT_{\epsilon_1}^{(n)}(U,X)\big|\right.$ $\left.U^n\in\calT_\epsilon^{(n)}(U),X^n\in\calT_\epsilon^{(n)}(X)\right\}.$ 
    Then, we have,
\begin{align}
    &\ent(\bar{M}_p)\nonumber\\
    &\mathop=\limits^{(a)}-\sum\limits_{j=1}^{2^{n\bar{R}_p}}t_n(1-t_n)^{j-1}\log\left(t_n(1-t_n)^{j-1}\right)+\lambda_n^{(1)}\nonumber\\
    &=-t_n\log t_n\sum\limits_{j=1}^{2^{n\bar{R}_p}}(1-t_n)^{j-1}-\nonumber\\
&\quad t_n\log(1-t_n)\sum\limits_{j=1}^{2^{n\bar{R}_p}}(j-1)(1-t_n)^{j-1}+\lambda_n^{(1)},\nonumber\\
    &=-\left[1-(1-t_n)^{2^{n\bar{R}_p}}\right]\log(t_n)-\nonumber\\
&\quad\left[-2^{n\bar{R}_p}t_n(1-t_n)^{2^{n\bar{R}_p}}+1-t_n-(1-t_n)^{2^{n\bar{R}_p}+1}\right]\times\nonumber\\
&\quad\frac{\log(1-t_n)}{t_n}+\lambda_n^{(1)}\nonumber\\
    &\mathop\le\limits^{(b)} -\left[1-(1-t_n)^{2^{n\bar{R}_p}}\right]\log(t_n)-\frac{\log(1-t_n)}{t_n}+\lambda_n^{(1)}\nonumber\\
    &\mathop\le\limits^{(c)} -\left[1-(1-t_n)^{2^{n\bar{R}_p}}\right]\log(t_n)+\frac{1}{1-t_n}+\lambda_n^{(1)}\nonumber\\
    &\mathop\le\limits^{(d)} -\log t_n+\frac{1}{1-t_n}+\lambda_n^{(1)}\nonumber\\
    &\mathop=\limits^{(e)}-\log t_n+\lambda_n^{(2)}+\lambda_n^{(1)}\nonumber\\
    &\mathop\le\limits^{(f)} n\big(\MI(U;X)+\epsilon_2\MI(U;X)\big)-\log\left(1-\delta_{\epsilon_1,\epsilon_2}^{(n)}\right)+\lambda_n^{(2)}+\nonumber\\
&\quad\lambda_n^{(1)},\label{eq:Typicality_Jith_4}
\end{align}where
\begin{itemize}
    \item[(a)] follows by defining $\lambda_n^{(1)}\triangleq-P_{e,2}\log P_{e,2}$ and  $P_{e,2}\triangleq\bbP\{\calE_1^c\cap\calE_2\}$ is the error probability of the event defined in \eqref{eq:error_events_2}, from Section~\ref{subsec:Discrete}, if $\bar{R}_p>\MI(U;X)+\epsilon_2\MI(U;X)$, then $P_{e,2}$, and therefore $\lambda_n^{(1)}$, vanish as $n$ grows;
    \item[(b)] follows since $\log(1-t_n)<0$;  
    \item[(c)] follows since $\log(1-x)\ge\frac{-x}{1-x}$ for $x< 1$;
    \item[(d)] follows since $\log t_n<0$; 
    \item[(e)] follows by defining, $\lambda_n^{(2)}\triangleq\frac{1}{1-t_n}$;
    \item[(f)] follows because, from \cite[Theorem~1.3]{Kramer_Book} we have,
\begin{align}
    &\left(1-\delta_{\epsilon_1,\epsilon_2}^{(n)}\right)2^{-n[\MI(U;X)+\epsilon_2\MI(U;X)]}\le t_n\nonumber\\
&\quad\le 2^{-n[\MI(U;X)-\epsilon_2\MI(U;X)]},\label{eq:Kramer_Lemma_1}
\end{align}where $\delta_{\epsilon_1,\epsilon_1}^{(n)}\triangleq2\abs{\calU}\abs{\calX}e^{-n\frac{(\epsilon_2-\epsilon_1)^2}{1+\epsilon_1}\mu_{U,X}}$ and $\mu_{U,X}\triangleq\min\limits_{(a,b)\in\text{supp}(P_{UX})}P_{UX}(a,b)$.
\end{itemize}
To bound the third term on the \ac{RHS} of \eqref{eq:All_Terms}, we define,
\begin{align}\label{eq:hatYB}
    \hat{\mathbf{Y}}_\calB^n\triangleq\left\{ \begin{array}{l}
\mathbf{Y}_\calB^n\qquad \text{if}\,\,\,\big(u^n(m_p,m'_p),\mathbf{Y}_\calB^n\big)\in\calT_{\epsilon_2}^{(n)}(U,\mathbf{Y}_\calB)\\
0\qquad\,\,\,\,\text{otherwise}
\end{array} \right.,
\end{align}and
\begin{align}
    &\ent\big(\mathbf{Y}_\calB^n|\bar{M}_p\big)\nonumber\\
    &=\sum\limits_{j=1}^{2^{n\bar{R}_p}}\bbP\{\bar{M}_p=j\}\ent\big(\mathbf{Y}_\calB^n|\bar{M}_p=j\big)\nonumber\\
    &\mathop=\limits^{(a)}\sum\limits_{j=1}^{2^{n\bar{R}_p}}\bbP\{\bar{M}_p=j\}\ent\big(\mathbf{Y}_\calB^n,\hat{\mathbf{Y}}_\calB^n|\bar{M}_p=j\big)\nonumber\\
    &=\sum\limits_{j=1}^{2^{n\bar{R}_p}}\bbP\{\bar{M}_p=j\}\left[\ent\big(\hat{\mathbf{Y}}_\calB^n|\bar{M}_p=j\big)+\ent\big(\mathbf{Y}_\calB^n|\hat{\mathbf{Y}}_\calB^n,\bar{M}_p=j\big)\right],\label{eq:Typicality_third_1}
\end{align}where $(a)$ follows since $\hat{\mathbf{Y}}_\calB^n$ is a deterministic function of $\mathbf{Y}_\calB^n$ and $\bar{M}_p$. Next we bound the first term on the \ac{RHS} of \eqref{eq:Typicality_third_1} as follows,
\begin{align}
    &\sum\limits_{j=1}^{2^{n\bar{R}_p}}\bbP\{\bar{M}_p=j\}\ent\big(\hat{\mathbf{Y}}_\calB^n|\bar{M}_p=j\big)\nonumber\\
    &\mathop\le\limits^{(a)}\sum\limits_{j=1}^{2^{n\bar{R}_p}}\bbP\{\bar{M}_p=j\}\log\big(\abs{\calT_{\epsilon_2}^{(n)}\big(\mathbf{Y}_\calB|u^n(j)\big)}+1\big)\nonumber\\
    &\mathop\le\limits^{(b)}n\sum\limits_{j=1}^{2^{n\bar{R}_p}}\bbP\{\bar{M}_p=j\}\big(\ent(\mathbf{Y}_\calB|U)+\epsilon_2\ent(\mathbf{Y}_\calB|U)\big)\nonumber\\
    &\le n\ent(\mathbf{Y}_\calB|U)+n\epsilon_2\ent(\mathbf{Y}_\calB|U),\label{eq:Typicality_third_2}
\end{align}where
\begin{itemize}
    \item[$(a)$] follows from the definition of $\hat{\mathbf{Y}}_\calB^n$ in \eqref{eq:hatYB};
    \item[$(b)$] follows from \cite[Theorem~1.2]{Kramer_Book}, which bounds the size of the conditional typical set.
\end{itemize}We now bound the second term on the \ac{RHS} of \eqref{eq:Typicality_third_1} as follows,
\begin{align}
    &\sum\limits_{j=1}^{2^{n\bar{R}_p}}\bbP\{\bar{M}_p=j\}\ent\big(\mathbf{Y}_\calB^n|\hat{\mathbf{Y}}_\calB^n,\bar{M}_p=j\big)\nonumber\\
    &\le\sum\limits_{j=1}^{2^{n\bar{R}_p}}\left(1+\bbP\left\{\mathbf{Y}_\calB^n\ne\hat{\mathbf{Y}}_\calB^n\Big|\bar{M}_p=j\right\}\log\abs{\boldsymbol{\calY}_\calB^n}\right)\bbP\{\bar{M}_p=j\}\nonumber\\
    &\le 1+\sum\limits_{j=1}^{2^{n\bar{R}_p}}\bbP\left\{\big(u^n(\bar{M}_p),\mathbf{Y}_\calB^n\big)\notin\calT_{\epsilon_2}^{(n)}(U,\mathbf{Y}_\calB)\Big|\bar{M}_p=j\right\}\times\nonumber\\
&\quad\log\abs{\boldsymbol{\calY}_\calB^n}\bbP\{\bar{M}_p=j\}\nonumber\\
    &\le 1+nP_{e,1}P_{e,2}\log\abs{\boldsymbol{\calY}_\calB},\label{eq:Typicality_third_3}
\end{align}where $P_{e,1}\triangleq\bbP\{\calE_1\}$ is the error probability of the event defined in \eqref{eq:error_events_1}. Substituting \eqref{eq:Typicality_third_2} and \eqref{eq:Typicality_third_3} in \eqref{eq:Typicality_third_1} leads to
\begin{align}
    \ent\big(\mathbf{Y}_\calB^n|\bar{M}_p\big)&\le n\ent(\mathbf{Y}_\calB|U)+\lambda_n^{(3)},\label{eq:Typicality_third_F}
\end{align}where $\lambda_n^{(3)}\triangleq n\epsilon_2\max\limits_{\calB\in\bbB}\ent(\mathbf{Y}_\calB|U)+1+nP_{e,1}P_{e,2}\log\max\limits_{\calB\in\bbB}\abs{\boldsymbol{\calY}_\calB}$. 
Finally, substituting \eqref{eq:Typicality_Jith_4} and \eqref{eq:Typicality_third_F} in \eqref{eq:All_Terms} results to,
\begin{align}
    &\frac{1}{n}\ent\big(X^n|M_p,M'_p,\mathbf{Y}_\calB^n\big)\nonumber\\
   & \ge \ent\big(X,\mathbf{Y}_\calB\big) -\big(\MI(U;X)+\epsilon_2\MI(U;X)\big)-\nonumber\\
&\quad\frac{\lambda_n^{(1)}}{n}-\frac{\lambda_n^{(2)}}{n}-\ent(\mathbf{Y}_\calB|U)-\frac{\lambda_n^{(3)}}{n}\nonumber\\
    &\mathop=\limits^{(a)} \ent\big(X|\mathbf{Y}_\calB,U\big) -\epsilon_2\MI(U;X)-\frac{\lambda_n^{(1)}}{n}-\frac{\lambda_n^{(2)}}{n}-\frac{\lambda_n^{(3)}}{n}\nonumber\\
    &\mathop=\limits^{(b)} \ent\big(X|\mathbf{Y}_\calB,U\big) -\eta_n,\nonumber
\end{align}where
\begin{itemize}
    \item[$(a)$] follows since $U-X-\mathbf{Y}_\calB$ forms a Markov chain;
    \item[$(b)$] holds by defining,
    \begin{align}
        \eta_n&\triangleq\epsilon_2\MI(U;X)-\log\left(1-\delta_{\epsilon_1,\epsilon_2}^{(n)}\right)+\frac{-P_{e,2}\log P_{e,2}}{n}+\nonumber\\
&\quad\frac{1}{n(1-t_n)}+\epsilon_2\max\limits_{\calB\in\bbB}\ent(\mathbf{Y}_\calB|U)+\nonumber\\
&\quad\frac{1}{n}+P_{e,1}P_{e,2}\log\max\limits_{\calB\in\bbB}\abs{\boldsymbol{\calY}_\calB}.\nonumber
    \end{align}
\end{itemize}

\end{appendices}

\bibliographystyle{IEEEtran}
\bibliography{IEEEabrv,bibfile}

\begin{thebibliography}{10}
\providecommand{\url}[1]{#1}
\csname url@samestyle\endcsname
\providecommand{\newblock}{\relax}
\providecommand{\bibinfo}[2]{#2}
\providecommand{\BIBentrySTDinterwordspacing}{\spaceskip=0pt\relax}
\providecommand{\BIBentryALTinterwordstretchfactor}{4}
\providecommand{\BIBentryALTinterwordspacing}{\spaceskip=\fontdimen2\font plus
\BIBentryALTinterwordstretchfactor\fontdimen3\font minus \fontdimen4\font\relax}
\providecommand{\BIBforeignlanguage}[2]{{%
\expandafter\ifx\csname l@#1\endcsname\relax
\typeout{** WARNING: IEEEtran.bst: No hyphenation pattern has been}%
\typeout{** loaded for the language `#1'. Using the pattern for}%
\typeout{** the default language instead.}%
\else
\language=\csname l@#1\endcsname
\fi
#2}}
\providecommand{\BIBdecl}{\relax}
\BIBdecl

\bibitem{Secure_SC22}
H.~ZivariFard and R.~A. Chou, ``Secure data storage resilient against compromised users via an access structure,'' in \emph{Proc. {IEEE} Info. Theory Workshop (ITW)}, Mumbai, India, Nov. 2022, pp. 404--469.

\bibitem{Secret_Sharing_1}
G.~R. Blakley, ``Safeguarding cryptographic keys,'' in \emph{Proc. {AFIPS} 79th {N}at. {C}omput. {C}onf.}, New York, NY, USA, Jun. 1979, pp. 313--317.

\bibitem{Secret_Sharing_2}
A.~Shamir, ``How to share a secret,'' \emph{Commun. {ACM}}, vol.~22, no.~11, pp. 612--613, Nov. 1979.

\bibitem{Berger71}
T.~Berger, \emph{Rate Distortion Theory: {A} Mathematical Basis for Data Compression}.\hskip 1em plus 0.5em minus 0.4em\relax \hspace{-0.25cm}Englewood Cliffs, NJ: Prentice-Hall, 1971.

\bibitem{beimel2011secret}
A.~Beimel, ``Secret-sharing schemes: {A} survey,'' in \emph{International conference on coding and cryptology}.\hskip 1em plus 0.5em minus 0.4em\relax \hspace{-0.2cm}Springer, 2011, pp. 11--46.

\bibitem{Prabhakaran07}
V.~Prabhakaran and K.~Ramchandran, ``On secure distributed source coding,'' in \emph{Proc. {IEEE} Info. Theory Workshop (ITW)}, CA, USA, May 2007, pp. 442--447.

\bibitem{GunduzITW08}
D.~G\"{u}nd\"{u}z, E.~Erkip, and H.~V. Poor, ``Secure lossless compression with side information,'' in \emph{Proc. {IEEE} Info. Theory Workshop (ITW)}, Porto, Portugal, May 2008, pp. 169--173.

\bibitem{GunduzISIT08}
------, ``Lossless compression with security constraints,'' in \emph{Proc. {IEEE} Int. Symp. on Info. Theory (ISIT)}, Toronto, Canada, Jul. 2008, pp. 111--115.

\bibitem{VillardPianta}
J.~Villard and P.~Piantanida, ``Secure multiterminal source coding with side information at the eavesdropper,'' \emph{{IEEE} Trans. Inf. Theory}, vol.~59, no.~6, pp. 3668--3692, Jun. 2013.

\bibitem{Tandon13}
R.~Tandon, S.~Ulukus, and K.~Ramchandran, ``Secure source coding with a helper,'' \emph{{IEEE} Trans. Inf. Theory}, vol.~59, no.~4, pp. 2178--2187, Apr. 2013.

\bibitem{Chia13}
Y.-K. Chia and K.~Kittichokechai, ``On secure source coding with side information at the encoder,'' in \emph{Proc. {IEEE} Int. Symp. on Info. Theory (ISIT)}, Istanbul, Turkey, Jul. 2013, pp. 2204--2208.

\bibitem{Kittichokechai16}
K.~Kittichokechai, Y.-K. Chia, T.~J. Oechtering, M.~Skoglund, and T.~Weissman, ``Secure source coding with a public helper,'' \emph{{IEEE} Trans. Inf. Theory}, vol.~62, no.~7, pp. 3930--3949, Jul. 2016.

\bibitem{EkremUlukus13_Lossy}
E.~Ekrem and S.~Ulukus, ``Secure lossy transmission of vector {G}aussian sources,'' \emph{{IEEE} Trans. Inf. Theory}, vol.~59, no.~9, pp. 5466--5487, Sep. 2013.

\bibitem{Cover_Book}
T.~M. Cover and J.~A. Thomas, \emph{Elements of Information Theory}, 2nd~ed.\hskip 1em plus 0.5em minus 0.4em\relax \hspace{-0.2cm}John Wiley \& Sons, Inc., 2001.

\bibitem{liang2009compound}
Y.~Liang, G.~Kramer, and H.~V. Poor, ``Compound wiretap channels,'' \emph{EURASIP J. on Wireless Commun. and Netw.}, vol. 2009, pp. 1--12, 2009.

\bibitem{AllertonCuffSong}
E.~C. Song, P.~Cuff, and H.~V. Poor, ``A rate-distortion based secrecy system with side information at the decoders,'' in \emph{Proc. 52th Annual Allerton Conference on Communication, Control, and Computing}, Monticello, IL, Sep. 2014, pp. 755--762.

\bibitem{SchielerCuff}
C.~Schieler and P.~Cuff, ``Rate-distortion theory for secrecy systems,'' \emph{{IEEE} Trans. Inf. Theory}, vol.~60, no.~12, pp. 7584--7605, Oct. 2014.

\bibitem{Yamamoto83}
H.~Yamamoto, ``A source coding problem for sources with additional outputs to keep secret from the receiver or wiretappers,'' \emph{{IEEE} Trans. Inf. Theory}, vol.~29, no.~6, pp. 918--923, Nov. 1983.

\bibitem{Yamamoto88}
------, ``A rate-distortion problem for a communication system with a secondary decoder to be hindered,'' \emph{{IEEE} Trans. Inf. Theory}, vol.~34, no.~4, pp. 835--842, Jul. 1988.

\bibitem{Yamamoto94}
------, ``Coding theorems for {S}hannon’s cipher system with correlated source outputs, and common information,'' \emph{{IEEE} Trans. Inf. Theory}, vol.~40, no.~1, pp. 85--95, Jan. 1994.

\bibitem{Yamamoto97}
------, ``Rate-distortion theory for the {S}hannon cipher system,'' \emph{{IEEE} Trans. Inf. Theory}, vol.~43, no.~3, pp. 827--835, May 1997.

\bibitem{Merhav06}
N.~Merhav, ``On the {S}hannon cipher system with a capacity-limited key-distribution channel,'' \emph{{IEEE} Trans. Inf. Theory}, vol.~52, no.~3, pp. 1269--1273, Mar. 2006.

\bibitem{Merhav08}
------, ``Shannon’s secrecy system with informed receivers and its application to systematic coding for wiretapped channels,'' \emph{{IEEE} Trans. Inf. Theory}, vol.~54, no.~6, pp. 2723--2734, Jun. 2006.

\bibitem{SatpathyCuff15}
P.~Cuff and S.~Satpathy, ``{G}aussian secure source coding and wyner's common information,'' in \emph{Proc. {IEEE} Int. Symp. on Info. Theory (ISIT)}, Hong Kong, China, Jun. 2015, pp. 116--120.

\bibitem{VidhiRemi22}
R.~Vidhi, R.~A. Chou, and H.~M. Kwon, ``Information-theoretic secret sharing from correlated {G}aussian random variables and public communication,'' \emph{{IEEE} Trans. Inf. Theory}, vol.~68, no.~1, pp. 549--559, Jan. 2022.

\bibitem{ElGamalKim}
A.~El~Gamal and Y.-H. Kim, \emph{Network Information Theory}, 1st~ed.\hskip 1em plus 0.5em minus 0.4em\relax \hspace{-0.2cm}Cambridge, U.K: Cambridge University Press, 2012.

\bibitem{Monotone_Property}
J.~Benaloh and J.~Leichter, ``Generalized secret sharing and monotone functions,'' in \emph{Proc. Conf. Theory Appl. Cryptogr.}, New York, NY, USA, Feb. 1988, pp. 27--35.

\bibitem{Gallager_Stochastic_Book}
R.~G. Gallager, \emph{Stochastic Processes: Theory for Applications}.\hskip 1em plus 0.5em minus 0.4em\relax \hspace{-0.2cm}Cambridge, U.K: Cambridge University Press, 2013.

\bibitem{SIMO_Fading_WTC}
P.~Parada and R.~Blahut, ``Secrecy capacity of {SIMO} and slow fading channels,'' in \emph{Proc. {IEEE} Int. Symp. on Info. Theory (ISIT)}, Adelaide, SA, Australia, Sep. 2005, pp. 2152--2155.

\bibitem{Thomas87}
J.~A. Thomas, ``Feedback can at most double {G}aussian multiple access channel capacity,'' \emph{{IEEE} Trans. Inf. Theory}, vol.~33, no.~5, pp. 711--716, Sep. 1987.

\bibitem{Liu07}
T.~Liu and P.~Viswanath, ``An extremal inequality motivated by multiterminal information theoretic problems,'' \emph{{IEEE} Trans. Inf. Theory}, vol.~53, no.~5, pp. 1839--1851, Jul. 2007.

\bibitem{WangChen13}
J.~Wang and J.~Chen, ``Vector {G}aussian two-terminal source coding,'' \emph{{IEEE} Trans. Inf. Theory}, vol.~59, no.~6, pp. 3693--3708, Jun. 2013.

\bibitem{EkremUlukus13_MIMO}
E.~Ekrem and S.~Ulukus, ``The secrecy capacity region of the {G}aussian {MIMO} multi-receiver wiretap channel,'' \emph{{IEEE} Trans. Inf. Theory}, vol.~57, no.~4, pp. 2083--2114, Apr. 2011.

\bibitem{MI_MMSE}
D.~Guo, S.~Shamai, and S.~Verd\'u, ``Mutual information and minimum mean-square error in {G}aussian channels,'' \emph{{IEEE} Trans. Inf. Theory}, vol.~51, no.~4, pp. 1261--1282, Apr. 2005.

\bibitem{Kramer_Book}
G.~Kramer, ``Topics in multi-user information theory,'' \emph{Found. Trends Comm. Inf. Theory}, vol.~4, no. 4-5, pp. 265--444, 2008.

\bibitem{Pinsker64}
M.~S. Pinsker, \emph{Information and Information Stability of Random Variables and Processes}.\hskip 1em plus 0.5em minus 0.4em\relax \hspace{-0.2cm}Holden-Day Inc., 1964.

\bibitem{Fano61}
R.~Fano, \emph{Transmission of Information: A Statistical Theory of Communications}.\hskip 1em plus 0.5em minus 0.4em\relax \hspace{-0.2cm}Cambridge, MA, USA: MIT Press, 1961.

\end{thebibliography}

\end{document}